\documentclass[12pt,letterpaper]{article}

\usepackage{booktabs}

\usepackage{latexsym}
\usepackage{amssymb,amsmath, bm,pgfplots,tikz,bbm}
\usepackage{graphicx}
\usepackage{algpseudocode}
\usepackage{marvosym}
\usepackage{multirow,float}
\usepackage{caption}
\usepackage{subcaption}
\usepackage{comment}

\usepackage{natbib}
\usepackage{color}
\usepackage[bookmarksopen=true, bookmarksnumbered=true,
pdfstartview=FitH, breaklinks=true, urlbordercolor={0 1 0}, citebordercolor={0 0 1}]{hyperref}

\usepackage{dcolumn}
\newcolumntype{.}{D{.}{.}{-1}}
\newcolumntype{d}[1]{D{.}{.}{#1}}

\usepackage{theorem}
\theoremstyle{plain}
\theoremheaderfont{\scshape}
\newtheorem{theorem}{Theorem}

\newtheorem{assumption}{Assumption}

\newtheorem{lemma}{Lemma}

\newcommand{\qed}{\hfill \ensuremath{\Box}}

\newcommand{\mc}[1]{\mathbb{#1}}
\renewcommand{\d}{\mathrm{d}}

\DeclareMathOperator\Cov{Cov}

\newenvironment{proof}{\vspace{1ex}\noindent{\bf Proof}\hspace{0.5em}}
{\hfill\qed\vspace{1ex}}
\usetikzlibrary{decorations.markings}
\usetikzlibrary{decorations.pathmorphing}
\usetikzlibrary{shapes.geometric, arrows}
\usetikzlibrary{arrows,decorations.pathmorphing,backgrounds,positioning,fit,matrix}
\usetikzlibrary{shapes,decorations,arrows,calc,arrows.meta,fit,positioning}
\tikzset{auto,node distance =1 cm and 1 cm,semithick,
	state/.style ={circle, draw, minimum width = 0.7 cm},
	point/.style = {circle, draw, inner sep=0.04cm,fill,node contents={}},
	bidirected/.style={Latex-Latex,dashed},
	el/.style = {inner sep=2pt, align=left, sloped}
}
\usetikzlibrary{positioning}
\usetikzlibrary{fadings}
\usetikzlibrary{intersections}
\usepackage{kantlipsum}
\allowdisplaybreaks

\usepackage{rotating}

\usepackage{arydshln}
\usepackage{threeparttable}

\usepackage[compact]{titlesec}

\newcommand{\blind}{1}

\usepackage{times}

\usepackage{algorithm}
\usepackage{algpseudocode}
\usepackage{pifont}
\usepackage{xcolor}

\begin{document}
	
	\newcommand\ud{\mathrm{d}}
	\newcommand\dist{\buildrel\rm d\over\sim}
	\newcommand\ind{\stackrel{\rm indep.}{\sim}}
	\newcommand\iid{\stackrel{\rm i.i.d.}{\sim}}
	\newcommand\logit{{\rm logit}}
	\renewcommand\r{\right}
	\renewcommand\l{\left}
	\newcommand\cO{\mathcal{O}}
	\newcommand\cY{\mathcal{Y}}
	\newcommand\cZ{\mathcal{Z}}
	\newcommand\E{\mathbb{E}}
	\newcommand\cL{\mathcal{L}}
	\newcommand\V{\mathbb{V}}
	\newcommand\cA{\mathcal{A}}
	\newcommand\cB{\mathcal{B}}
	\newcommand\cD{\mathcal{D}}
	\newcommand\cE{\mathcal{E}}
	\newcommand\cM{\mathcal{M}}
	\newcommand\cU{\mathcal{U}}
	\newcommand\cN{\mathcal{N}}
	\newcommand\cT{\mathcal{T}}
	\newcommand\cX{\mathcal{X}}
	\newcommand\bA{\mathbf{A}}
	\newcommand\bH{\mathbf{H}}
	\newcommand\bB{\mathbf{B}}
	\newcommand\bP{\mathbf{P}}
	\newcommand\bQ{\mathbf{Q}}
	\newcommand\bU{\mathbf{U}}
	\newcommand\bD{\mathbf{D}}
	\newcommand\bS{\mathbf{S}}
	\newcommand\bx{\mathbf{x}}
	\newcommand\bX{\mathbf{X}}
	\newcommand\bV{\mathbf{V}}
	\newcommand\bW{\mathbf{W}}
	\newcommand\bM{\mathbf{M}}
	\newcommand\bZ{\mathbf{Z}}
	\newcommand\bY{\mathbf{Y}}
	\newcommand\bt{\mathbf{t}}
	\newcommand\bbeta{\bm{\beta}}
	\newcommand\bpi{\bm{\pi}}
	\newcommand\bdelta{\bm{\delta}}
	\newcommand\bgamma{\bm{\gamma}}
	\newcommand\balpha{\bm{\alpha}}
	\newcommand\bone{\mathbf{1}}
	\newcommand\bzero{\mathbf{0}}
	\newcommand\tomega{\tilde\omega}
	\newcommand{\argmax}{\operatornamewithlimits{argmax}}
	
	\newcommand{\R}{\textsf{R}}
	
	\newcommand\spacingset[1]{\renewcommand{\baselinestretch}%
		{#1}\small\normalsize}
	
	\spacingset{1}
	
	\newcommand{\tit}{\bf Experimental Evaluation of Individualized Treatment
		Rules}
	
	
	\if0\blind
	\title{\tit}
	\fi
	
	\if1\blind
	
	\title{\tit\thanks{We thank Naoki Egami, Colin Fogarty, Zhichao
				Jiang, Susan Murphy, Nicole Pashley, Stefan Wager, three
				anonymous reviewers, and especially the Associate Editor for
				many helpful comments. We thank Dr. Hsin-Hsiao Wang who inspired
				us to write this paper.  Imai thanks the Alfred P. Sloan
				Foundation for partial support of this research.  The proposed
				methodology is implemented through an open-source \R\ package,
				{\sf evalITR}, which is freely available for download at the
				Comprehensive R Archive Network (CRAN;
				\url{https://CRAN.R-project.org/package=evalITR}).  } }
		
		\author{Kosuke Imai\thanks{Professor, Department of Government and
				Department of Statistics, Harvard University, Cambridge, MA
				02138. Phone: 617--384--6778, Email:
				\href{mailto:Imai@Harvard.Edu}{Imai@Harvard.Edu}, URL:
				\href{https://imai.fas.harvard.edu}{https://imai.fas.harvard.edu}}
			\hspace{.75in} Michael Lingzhi Li\thanks{Operation Research
				Center, Massachusetts Institute of Technology, Cambridge, MA
				02139. \href{mailto:mlli@mit.edu}{mlli@mit.edu}}}
		
		\date{First version: May 15, 2019\\ This version: \today}
		
		\fi
		\maketitle
		
		\pdfbookmark[1]{Title Page}{Title Page}
		
		\thispagestyle{empty}
		\setcounter{page}{0}
		
		\begin{abstract}
			
			The increasing availability of individual-level data has led to
			numerous applications of individualized (or personalized) treatment
			rules (ITRs).  Policy makers often wish to empirically evaluate ITRs
			and compare their relative performance before implementing them in a
			target population.  We propose a new evaluation metric, the
			population average prescriptive effect (PAPE).  The PAPE compares
			the performance of ITR with that of non-individualized treatment
			rule, which randomly treats the same proportion of units.  Averaging
			the PAPE over a range of budget constraints yields our second
			evaluation metric, the area under the prescriptive effect curve
			(AUPEC).  The AUPEC represents an overall performance measure for
			evaluation, like the area under the receiver and operating
			characteristic curve (AUROC) does for classification, and is a
			generalization of the QINI coefficient utilized in uplift modeling.
			We use Neyman's repeated sampling framework to estimate the PAPE and
			AUPEC and derive their exact finite-sample variances based on random
			sampling of units and random assignment of treatment.  We extend our
			methodology to a common setting, in which the same experimental data
			is used to both estimate and evaluate ITRs.  In this case, our
			variance calculation incorporates the additional uncertainty due to
			random splits of data used for cross-validation.  The proposed
			evaluation metrics can be estimated without requiring modeling
			assumptions, asymptotic approximation, or resampling methods. As a
			result, it is applicable to any ITR including those based on complex
			machine learning algorithms.  The open-source software package is
			available for implementing the proposed methodology. 
			
			
		\end{abstract}
		\clearpage
		\spacingset{1.83}
		
		\section{Introduction}
		\label{sec:intro}
		
		In today's data-rich society, the individualized (or personalized)
		treatment rules (ITRs), which assign different treatments to
		individuals based on their observed characteristics, play an essential
		role.  Examples include personalized medicine and micro-targeting in
		business and political campaigns
		\citep[e.g.,][]{hamb:coll:10,imai:stra:11}.  In the causal inference
		literature, a number of researchers have developed methods to estimate
		optimal ITRs using a variety of machine learning algorithms \citep[see
		e.g.,][]{qian:murp:11,zhan:etal:12a,fu2016estimating,
			lued:vand:16a,lued:vand:16, zhou:etal:17,athe:wage:18,kita:tete:18}.
		In addition, applied researchers often use machine learning algorithms
		to estimate heterogeneous treatment effects and then construct ITRs
		based on the resulting estimates.
		
		In this paper, we consider a common setting, in which a policy-maker
		wishes to experimentally evaluate the empirical performance of an ITR
		before implementing it in a target population. Such evaluation is also
		essential for comparing the efficacy of alternative ITRs.
		Specifically, we show how to use a randomized experiment for
		evaluating ITRs. We propose two new evaluation metrics.  The first is
		the population average prescriptive effect (PAPE), which compares an
		ITR with a non-individualized treatment rule that randomly assigns the
		same proportion of units to the treatment condition. The PAPE
		represents the difference between the average outcome under the ITR
		and that under the random treatment rule.  The key idea is that a
		well-performing ITR should outperform the random treatment rule, which
		does not utilize any individual-level information.
		
		Averaging the PAPE over a range of budget constraints yields our
		second evaluation metric, the area under the prescriptive effect curve
		(AUPEC).  Like the area under the receiver and operating
		characteristic curve (AUROC) for classification, the AUPEC represents
		an overall summary measure of how well an ITR performs over the random
		treatment rule that treats the same proportion of units.
		
		We estimate these evaluation metrics using \citet{neym:23}'s repeated
		sampling framework \citep[see][Chapter~6]{imbe:rubi:15}.  An advantage
		of this approach is that it does not require any modeling assumption
		or asymptotic approximation.  As a result, we can evaluate a broad
		class of ITRs including those based on complex machine learning
		algorithms.  We show how to estimate the PAPE and AUPEC with a minimal
		amount of finite sample bias and derive the exact variance solely
		based on random sampling of units and random assignment of treatment.
		
		We further extend this methodology to a common evaluation setting, in
		which the same experimental data is used to both estimate and evaluate
		ITRs.  In this case, our finite-sample variance calculation is exact
		and directly incorporates the additional uncertainty due to random
		splits of data used for cross-validation.  We implement the proposed
		methodology through an open-source \R\ package\if0\blind, {\sf
			evalITR}\fi.
		
		Our simulation study demonstrates the accurate coverage of the
		proposed confidence intervals in small samples
		(Section~\ref{sec:synthetic}).  We also apply our methods to the
		Project STAR (Student-Teacher Achievement Ratio) experiment and
		compare the empirical performance of ITRs based on several popular
		methods (Section~\ref{sec:empirical}).  Our evaluation approach
		addresses theoretical and practical difficulties of conducting
		reliable statistical inference for ITRs.
		
		\paragraph{Relevant Literature.}
		A large number of existing studies have focused on the derivation of
		optimal ITRs that maximize the population average value.  For example,
		\citet{qian:murp:11} use penalized least squares whereas
		\citet{zhao:etal:12} show how a support vector machine can be used to
		derive an optimal ITR.  Another popular approach is based on
		doubly-robust estimation \citep[e.g.,][]{dudi:etal:11, zhan:etal:12a,
			chak:labe:zhao:14, jian:li:16, athe:wage:18, kall:18}.
		
		We propose a general methodology for empirically evaluating and
		comparing the performance of various ITRs including the ones proposed
		by these and other authors.  While many of these methods come with
		uncertainty measures, even those that produce standard errors rely on
		asymptotic approximation, modeling assumptions, or resampling
		methods. In contrast, our methodology utilizes Neyman's repeated
		sampling framework and does not require any of these assumptions or
		approximations.
		
		There also exists a related literature on policy evaluation.  Starting
		with \citet{mans:04}, many studies focus on the derivation of regret
		bounds given a class of ITRs.  For example, \citet{kita:tete:18} show
		that an ITR, which maximizes the empirical average value, is minimax
		optimal without a strong restriction on the class of ITRs, whereas
		\citet{athe:wage:18} establish a regret bound for an ITR based on
		doubly-robust estimation in observational studies \citep[see
		also][]{zhou:athe:wage:18}. In addition,
		\citet{lued:vand:16a,lued:vand:16} propose consistent estimators of
		the optimal average value even when an optimal ITR is not unique
		\citep[see also][]{rai:18}.
		
		Our goal is different from these studies.  We focus on statistical
		inference using the Neyman's repeated sampling framework for the
		experimental evaluation of arbitrary ITRs including optimal or
		non-optimal and simple or complex ones.  Our evaluation metric is also
		different from the existing metrics.  In particular, to the best of
		our knowledge, we are the first to formally study the AUPEC as an
		AUROC-like summary measure for evaluation.
		
		In contrast, much of the policy evaluation literature focus on the
		optimal average value, which is required to compute the regret of an
		ITR.  \citet{athe:wage:18} briefly discusses a quantity related to the
		PAPE in their empirical application, but this quantity evaluates an
		ITR against the treatment rule that randomly treats exactly one half
		of units rather than the same proportion as the one treated under the
		ITR.  Although empirical studies in the campaign and marketing
		literatures have used ``uplift modeling,'' which is based on the PAPE
		\citep[e.g.,][]{imai:stra:11,rzep:jaro:12,guti:gera:16,asca:18,fifi:18},
		none develops formal estimation and inferential methods.  We show that
		the AUPEC is a generalization of the QINI coefficient, which is a
		widely utilized statistic in uplift modeling
		\citep{radcliffe2007using,diemert2018large}.  Thus, our theoretical
		results for the AUPEC apply directly to the QINI coefficient as well.
		
		Another related literature is concerned with the estimation of
		heterogeneous effects.  Researchers have explored the use of
		tree-based methods
		\citep[e.g.,][]{imai:stra:11,athe:imbe:16,wage:athe:18,hahn2020bayesian},
		regularized regressions \citep[e.g.,][]{imai:ratk:13, kunz:etal:18},
		and ensamble methods \citep[e.g.,][]{vand:rose:11, grim:mess:west:17}.
		In practice, the estimated heterogeneous treatment effects based on
		these machine learning algorithms are used to construct ITRs.
		
		However, as \cite{cher:etal:19} point out, most machine learning
		algorithms, which require data-driven tuning parameters, cannot be
		regarded as consistent estimators of the conditional average treatment
		effect (CATE) unless strong assumptions are imposed.  They propose a
		methodology to estimate heterogeneous treatment effects without such
		assumptions.  Similar to theirs, our methodology does not depend on
		any modeling assumption and accounts for the uncertainty due to
		splitting of data.  The key difference is that we focus on the
		evaluation of ITRs.  In addition, our variance calculation is based on
		randomization and does not rely on asymptotic approximation.
		
		Finally, \cite{andr:kita:mccl:20} develops a conditional inference
		procedure, based on normal approximation, for the average value of the
		best-performing policy based on experimental or observational data.
		In contrast, we develop an unconditional exact inference for the
		difference in the average value between any pair of policies under a
		budget constraint.  We also consider the evaluation of estimated
		policies based on the same experimental data using cross-validation
		whereas \citeauthor{andr:kita:mccl:20} focus on the evaluation of
		fixed policies.
		
		\section{Evaluation Metrics}
		\label{sec:evalmetrics}
		
		In this section, we introduce our evaluation metrics.  We first
		propose the population average prescriptive effect (PAPE) that, unlike
		the population average value, adjusts for the proportion of units
		treated by an ITR.  The idea is that an efficacious ITR should
		outperform a non-individualized treatment rule, which randomly assigns
		the same proportion of units to the treatment condition.  We extend
		the PAPE to the settings with a binding budget constraint.  Finally,
		we propose the area under the prescriptive effect curve (AUPEC) as a
		univariate summary performance measure of an ITR under a range of
		budget constraint.
		
		\subsection{The Setup}
		
		Following the literature, we define an ITR as a deterministic map from
		the covariate space $\cX$ to the binary treatment assignment
		\citep[e.g.,][]{qian:murp:11,zhao:etal:12},
		\begin{equation*}
			f: \cX \longrightarrow \{0, 1\}.
		\end{equation*}
		Let $T_i$ denote the treatment assignment indicator variable, which is
		equal to $1$ if unit $i$ is assigned to the treatment condition, i.e.,
		$T_i \in \cT = \{0,1\}$. For each unit, we observe the outcome
		variable $Y_i \in \cY$ as well as the vector of pre-treatment
		covariates, $\bX_i \in \cX$, where $\cY$ is the support of the outcome
		variable.  We assume no interference between units and denote the
		potential outcome for unit $i$ under the treatment condition $T_i = t$
		as $Y_i(t)$ for $t=0,1$.  Then, the observed outcome is given by
		$Y_i = Y_i(T_i)$.
		\begin{assumption}[No Interference between Units] \label{asm:SUTVA}
			\spacingset{1} The potential outcomes for unit $i$ do not depend on
			the treatment status of other units.  That is, for all
			$t_1, t_2,\ldots,t_n \in \{0, 1\}$, we have,
			$Y_i(T_1 = t_1, T_2 = t_2, \ldots, T_n = t_n) = Y_i(T_i = t_i).$
		\end{assumption}
		Under this assumption, the existing literature almost exclusively
		focuses on the derivation of an optimal ITR that maximizes the
		following population average value
		\citep[e.g.,][]{qian:murp:11,zhao:etal:12,zhou:etal:17},
		\begin{equation}
			\lambda_f \ = \ \E\{Y_i(f(\bX_i))\}. \label{eq:PAV}
		\end{equation}
		Next, we show that $\lambda_f$ may not be the best evaluation metric
		in some cases.
		
		\subsection{The Population Average Prescriptive Effect}
		\label{subsec:PAPE}
		
		We now introduce our main evaluation metric, the Population Average
		Prescriptive Effect (PAPE).  The PAPE is based on two ideas.  First,
		it is reasonable to expect a good ITR to outperform a {\it
			non-individualized} treatment rule, which does not use any
		information about individual units when deciding who should receive
		the treatment.  Second, a budget constraint should be considered since
		the treatment is often costly.  This means that a good ITR should
		identify units who benefit from the treatment most.  These two
		considerations lead to the random treatment rule, which
		assigns, with equal probability, the same proportion of units to the
		treatment condition, as a natural baseline for comparison.
		
		\begin{figure}[t!]
			\centering \spacingset{1}
			\begin{tikzpicture}[scale=3.5]
				\draw[->] (0,0) -- (2.5,0) coordinate (x axis)  node[right, black]{Proportion treated, $p$};
				\draw[->] (0,0) -- (0,2.5) coordinate (y axis) node[above, black]{Average outcome};
				\node[circle,fill=black,inner sep=0pt,minimum size=3pt] (a) at (0,0.4) {};
				\node[left=0.1cm of a] {$\E[Y_i(0)]$};
				\node[circle,fill=black,inner sep=0pt,minimum size=3pt] (b) at (2.2,2.4) {};
				\node[right=0.1cm of b] {$\E[Y_i(1)]$};
				\node[circle,fill=black,inner sep=0pt,minimum size=3pt] (bx) at (2.2, 0) {};
				\node[below=0.1cm of bx] {$1$};
				\draw[dashed,-] (b) -- (bx);
				\draw[black,thick,-] (a) --  (b);
				\node[circle,fill=blue,inner sep=0pt,minimum size=3pt] (g) at (0.44, 1) {};
				\node[] (gy) at (0, 1) {};
				\node[above=0.1cm of g] {\color{blue} $\E[Y_i(f)]$};
				\node[circle,fill=blue,inner sep=0pt,minimum size=3pt] (gx) at (0.44, 0) {};
				\node[below=0.1cm of gx] {\color{blue} $0.2$};
				\draw[blue, dashed,-] (g) -- (gx);
				\node[circle,fill=red,inner sep=0pt,minimum size=3pt] (f) at (1.76, 1.3) {};
				\node[] (fy) at (0, 1.4) {};
				\node[above=0.1cm of f] {\color{red} $\E[Y_i(g)]$};
				\node[circle,fill=red,inner sep=0pt,minimum size=3pt] (fx) at (1.76, 0) {};
				\node[below=0.1cm of fx] {\color{red} $0.8$};
				\draw[red, dashed,-]  (f) -- (fx);
				\node[] at (1.4, 2.3) {Random assignment};
				\node[] at (1.6, 2.15) {rule};
				\node[] (g2) at (0.44, 0.8) {};
				\draw[decoration={brace,mirror,raise=3pt},decorate] (gx) -- node[right=4pt, align=left] {Population \\ Average Value} (g);
				\draw[decoration={brace,raise=3pt},decorate] (g2) -- node[left=3pt]
				{\bf PAPE} (g);
			\end{tikzpicture}
			\caption{The Importance of Accounting for the Proportion of
				Treated Units. In this illustrative example, an ITR $g$
				(red) outperforms another ITR $f$ (blue) in terms of the
				population average value, i.e., $\E[Y_i(f)] <\E[Y_i(g)]$.
				However, unlike $f$, the ITR $g$ is doing worse than the
				random treatment rule (black).  In contrast, the population
				average prescriptive effect (PAPE) measures the performance
				of an ITR as the difference in the average value between the
				ITR and random treatment
				rule.} \label{fig:PAVcounterexample}
		\end{figure}
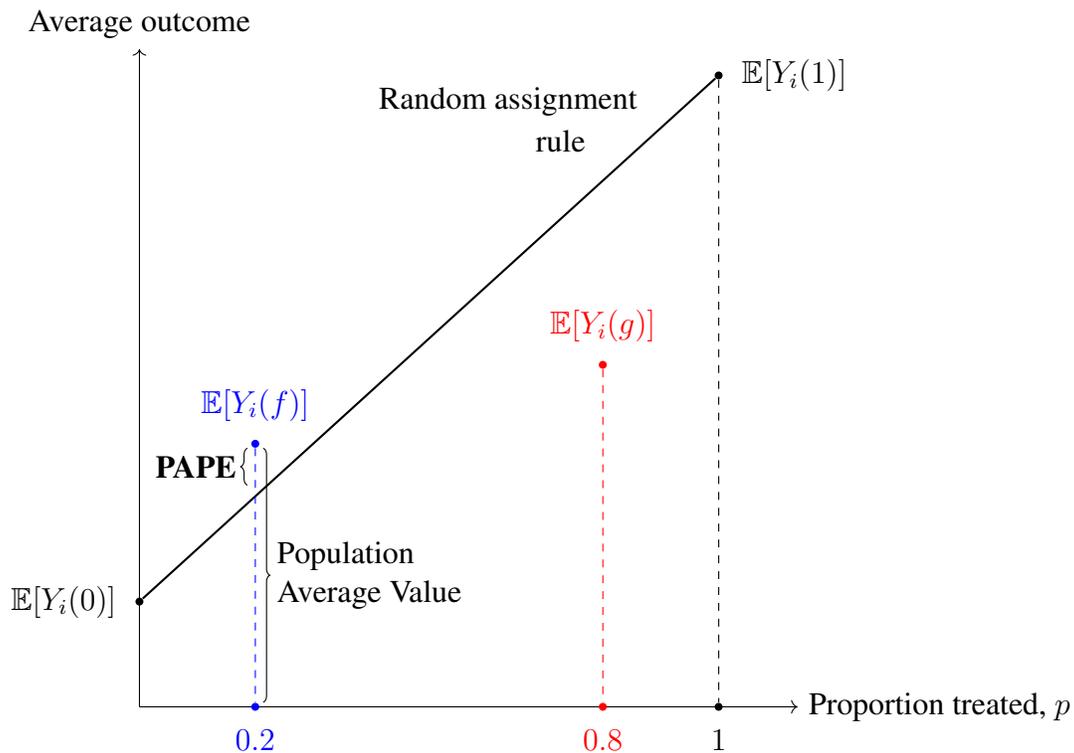
		
		Figure~\ref{fig:PAVcounterexample} illustrates the importance of
		accounting for the proportion of units treated by an ITR.  In this
		figure, the horizontal axis represents the proportion treated and the
		vertical axis represents the average outcome under an ITR.  The
		example shows that an ITR $g$ (red), which treats 80\% of units, has a
		greater average value than another ITR $f$ (blue), which treats 20\%
		of units, i.e., $\E[Y_i(f)]<\E[Y_i(g)]$.  Despite this fact, $g$ is
		outperformed by the random treatment rule (black), which treats the
		same proportion of units, whereas $f$ does a better job than the
		random treatment rule.  This is indicated by the fact that the black
		solid line is placed above $\E[Y_i(g)]$ and below $\E[Y_i(f)]$.
		
		To overcome this undesirable property of the average value, we propose
		an alternative evaluation metric that compares the performance of an
		ITR with that of the random treatment rule.  The random treatment rule
		serves as a natural baseline because it treats the same proportion of
		units without any individual information.  This is analogous to the
		predictive setting, in which a classification algorithm is often
		compared to random classification.
		
		Formally, let $p_f = \Pr(f(\bX_i) = 1)$ denote the population
		proportion of units assigned to the treatment condition under ITR
		$f$. Without loss of generality, we assume a positive average
		treatment effect $\tau = \E\{Y_i(1)-Y_i(0)\}>0$ so that the random
		treatment rule assigns the exactly proportion $p_f > 0$ of the units
		to the treatment group. If the treatment is on average harmful (a
		testable condition using the experimental data), the best random
		treatment rule is to treat no one.  In that case, the estimation of
		the average value is sufficient for the evaluation.  We define the
		population average prescription effect (PAPE) of ITR $f$ as the
		following difference in the average value between the ITR and random
		treatment rule,
		\begin{equation}
			\tau_f \ = \ \E\{Y_i(f(\bX_i)) - p_f Y_i(1) - (1-p_f) Y_i(0)\}. \label{eq:PAPE}
		\end{equation}
		
		One motivation for the PAPE is that administering a treatment is often
		expensive.  Consider a costly treatment that does not harm anyone but
		only benefits a relatively small fraction of people.  If we do not
		impose a budget constraint, treating everyone is the best ITR but such
		a policy does not use any individual-level information.  Thus, to
		further evaluate the efficacy of an ITR, we extend the PAPE to the
		settings with a budget constraint.
		
		\subsection{Incorporating a Budget Constraint}
		\label{subsec:budget}
		
		With a budget constraint, we cannot simply treat all units who are
		predicted to benefit from the treatment.  Instead, an ITR must be
		based on a {\it scoring rule} that sorts units according to their
		treatment priority: a unit with a greater score has a higher priority
		to receive the treatment. Let $s: \cX \longrightarrow \mathcal{S}$ be
		such a scoring rule where $\mathcal{S} \subset \mathbb{R}$.  For
		simplicity, we assume that the scoring rule is bijective, i.e.,
		$s(\bx) \ne s(\bx^\prime)$ for any $\bx, \bx^\prime \in \cX$ with
		$\bx \ne \bx^\prime$.  This assumption is not restrictive as we can
		always redefine $\cX$ such that the assumption holds.
		
		We define an ITR based on a scoring rule by assigning a unit to the
		treatment group if and only if its score is higher than a threshold,
		$c$,
		\begin{equation*}
			f(\bX_i, c) \ = \ \mathbf{1}\{s(\bX_i) > c \}.
		\end{equation*}
		Under a binding budget constraint $p$, we define the threshold that
		corresponds to the maximal proportion of treated units under the
		budget constraint, i.e.,
		\begin{equation*}
			c_{p}(f) \ = \ \inf \{c \in \mathbb{R}: \Pr(f(\bX_i, c) = 1) \le p\}.
		\end{equation*}
		
		Our framework allows for any arbitrary scoring rule.  A popular
		scoring rule is the conditional average treatment effect (CATE),
		\begin{equation*}
			s(\bx) \ = \ \E(Y_i(1) - Y_i(0) \mid \bX_i = \bx).
		\end{equation*}
		Researchers have studied the estimation of the CATE using various
		machine learning algorithms such as tree-based methods and regularized
		regressions. 
		
		We emphasize that the scoring rule need not be based on the CATE.  In
		fact, policy makers rely on various indexes.  Such examples include
		the MELD \citep{kamath2001model} that determines liver transplant
		priority, and the IDA Resource Allocation Index that informs the World
		Bank about the provision of economic aid.
		
		Given this setup, we generalize the PAPE to the setting with a budget
		constraint. As before, without loss of generality, we assume the
		treatment is on average beneficial, i.e.,
		$\tau = \E\{Y_i(1)-Y_{i}(0)\}> 0$, so that the constraint is binding
		for the random treatment rule treating at most $100 \times p\%$ of
		units.  The PAPE with a budget constraint $p$ is defined as,
		\begin{equation}
			\tau_{fp} \ = \ \E\{Y_i(f(\bX_i, c_{p}(f))) - p Y_i(1) - (1-p) Y_i(0)\}. \label{eq:PAPEp}
		\end{equation}
		
		A budget constraint facilitates the comparison of multiple ITRs on the
		same footing.  Suppose that we compare two ITRs, $f$ and $g$, using
		the difference in their average values,
		\begin{equation}
			\Delta(f,g) \ = \ \lambda_f - \lambda_g \ = \ \E\{Y_i(f(\bX_i)) - Y_i(g(\bX_i))\}. \label{eq:simplediff}
		\end{equation}
		While this quantity is useful, like the average value, it also fails
		to take into account the proportion of units assigned to the treatment
		condition under each ITR.
		
		We can address this issue by comparing the efficacy of two ITRs under
		the same budget constraint. Formally, we define the Population Average
		Prescriptive Effect Difference (PAPD) under budget $p$ as,
		\begin{equation}
			\Delta_p(f,g) \ = \
			\tau_{fp}-\tau_{gp} \ = \
			\E\{Y_i(f(\bX_i,c_p(f)))-Y_i(g(\bX_i,c_p(g)))\}. \label{eq:PAPD}
		\end{equation}
		
		\subsection{The Area under the Prescriptive Effect Curve}
		\label{subsec:AUPEC}
		
		Since the PAPE (Eqn~\eqref{eq:PAPEp}) varies as a function of budget
		constraint $p$, it would be useful to develop a summary performance
		metric of an ITR over a range of $p$.  We propose the area under the
		prescriptive effect curve (AUPEC) as a metric analogous to the area
		under the receiver operating characteristic curve (AUROC) for
		classification performance. 
		
		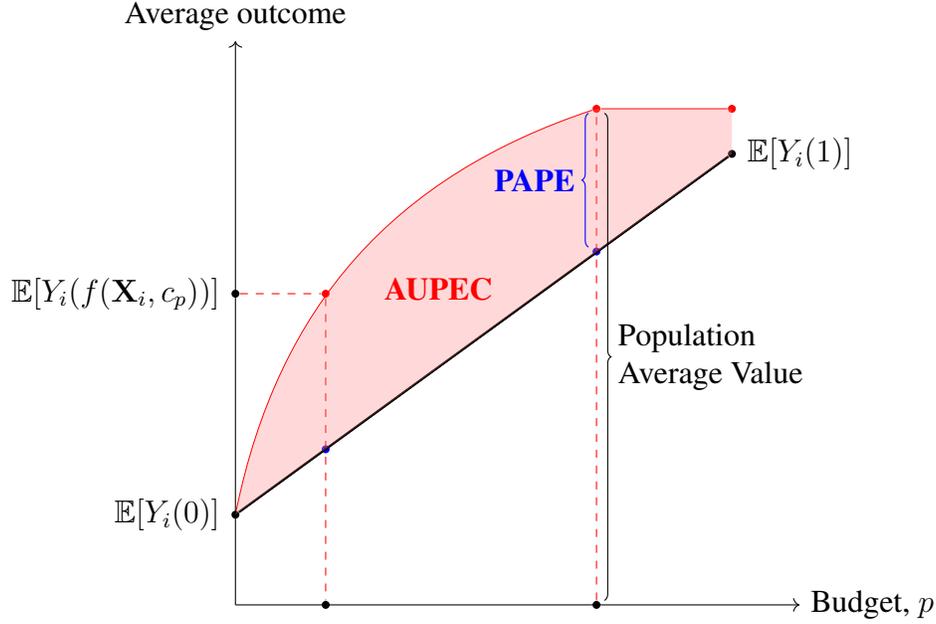
\begin{figure}[t!]
			\centering \spacingset{1}
			\begin{tikzpicture}[scale=3]
				\draw[->] (0,0) -- (2.5,0) coordinate (x axis)  node[right, black]{Budget, $p$};
				\draw[->] (0,0) -- (0,2.5) coordinate (y axis) node[above, black]{Average outcome};
				\node[circle,fill=black,inner sep=0pt,minimum size=3pt,label=left:{$\E[Y_i(0)]$}] (a) at (0,0.4) {};
				\node[circle,fill=black,inner sep=0pt,minimum size=3pt,label=right:{$\E[Y_i(1)]$}] (b) at (2.2,2) {};
				\node[circle,fill=black,inner sep=0pt,minimum size=3pt,label=left:{$\E[Y_i(f(\bX_i,c_{p}))]$}] (dy) at (0, 1.38) {};
				\node[circle,fill=black,inner sep=0pt,minimum size=3pt,label=below:{}] (dx) at (0.4, 0) {};
				\node[circle,fill=black,inner sep=0pt,minimum size=3pt,label=below:{}] (jx) at (1.6, 0) {};
				\node[circle,fill=red,inner sep=0pt,minimum size=3pt] (d) at (0.4, 1.38) {};
				\node[circle,fill=red,inner sep=0pt,minimum size=3pt] (j) at (1.6, 2.2) {};
				\node[circle,fill=red,inner sep=0pt,minimum size=3pt] (m) at (2.2, 2.2) {};
				\node[circle,fill=blue,inner sep=0pt,minimum size=3pt] (j2) at (1.6, 1.567) {};
				\node[circle,fill=blue,inner sep=0pt,minimum size=3pt] (d2) at (0.4, 0.69) {};
				\draw[red, dashed,-] (dy) -- (d) -- (dx);
				\draw[red, dashed,-] (jx) -- (j);
				\draw[black,thick,-] (a) -- (b);
				\draw[black,thick,-] (a) -- (b);
				\draw[red] (a) to [bend left=30] (j);
				\draw[red,-] (j) -- (m);
				\path[fill=red!50,opacity=.3] (a.center) to [bend left=30] (j.center) to (m.center) -- (b.center) --  (a.center);
				\draw[decoration={brace,mirror,raise=3pt},decorate] (jx) --
				node[right=4pt, align=left] {Population \\ Average Value} (j);
				\draw[decoration={brace,raise=3pt},decorate, blue] (j2) --
				node[left=4pt] {\bf PAPE} (j);
				\node[red] at (0.9,1.4) {\bf AUPEC};
			\end{tikzpicture}
			\caption{The Area Under the Prescriptive Effect Curve (AUPEC).  The
				black solid line represents the average value under the random
				treatment rule while the red solid curve represents the average
				value under an ITR $f$.  The difference between the line and the
				curve at a given budget constraint corresponds to the population
				average prescriptive effect (PAPE).  The shaded area between the
				line and the curve represents the AUPEC of $f$.
			} \label{fig:AUPEC}
		\end{figure}

		Figure~\ref{fig:AUPEC} graphically illustrates the AUPEC.  Similar to
		Figure~\ref{fig:PAVcounterexample}, the vertical and horizontal axes
		represent the average outcome and the budget, respectively. The budget
		is operationalized as the maximal proportion treated.  The red solid
		curve corresponds to the average value of an ITR $f$ as a function of
		budget constraint $p$, i.e., $\E\{Y_i(f(\bX_i, c_{p}(f)))\}$, whereas
		the black solid line represents the average value of the random
		treatment rule.  The AUPEC corresponds to the area under the red curve
		minus the area under the black line, which is shown as a red shaded
		area.
		
		Thus, the AUPEC represents the  average performance of an ITR relative
		to  the  random  treatment  rule  over  the  entire  range  of  budget
		constraint (one could also compute the  AUPEC over a specific range of
		budgets).  Unlike the  previous work \citep[e.g.,][]{rzep:jaro:12}, we
		do not require an ITR to assign the maximal proportion of units to the
		treatment condition  though such a  constraint can also be  imposed if
		desired.   For example,  treating more  than a  certain proportion  of
		units will  reduce the average  outcome if these additional  units are
		harmed by the treatment.  This is indicated by the flatness of the red
		line after $p_f$ in Figure~\ref{fig:AUPEC}.
		
		Formally, for a given ITR $f$, we define the AUPEC as,
		\begin{equation}
			\Gamma_{f} \ = \ \int^{p_f}_0\E\{Y_i(f(\bX_i,c_{p}(f)))\} \d p
			+(1-p_f)\E\{Y_i(f(\bX_i,c^\ast))\} - \frac{1}{2}\E(Y_i(0)+Y_i(1)), \label{eq:AUPEC}
		\end{equation}
		where $c^\ast$ is the pre-determined minimum score such that one would
		be treated in the absence of a budget constraint, and
		$p_f=\Pr(f(\bX_i,c^*)=1)$ denotes the maximal proportion of units
		assigned to the treatment condition under the ITR with no budget
		constraint.  The last term represents the area under the random
		treatment rule.
		
		Different values of the threshold $c^\ast$ are possible.  If the goal
		is to treat only those who on average benefit from the treatment, then
		we could use the CATE as a scoring rule and set $c^\ast = 0$.  Another
		example is the use of the MELD score as the scoring rule and choose an
		appropriate value of $c^\ast$ so that a sufficiently healthy patient
		is never considered for a transplant.  Finally, setting
		$c^\ast = -\infty$ would represent the setting where the maximum
		possible units should be treated regardless of the scoring rule.
		
		We further note that the AUPEC is a generalization of the QINI
		coefficient widely utilized in literature for uplift modeling
		\citep{radcliffe2007using}. Formally, the population-level QINI
		coefficient is commonly defined as,
		\begin{equation*}
			\text{QINI}  \ = \ n\left(\int^{1}_0 p\E\{Y_i(1) - Y_i(0) \mid
			f(\bX_i,c_{p}(f)) = 1\} \d p - \frac{1}{2}\E(Y_i(1)-Y_i(0))\right).
		\end{equation*}
		After some algebra, we can rewrite this quantity in the following form,
		\begin{align*}
			n\left(\int^{1}_0\E\{Y_i(f(\bX_i,c_{p}(f)))\} \d p - \frac{1}{2}\E(Y_i(0)+Y_i(1))\right)
		\end{align*}
		Thus, the QINI coefficient is (up to a constant factor $n$) a special
		case of AUPEC when $c^\ast=-\infty$ and $p_f=1$. The choice of
		$c^\ast=-\infty$ may be reasonable in the applications where the
		treatment can be assumed to be never harmful, i.e.,
		$Y_i(1)\geq Y_i(0)$.  In such cases, under no budget constraint one
		would treat the entire population.
		
		To enable a comparison of efficacy across different datasets, the
		AUPEC can be normalized to be scale-invariant by shifting $\Gamma_f$
		by $\E(Y_i(0))$ and dividing by $\tau=\E(Y_i(1)-Y_i(0))$,
		\begin{equation*}
			\Gamma_{f}^\ast \ = \ \frac{1}{\tau}\l[\int^{p_f}_0\E\{Y_i(f(\bX_i,c_p(f)))\} \d p +(1-p_f)\E\{Y_i(f(\bX_i,c^*))\}-\E(Y_i(0))\r] - \frac{1}{2}.
		\end{equation*}
		This normalized AUPEC is invariant to the affine transformation of the
		outcome variables, $Y_i(1), Y_i(0)$, while the standard AUPEC is only
		invariant to a constant shift. The normalized AUPEC takes a value in
		$[0,1]$, and has an intuitive interpretation as the average percentage
		outcome gain using the ITR $f$ compared to the random treatment rule, under the uniform
		prior distribution over the percentage treated.

		\section{Estimation and Inference}
		\label{sec:est-inf}
		
		Having introduced the evaluation metrics,
		we show how to estimate them and compute standard errors under the
		repeated sampling framework of \cite{neym:23}.  Here, we consider the
		setting, in which researchers are interested in evaluating the
		performance of fixed ITRs. In other words, throughout this section,
		ITR $f$ is assumed to be known and has no estimation uncertainty.  For
		example, one may first construct an ITR based on an existing data set
		(experimental or observational), and then conduct a new experiment to
		evaluate its performance.  In Section~\ref{sec:estimatedITR}, we
		extend the methodology developed here to the setting, in which the
		same experimental data set is used to both construct and evaluate an
		ITR.
		
		For the rest of this paper, we assume that we have a simple random
		sample of $n$ units from a super-population, $\mathcal{P}$.  We
		conduct a completely randomized experiment, in which $n_1$ units are
		randomly assigned to the treatment condition with probability $n_1/n$
		and the rest of the $n_0 (= n-n_1)$ units are assigned to the control
		condition.  While it is straightforward to allow for unequal treatment
		assignment probabilities across units, for the sake of simplicity, we
		assume complete randomization.  We formalize these assumptions below.
		\begin{assumption}[Random Sampling of Units] \label{asm:randomsample}
			\spacingset{1} Each of $n$ units, represented by a three-tuple
			consisting of two potential outcomes and pre-treatment covariates,
			is assumed to be independently sampled from a super-population
			$\mathcal{P}$, i.e.,
			$$(Y_i(1), Y_i(0), \bX_i) \ \iid \ \mathcal{P}$$
		\end{assumption}
		\begin{assumption}[Complete Randomization] \label{asm:comrand}
			\spacingset{1} For any $i=1,2,\ldots,n$, the treatment assignment
			probability is given by,
			$$\Pr(T_i = 1 \mid Y_i(1), Y_i(0), \bX_i) \ = \ \frac{n_1}{n}$$
			where $\sum_{i=1}^n T_i = n_1$.
		\end{assumption}
		
		We now present the results under a binding budget constraint.  The
		results for the average value and PAPE with no budget constraint
		appear in Appendix~A.1.
		
		\subsection{The Population Average Prescription Effect (PAPE)}
		\label{subsec:estimationPAPE}
		
		To estimate the PAPE with a binding budget constraint $p$
		(Eqn~\eqref{eq:PAPEp}), we consider the following estimator,
		\begin{equation}
			\hat \tau_{fp}(\cZ) \ = \ \frac{1}{n_1} \sum_{i=1}^n Y_i T_i f(\bX_i,\hat{c}_{p}(f)) +
			\frac{1}{n_0} \sum_{i=1}^n Y_i(1-T_i)(1-f(\bX_i,\hat{c}_{p}(f))) -
			\frac{p}{n_1} \sum_{i=1}^n Y_iT_i -
			\frac{1-p}{n_0} \sum_{i=1}^n Y_i(1-T_i)  \label{eq:PAPEpest}
		\end{equation}
		where
		$\hat{c}_{p}(f) =\inf \{c \in \mathbb{R}: \sum_{i=1}^n f(\bX_i, c) \le
		np\}$ represents the estimated threshold given the maximal proportion
		of treated units $p$.  Unlike the case of no budget constraint (see
		Appendix~A.1.2), the bias is not zero because the
		threshold $c_p(f)$ needs to be estimated without an assumption about
		the distribution of the score produced by the scoring rule.  We derive
		an upper bound of bias and the exact variance.
		\begin{theorem} {\sc (Bias Bound and Exact Variance of the PAPE
				Estimator with a Budget Constraint)} \label{thm:PAPEpest}
			\spacingset{1} Under
			Assumptions~\ref{asm:SUTVA},~\ref{asm:randomsample},~and~\ref{asm:comrand},
			the bias of the proposed estimator of the PAPE with a budget
			constraint $p$ defined in Eqn~\eqref{eq:PAPEpest} can be
			bounded as follows,
			\begin{eqnarray*}
				& & \mc{P}_{\hat{c}_{p}(f)}(|\E\{\hat \tau_{fp}(\cZ) -\tau_{fp}\mid \hat{c}_{p}(f) \} |\geq
				\epsilon)  \\
				& \leq &   1-B(1-p+\gamma_{fp}(\epsilon), n-\lfloor np
				\rfloor , \lfloor np \rfloor+1)+B(1-p-\gamma_{fp}(\epsilon),
				n-\lfloor np
				\rfloor , \lfloor np \rfloor+1),
			\end{eqnarray*}
			where any given constant $\epsilon > 0$,
			$B(\epsilon, \alpha, \beta)$ is the incomplete beta function (if
			$\alpha \leq 0$ and $\beta > 0$, we set
			$B(\epsilon,\alpha, \beta):=H(\epsilon)$ for all $\epsilon$ where
			$H(\epsilon)$ is the Heaviside step function), and
			\begin{equation*}
				\gamma_{fp}(\epsilon)\ = \ \frac{\epsilon}{\max_{c  \in
						[ c_{p}(f) -\epsilon,\   c_{p}(f) +\epsilon]} \E(\tau_i\mid s(\bX_i)=c)}.
			\end{equation*}
			The variance of the estimator is given by,
			\begin{equation*}
				\V(\hat \tau_{fp}(\cZ)) \  = \
				\frac{\E(S_{fp1}^2)}{n_1} +
				\frac{\E(S_{fp0}^2)}{n_0} + \frac{\lfloor np
					\rfloor(n-\lfloor np
					\rfloor)}{n^2(n-1)}\l\{(2p-1)\kappa_{f1}(p)^2-2p\kappa_{f1}(p) \kappa_{f0}(p)\r\},
			\end{equation*}
			where
			$S_{fpt}^2 = \sum_{i=1}^n (Y_{fpi}(t) -
			\overline{Y_{fp}(t)})^2/(n-1)$ and
			$\kappa_{ft}(p) = \E(\tau_i\mid f(\bX_i,\hat{c}_{p}(f))=t)$
			with $Y_{fpi}(t) = \left(f(\bX_i,\hat{c}_{p}(f))- p\right)Y_i(t)$,
			and $\overline{Y_{fp}(t)} = \sum_{i=1}^n Y_{fpi}(t)/n$, for $t= 0,1$.
		\end{theorem}
		Proof is given in Appendix~A.2.
		The last term in the variance accounts for the variance due to
		estimating $c_{p}(f)$.  The variance can be consistently estimated by
		replacing each unknown parameter with its sample analogue, i.e., for
		$t=0,1$,
		\begin{eqnarray}
			\widehat{\E(S_{fpt}^2)} & = &
			\frac{1}{n_t-1} \sum_{i=1}^n \mathbf{1}\{T_i = t\} (Y_{fpi} -
			\overline{Y_{fpt}})^2, \nonumber \\ 
			\widehat{\kappa_{ft}(p)} & = & \frac{\sum_{i=1}^n \mathbf{1}\{f(\bX_i, \hat{c}_{p}(f)) =
				t\} T_i  Y_i}{\sum_{i=1}^n  \mathbf{1}\{f(\bX_i, \hat{c}_{p}(f)) =
				t\} T_i } - \frac{\sum_{i=1}^n   \mathbf{1}\{f(\bX_i,
				\hat{c}_{p}(f)) =  t\} (1-T_i)  Y_i}{\sum_{i=1}^n
				\mathbf{1}\{f(\bX_i,  \hat{c}_{p}(f)) =  t\}
				(1-T_i)}, \nonumber
		\end{eqnarray}
		where $Y_{fpi} = (f(\bX_i,\hat{c}_{p}(f))- p)Y_i$ and
		$\overline{Y}_{fpt} = \sum_{i=1}^n \mathbf{1}\{T_i = t\} Y_{fpi}/n_t$.
		To estimate the term that appears in the denominator of
		$\gamma_{fp}(\epsilon)$ as part of the upper bound of bias, we may
		assume that the CATE, $\E(\tau_i \mid s(\bX_i)=c)$, is continuous in
		$c$. Continuity is often assumed when estimating the CATE
		\citep[e.g.,][]{kunz:etal:18,wage:athe:18}.  We may also utilize an
		upper bound for CATE, if known, to estimate the bias conservatively.
		
		Building on the above results, we also consider the comparison of two
		ITRs under the same budget constraint, using the following estimator
		of the PAPD (Eqn~\eqref{eq:PAPD}),
		\begin{equation}
			\widehat \Delta_p(f,g,\cZ) \ = \ \frac{1}{n_1} \sum_{i=1}^n Y_i T_i \{f(\bX_i,\hat{c}_p(f)) - g(\bX_i,\hat{c}_p(g))\} +
			\frac{1}{n_0} \sum_{i=1}^n Y_i(1-T_i)\{g(\bX_i,\hat{c}_p(g))-f(\bX_i,\hat{c}_p(f))\}. \label{eq:PAPDpest}
		\end{equation}
		Theorem~A3 in Appendix~A.3 derives the
		bias bound and exact variance of this estimator.  In one of their
		applications, \citet{zhou:athe:wage:18} apply the $t$-test to the
		cross-validated test statistic similar to the one introduced here
		under no budget constraint.  However, no formal justification of this
		procedure is given under the cross-validation setting and it cannot be
		readily extended to the case with a budget constraint.  In contrast,
		our methodology is applicable under these settings as well (see
		Section~\ref{sec:estimatedITR}).

		\subsection{The Area under the Prescriptive Effect Curve}
		\label{subsec:estimationAUPEC}
		
		Next, we consider the estimation and inference about the AUPEC
		(Eqn~\eqref{eq:AUPEC}).  Let $n_f$ represent the maximum number
		of units in the sample that the ITR $f$ would assign under no budget
		constraint, i.e., $\hat{p}_f = n_f / n = \sum_{i=1}^n f(\bX_i,c^*)/n$.
		We propose the following estimator of the AUPEC,
		\begin{eqnarray}
			\widehat\Gamma_{f}(\cZ) & = & \frac{1}{n_1} \sum_{i=1}^n Y_i T_i
			\l\{\frac{1}{n}\left(\sum_{k=1}^{n_f}
			f(\bX_i,\hat{c}_{k/n}(f))+(n-n_f)
			f(\bX_i,\hat{c}_{\hat{p}_f}(f))\right) \r\}
			\nonumber \\
			& & +
			\frac{1}{n_0} \sum_{i=1}^n
			Y_i(1-T_i)\left\{1-\frac{1}{n}\left(\sum_{k=1}^{n_f}
			f(\bX_i,\hat{c}_{k/n}(f))+(n-n_f)
			f(\bX_i,\hat{c}_{\hat{p}_f}(f))\right)\right\} \nonumber
			\\ & & -
			\frac{1}{2n_1} \sum_{i=1}^n Y_iT_i -
			\frac{1}{2n_0} \sum_{i=1}^n Y_i(1-T_i).\label{eq:AUPECest}
		\end{eqnarray}
		The following theorem shows a bias bound and the exact variance of
		this estimator.
		\begin{theorem}[Bias and Variance of the AUPEC
			Estimator] \label{thm:AUPECest} \spacingset{1} Under
			Assumptions~\ref{asm:SUTVA},~\ref{asm:randomsample},~and~\ref{asm:comrand},
			the bias of the AUPEC estimator defined in
			Eqn~\eqref{eq:AUPECest} can be bounded as follows,
			\begin{eqnarray*}
				\mc{P}_{\hat{p}_f}(|\E(\widehat\Gamma_{f}(\cZ)-\Gamma_{f}\mid \hat{p}_f) |\geq
				\epsilon) & \leq & 1-B(1-p_f +\gamma_{p_f}(\epsilon), n-\lfloor np_f
				\rfloor , \lfloor np_f \rfloor+1) \\ &+ &B(1-p_f-\gamma_{p_f}(\epsilon),
				n-\lfloor np_f
				\rfloor , \lfloor np_f \rfloor+1)
			\end{eqnarray*}
			where any given constant $\epsilon > 0$,
			$B(\epsilon, \alpha, \beta)$ is the incomplete beta function (if
			$\alpha = 0$ and $\beta > 0$, we set
			$B(\epsilon,\alpha, \beta):=H(\epsilon)$ for all $\epsilon$ where
			$H(\epsilon)$ is the Heaviside step function), and
			\begin{equation*}
				\gamma_{p_f}(\epsilon)\ = \ \frac{\epsilon}{2\max_{c  \in
						[c^*-\epsilon,\  c^*+\epsilon]} \E(\tau_i\mid s(\bX_i)=c)}.
			\end{equation*}
			The variance is given by,
			\begin{eqnarray*}
				& & \V(\widehat\Gamma_{f}(\cZ)) \\
				& = &
				\frac{\E(S_{f1}^{\ast 2})}{n_1} +
				\frac{\E(S_{f0}^{\ast 2})}{n_0} + \E\l[- \frac{1}{n} \l\{\sum_{z=1}^Z
				\frac{z(n-z)}{n^2(n-1)}\kappa_{f1}(z/n)\kappa_{f0}(z/n)
				+  \frac{Z(n-Z)^2}{n^2(n-1)}\kappa_{f1}(Z/n)\kappa_{f0}(Z/n)\r\}\r.\\
				& & \hspace{.275in}
				- \frac{2}{n^4(n-1)}\sum_{z=1}^{Z-1}\sum_{z^\prime=z+1}^{Z}
				z(n-z^\prime)\kappa_{f1}(z/n)\kappa_{f1}(z^\prime/n)
				-  \frac{Z^2(n-Z)^2}{n^4(n-1)}\kappa_{f1}(Z/n)^2\\
				& & \hspace{.275in} \l.
				- \frac{2(n-Z)^2}{n^4(n-1)}\sum_{z=1}^Z z
				\kappa_{f1}(Z/n)\kappa_{f1}(z/n) +\frac{1}{n^4}\sum_{z=1}^Z
				z(n-z)\kappa_{f1}(z/n)^2\r] \\
				& & \hspace{.1in}+ \V\l(\sum_{z=1}^Z \frac{z}{n} \kappa_{f1}(z/n) + \frac{(n-Z)Z}{n} \kappa_{f1}(Z/n)\r),
			\end{eqnarray*}
			where $Z$ is a Binomial random variable with size $n$ and success
			probability $p_f$, and
			$S^{\ast 2}_{ft} = \sum_{i=1}^n (Y^\ast_i(t) -
			\overline{Y^\ast(t)})^2/(n-1)$,
			$\kappa_{ft}(k/n) \ = \ \E(Y_i(1)-Y_i(0)\mid
			f(\bX_i,\hat{c}_{k/n}(f))=t)$, with
			$Y_i^\ast(t) = \left[\left\{\sum_{z=1}^{n_f}
			f(\bX_i,\hat{c}_{z/n}(f))+(n-n_f)
			f(\bX_i,\hat{c}_{\hat{p}_f}(f))\right\}/n -
			\frac{1}{2}\right]Y_i(t)$ and
			$\overline{Y^\ast(t)} = \sum_{i=1}^n Y_i^\ast(t)/n$, for $t=0,1$.

		\end{theorem}
		Proof is given in in Appendix~A.4.  When
		$c^\ast=-\infty$ (i.e., the AUPEC equals the QINI coefficient), the
		estimator is unbiased, implying that the bias comes from estimating
		the proportion treated $p_f$ under no budget constraints and
		$c^\ast>-\infty$. As before, $\E({S_{ft}^\ast}^2)$ does not equal
		$\V(Y_i^\ast(t))$ due to the need to estimate the terms $c_{z/n}$ for
		all $z$, and the additional terms account for the variance of
		estimation.  We can consistently estimate the upper bound of bias, for
		example, under by assuming that the CATE is continuous. We may also
		utilize an upper bound for CATE, if known, to estimate the bias
		conservatively.
		
		To estimate the variance, we replace each unknown parameter with its
		sample analogue,
		\begin{eqnarray}
			\widehat{\E(S^{\ast 2}_{ft})} & = & \frac{1}{n_t-1}
			\sum_{i=1}^n \mathbf{1}\{T_i = t\} (Y^\ast_i -
			\overline{Y^\ast_t})^2, \nonumber \\ 
			\hat{\kappa}_{ft}(z/n)  & = & \frac{\sum_{i=1}^n
				\mathbf{1}\{f(\bX_i,
				\hat{c}_{z/n}(f)) = t\}  T_i Y_i}{\sum_{i=1}^n
				\mathbf{1}\{f(\bX_i,
				\hat{c}_{z/n}(f)) = t\}  T_i }  - \frac{\sum_{i=1}^n
				\mathbf{1}\{f(\bX_i,
				\hat{c}_{z/n}(f)) = t\}  (1-T_i) Y_i}{\sum_{i=1}^n \mathbf{1}\{f(\bX_i,
				\hat{c}_{z/n}(f)) = t\}  (1-T_i)}, \nonumber \\ \label{eq:kappa2est}
		\end{eqnarray}
		for $t=0,1$.
		In the extreme cases with $z\to 1$ for $t=1$ and $z\to n$ for $t=0$,
		each denominator in Eqn~\eqref{eq:kappa2est} is likely to be
		close to zero. In such cases, we instead use the estimator
		$\widehat{\kappa_{f1}(z_{\min}/n)}$ for all $z<z_{\min}$ where
		$z_{\min}$ is the smallest $z$ such that Eqn~\eqref{eq:kappa2est}
		for $\kappa_{f1}(z/n)$ does not lead to division by zero. Similarly,
		for $t=0$, we use $\widehat{\kappa_{f0}(z_{\max}/n)}$ for all
		$z>z_{\max}$ where $z_{\max}$ is the largest $z$.
		
		For the terms involving the binomial random variable $Z$, we first
		note that, when fully expanded out, they are the polynomials of
		$p_f=\E(f(\bX_i))$.  To estimate the polynomials, we can utilize their
		unbiased estimators as discussed in \cite{stuard1994kendall}, i.e.,
		$\hat{p}_f^z = s(s-1)\cdots (s-z+1)/\{n(n-1)\cdots (n-z+1)\}\}$ where
		$s = \sum_{i=1}^n f(\bX_i)$ is unbiased for $p_f^z$ for all
		$z \leq n$.  When the sample size is large, this estimation method is
		computationally inefficient and unstable due to the presence of high
		powers.  Hence, we may use the Monte Carlo sampling of $Z$ from a
		Binomial distribution with size $n$ and success probability
		$\hat{p}_f$.  In our simulation study, we show that this Monte Carlo
		approach is effective even when the sample size is small (see
		Section~\ref{sec:synthetic}).
		
		Finally, a consistent estimator for the normalized AUPEC is given by,
		\begin{eqnarray}
			\widehat\Gamma^\ast_{f}(\cZ) & = & \frac{1}{\sum_{i=1}^n
				Y_iT_i/n_1-Y_i(1-T_i)/n_0}\l\{\frac{1}{nn_1}
			\sum_{i=1}^n Y_i T_i
			\left(\sum_{z=1}^{n_f} f(\bX_i,\hat{c}_{z/n}(f))+(n-n_f)
			f(\bX_i,\hat{c}_{\hat{p}_f}(f))\right)\r.
			\nonumber \\
			& & \l. -
			\frac{1}{nn_0} \sum_{i=1}^n
			Y_i(1-T_i)\left(\sum_{z=1}^{n_f}
			f(\bX_i,\hat{c}_{z/n}(f))+(n-n_f)
			f(\bX_i,\hat{c}_{\hat{p}_f}(f))\right)\r\} - \frac{1}{2}. \label{eq:AUPECnest}
		\end{eqnarray}
		The variance of $\widehat\Gamma^\ast_{f}(\cZ)$ can be estimated using
		the Taylor expansion of quotients of random variables to an
		appropriate order as detailed in \cite{stuard1994kendall}.
		
		\section{Estimating and Evaluating ITRs Using the Same Experimental Data}
		\label{sec:estimatedITR}
		
		We next consider a common situation, in which researchers use the same
		experimental data to both estimate and evaluate an ITR via
		cross-validation.  This differs from the setting we have analyzed so
		far, in which a fixed ITR is given for evaluation. We first extend the
		evaluation metrics introduced in Section~\ref{sec:evalmetrics} to the
		current setting with estimated ITRs.  We then develop inferential
		methods under the Neyman's repeated sampling framework by accounting
		for both estimation and evaluation uncertainties.  Below, we consider
		the scenario, in which researchers face a binding budget constraint.
		Appendix~A.5 presents the results for the case with
		no budget constraint.
		
		\subsection{Evaluation Metrics}
		\label{subsec:evalcv}
		
		Suppose that we have the data from a completely randomized experiment
		as described in Section~\ref{sec:est-inf}.  We first estimate an ITR
		$f$ by applying a machine learning algorithm $F$ to training data
		$\cZ^{tr}$.  Then, under a budget constraint of the maximal proportion
		of treated units $p$, we use test data to evaluate the resulting
		estimated ITR $\hat{f}_{\cZ^{tr}}$.  As before, we assume that this
		constraint is binding, i.e., $p < p_F$ where
		$p_F = \Pr(\hat{f}_{\cZ^{tr}}(\bX_i) = 1)$ represents the proportion
		of treated units under the ITR without a budget constraint.
		
		Formally, an machine learning algorithm $F$ is a deterministic map
		from the space of training data $\cZ$ to that of scoring rules
		$\mathcal{S}$,
		\begin{equation*}
			F: \mathcal{Z} \to \mathcal{S}.
		\end{equation*}
		Then, for a given training data set $\cZ^{tr}$, the estimated ITR is
		given by,
		\begin{equation*}
			\hat{f}_{\cZ^{tr}}(\bX_i, c_{p}(\hat{f}_{\cZ^{tr}})) \ = \ \mathbf{1}\{\hat{s}_{\cZ^{tr}}(\bX_i) > c_{p}(\hat{f}_{\cZ^{tr}})\},
		\end{equation*}
		where $\hat{s}_{\cZ^{tr}}=F(\cZ^{tr})$ is the estimated scoring rule
		and
		$c_{p}(\hat{f}_{\cZ^{tr}})\ = \ \inf \{c \in \mathbb{R}:
		\Pr(\hat{f}_{\cZ^{tr}}(\bX_i, c) = 1 \mid \cZ^{tr}) \le p\}$ is the
		threshold based on the maximal proportion of treated units $p$.  The
		CATE is a natural choice for the scoring rule, i.e.,
		$s_{\cZ^{tr}}(\bX_i)= \E(\tau_i \mid \bX_i)$.  We need not assume that
		$\hat{s}_{\cZ^{tr}}(\bX_i)$ is consistent for the CATE.
		
		To extend the PAPE (Eqn~\eqref{eq:PAPEp}), we first define the
		population proportion of units with $\bX_i = \bx$ who are assigned to
		the treatment condition under the estimated ITR as,
		\begin{equation*}
			\bar{f}_{Fp}(\bx) \ = \ \E_{\cZ^{tr}}\{\hat{f}_{\cZ^{tr}}(\bX_i,
			c_{p}(\hat{f}_{\cZ^{tr}})) \mid \bX_i = \bx\} \ = \ \Pr \{\hat{f}_{\cZ^{tr}}(\bX_i,
			c_{p}(\hat{f}_{\cZ^{tr}})) = 1\mid \bX_i = \bx\}.
		\end{equation*}
		While $c_{p}(\hat{f}_{\cZ^{tr}})$ depends on the specific ITR
		generated from the training data $\cZ^{tr}$, the population proportion
		of treated units averaged over the sampling of training data,
		$\bar{f}_{Fp}(\bX_i)$, only depends on $p$.
		
		Lastly, the PAPE of the estimated ITR under budget constraint $p$ is
		defined as,
		\begin{equation*}
			\tau_{Fp} \ = \ \E\{\bar{f}_{Fp}(\bX_i) Y_i(1) + (1-\bar{f}_{Fp}(\bX_i))Y_i(0) - p Y_i(1) - (1-p) Y_i(0)\}. \label{eq:PAPEpcv}
		\end{equation*}
		This evaluation metric corresponds to neither that of a specific ITR
		estimated from the whole experimental data set nor its expectation.
		Rather, we are evaluating the efficacy of a learning algorithm that is
		used to estimate an ITR using the same experimental data.
		
		We can also compare {\it estimated} ITRs by further generalizing the
		definition of the PAPD (Eqn~\eqref{eq:PAPD}) to the current
		setting.  Specifically, we define the PAPD between two machine learning algorithms,
		$F$ and $G$, under budget constraint $p$ as,
		\begin{equation}
			\Delta_p(F,G)\ = \ \E_{\bX, Y}[\{\bar{f}_{Fp}(\bX_i) -
			\bar{f}_{Gp}(\bX_i)\}Y_{i}(1) + \{\bar{f}_{Gp}(\bX_i) -
			\bar{f}_{Fp}(\bX_i)\}Y_i(0)].
			\label{eq:PAPDpcv}
		\end{equation}
		Finally, we consider the AUPEC of an estimated ITR.  Specifically, the
		AUPEC of an machine learning algorithm $F$ is defined as,
		\begin{equation}
			\Gamma_F \ = \ \E_{\cZ^{tr}}\l[\int^{p_{\hat{f}}}_0
			\E\{Y_i(\hat{f}_{\cZ^{tr}}(\bX_i,c_{p}(\hat{f}_{\cZ^{tr}}))) \d p +
			(1-p_{\hat{f}})\E\{Y_i(\hat{f}_{\cZ^{tr}}(\bX_i,c^*))\}\r] -
			\frac{1}{2}\E(Y_i(0)+Y_i(1)), \label{eq:AUPECcv}
		\end{equation}
		where $p_{\hat{f}}=\Pr(\hat{f}_{\cZ^{tr}}(\bX_i)=1)$ is the maximal
		population proportion of units treated by the estimated ITR
		$\hat{f}_{\cZ^{tr}}$.
		
		\subsection{Estimation and Inference}
		\label{subsec:cv}
		
		Rather than simply splitting the data into training and test sets (in
		such a case, the inferential procedure for fixed ITRs is applicable),
		we maximize efficiency by using cross-validation to estimate the
		evaluation metrics introduced above.  First, we randomly split the
		data into $K$ subsamples of equal size $m = n/K$ by assuming, for the
		sake of notational simplicity, that $n$ is a multiple of $K$. Then,
		for each $k=1,2,\ldots,K$, we use the $k$th subsample as a test set
		$\cZ_k=\{\bX_i^{(k)}, T_i^{(k)}, Y_i^{(k)}\}_{i=1}^{m}$ with the data
		from all $(K-1)$ remaining subsamples as the training set
		$\cZ_{-k}=\{\bX_i^{(-k)}, T_i^{(-k)},
		Y_i^{(-k)}\}_{i=1}^{n-m}$.
		
		To simplify notation, we assume that the number of treated (control)
		units is identical across different folds and denote it as $m_1$
		($m_0=m-m_1$).  For each split $k$, we estimate an ITR by applying a
		learning algorithm $F$ to the training data $\cZ_{-k}$,
		\begin{equation}
			\hat{f}_{-k} \ = \ F(\cZ_{-k}). \label{eq:f.hat}
		\end{equation}
		We then evaluate the performance of the learning algorithm $F$ by
		computing an evaluation metric of interest $\tau$ based on the test
		data $\cZ_k$.  Repeating this $K$ times for each $k$ and averaging the
		results gives a cross-validation estimator of the evaluation metric.
		Algorithm~\ref{alg:cross_validation} formally presents this estimation
		procedure. The variance formula for the estimated evaluation metric is
		omitted as it is generally a complex function $v(\cdot)$ (see
		Theorem~\ref{thm:PAPEpcvest} below).
		
		
		\begin{algorithm}[!t]
			\spacingset{1}
			\hspace*{\algorithmicindent} \textbf{Input}: Data $\cZ=\{\bX_i, T_i,
			Y_i\}_{i=1}^n$, Machine learning algorithm $F$, Evaluation metric
			$\tau_f$, Number of folds $K$ \\
			\hspace*{\algorithmicindent} \textbf{Output}: Estimated evaluation
			metric $\hat\tau_F$, Estimated variance of $\hat\tau_F$
			\begin{algorithmic}[1]
				\State Split data into $K$ random subsets of equal size $(\cZ_1,\cdots,\cZ_k)$
				\State $k\gets 1$
				\While{$k\leq K$}
				\State $\cZ_{-k}=[\cZ_1, \cdots, \cZ_{k-1},\cZ_{k+1}, \cdots, \cZ_{K}]$
				\State $\hat{f}_{-k} \ = \ F(\cZ_{-k})$ \Comment{Estimate ITR $f$
					by applying $F$ to $\cZ_{-k}$}
				\State $ \hat{\tau}_k =  \hat{\tau}_{\hat{f}_{-k}}(\cZ_k)$ \Comment{Evaluate estimated ITR $\hat{f}$ using $\cZ_k$}
				\State $k \gets k + 1$
				\EndWhile
				\State \textbf{return} $\hat{\tau}_F = \frac{1}{K} \sum_{k=1}^K \hat{\tau}_k$, $\widehat{\V(\hat{\tau}_F)}=v(\hat{f}_{-1},\cdots,\hat{f}_{-k},\cZ_1, \cdots, \cZ_{K})$
			\end{algorithmic}
			\caption{Estimating and Evaluating an Individualized Treatment Rule
				(ITR) using the Same Experimental Data via Cross-Validation}
			\label{alg:cross_validation}
		\end{algorithm}
		
		We develop the inferential methodology for the evaluation based on the
		cross-validation procedure described above under the Neyman's repeated
		sampling framework.  We focus on the case with a binding budget
		constraint.  The results with no budget constraint appear in
		Appendix~A.5.  We begin by introducing the
		cross-validation estimator of the PAPE with a binding budget
		constraint $p$,
		\begin{equation}
			\hat{\tau}_{Fp} \ = \ \frac{1}{K}\sum_{k=1}^K \hat\tau_{\hat{f}_{-k},p}(\cZ_k), \label{eq:PAPEpcvest}
		\end{equation}
		where $\hat\tau_{fp}$ is defined in Eqn~\eqref{eq:PAPEpest}.
		
		Like the fixed ITR case, the bias of the proposed estimator is not
		exactly zero.  However, we are able to
		show that the bias can be upper bounded by a small quantity while the
		exact randomization variance can still be derived.
		\begin{theorem} {\sc (Bias Bound and Exact Variance of the
				Cross-Validation PAPE Estimator with a Budget
				Constraint)} \label{thm:PAPEpcvest} \spacingset{1} Under
			Assumptions~\ref{asm:SUTVA},~\ref{asm:randomsample},~and~\ref{asm:comrand},
			the bias of the cross-validation PAPE estimator with a budget
			constraint $p$ defined in Eqn~\eqref{eq:PAPEpcvest} can be
			bounded as follows,
			\begin{align*}
				\E_{\cZ^{tr}}[\mc{P}_{\hat{c}_{p}(\hat{f}_{\cZ^{tr}})}(|\E\{\hat{\tau}_{Fp} -\tau_{Fp} \mid  \hat{c}_{p}(\hat{f}_{\cZ^{tr}})\} |\geq
				\epsilon) ] \ &\leq \  1-B(1-p+\gamma_p(\epsilon), m - \lfloor mp \rfloor , \lfloor mp \rfloor+1)\\&+B(1-p-\gamma_p(\epsilon),
				m - \lfloor mp \rfloor , \lfloor mp \rfloor+1),
			\end{align*}
			where any given constant $\epsilon > 0$,
			$B(\epsilon, \alpha, \beta)$ is the incomplete beta function (if
			$\alpha = 0$ and $\beta > 0$, we set
			$B(\epsilon,\alpha, \beta):=H(\epsilon)$ for all $\epsilon$ where
			$H(\epsilon)$ is the Heaviside step function), and
			\begin{equation*}
				\gamma_{p}(\epsilon)\ = \
				\frac{\epsilon}{\E_{\cZ^{tr}}\{\max_{c  \in   [c_{p}(\hat{f}_{\cZ^{tr}})-\epsilon,\  c_{p}(\hat{f}_{\cZ^{tr}})+\epsilon]} \E_{\cZ}(\tau_i\mid \hat{s}_{\cZ^{tr}}(\bX_i)=c)\}}.
			\end{equation*}
			The variance of the estimator is given by,
			\begin{eqnarray*}
				\V(\hat \tau_{Fp})& = &
				\frac{\E(S_{\hat{f}p1}^2)}{m_1} +
				\frac{\E(S_{\hat{f}p0}^2)}{m_0} +
				\frac{\lfloor mp \rfloor(m-\lfloor mp
					\rfloor)}{m^2(m-1)}\l\{(2p-1)\kappa_{F1}(p)^2-2p\kappa_{F1}(p)\kappa_{F0}(p)\r\} \\
				& & \hspace{0.5in} - \frac{K-1}{K}\E(S_{Fp}^2),
			\end{eqnarray*}
			where
			$S_{\hat{f}pt}^2 = \sum_{i=1}^{m} (Y_{\hat{f}pi}(t) -
			\overline{Y_{\hat{f}p}(t)})^2/(m-1)$, and
			$S_{Fp}^2 = \sum_{k=1}^K (\hat\tau_{\hat{f}_{-k},p}(\cZ_k) -
			\overline{\hat\tau_{\hat{f}_{-k},p}(\cZ_k)})^2/(K-1)$, and
			$\kappa_{Ft}(p) = \E(\tau_i\mid
			\hat{f}_{\cZ^{tr}}(\bX_i,\hat{c}_{p}(\hat{f}_{\cZ^{tr}}))=t)$, with
			$Y_{\hat{f}pi}(t) =
			\{\hat{f}_{\cZ^{tr}}(\bX_i,\hat{c}_{p}(\hat{f}_{\cZ^{tr}}))-
			p\}Y_i(t)$,
			$\overline{Y_{\hat{f}p}(t)} = \sum_{i=1}^n Y_{\hat{f}pi}(t)/n$, and
			$\overline{\hat\tau_{\hat{f}_{-k},p}(\cZ_k)} = \sum_{k=1}^K
			\hat\tau_{\hat{f}_{-k},p}(\cZ_k)/K$, for $t= 0,1$.
		\end{theorem}
		Proof is given in Appendix~A.6.  The estimation of
		the term $\E(\widetilde{S}_{\hat{f}t}^2)$ is done similarly as
		before. For $\kappa_{Ft}(p)$, we replace it with its sample analogue:
		\begin{align}
			\widehat{\kappa_{Ft}(p)}  \ &= \ \frac{1}{K} \sum_{l=1}^K \frac{\sum_{i=1}^m
				\mathbf{1}\{\hat{f}_{-k}(\bX_i,
				\hat{c}_{p}(\hat{f}_{-k})) = t\}  T_i^{(k)} Y_i^{(k)}}{\sum_{i=1}^m
				\mathbf{1}\{\hat{f}_{-k}(\bX_i,
				\hat{c}_{p}(\hat{f}_{-k})) = t\}  T_i^{(k)} }  \nonumber \\ &
			\hspace{1in} - \frac{\sum_{i=1}^m
				\mathbf{1}\{\hat{f}_{-k}(\bX_i,
				\hat{c}_{p}(\hat{f}_{-k})) = t\}  (1-T_i^{(k)}) Y_i^{(k)}}{\sum_{i=1}^m \mathbf{1}\{\hat{f}_{-k}(\bX_i,
				\hat{c}_{p}(\hat{f}_{-k})) = t\}  (1-T_i^{(k)})}. \label{eq:kappacvest}
		\end{align}
		To estimate the term that appears in the denominator of
		$\gamma_p(\epsilon)$ as part of the upper bound of bias, we assume
		that the CATE, i.e.,
		$\E(Y_i(1)-Y_i(0) \mid \hat{s}_{\cZ^{tr}}(\bX_i)=c)$, is continuous in
		$c$, and replace the maximum with a point estimate.  We may also
		utilize an upper bound for CATE, if known, to estimate the bias
		conservatively. Building on this result, Appendix~A.7 shows how to
		compare two estimated ITRs by estimating the PAPD
		(Eqn~\eqref{eq:PAPDpcv}).
		
		Finally, we consider the following cross-validation estimator of the
		AUPEC for an estimated ITR (Eqn~\eqref{eq:AUPECcv}),
		\begin{equation}
			\widehat{\Gamma}_F \ = \ \frac{1}{K}\sum_{k=1}^K
			\widehat{\Gamma}_{\hat{f}_{-k}}(\cZ_k), \label{eq:AUPECcvest}
		\end{equation}
		where $\widehat{\Gamma}_f$ is defined in
		Eqn~\eqref{eq:AUPECest}. This can be seen as an overall statistic that measures the prescriptive performance of the machine learning algorithm $F$ on the dataset under cross-validation. The following theorem derives a bias bound
		and the exact variance of this cross-validation estimator.
		\begin{theorem} {\sc (Bias Bound and Exact Variance of the
				Cross-validation AUPEC Estimator)}
			\label{thm:AUPECcvest} \spacingset{1} Under
			Assumptions~\ref{asm:SUTVA},~\ref{asm:randomsample},~and~\ref{asm:comrand},
			the bias of the AUPEC estimator defined in
			Eqn~\eqref{eq:AUPECcvest} can be bounded as follows,
			\begin{eqnarray*}
				\E_{\cZ^{tr}}[\mc{P}_{\hat{p}_{\hat{f}}}(|\E(\widehat\Gamma_{F}-\Gamma_{F}\mid \hat{p}_{\hat{f}}) |\geq \epsilon)]
				& \leq & \E\{1-B(1-p_{\hat{f}} +\gamma_{p_{\hat{f}}}(\epsilon), m-\lfloor mp_{\hat{f}}
				\rfloor, \lfloor mp_{\hat{f}} \rfloor+1) \\
				& & \hspace{.5in} + B(1-p_{\hat{f}}-\gamma_{p_{\hat{f}}}(\epsilon), m-\lfloor mp_{\hat{f}}
				\rfloor, \lfloor mp_{\hat{f}} \rfloor+1)\},
			\end{eqnarray*}
			where any given constant $\epsilon > 0$, $B(\epsilon, \alpha, \beta)$
			is the incomplete beta function (if $\alpha = 0$ and $\beta > 0$, we
			set $B(\epsilon,\alpha, \beta):=H(\epsilon)$ for all $\epsilon$ where
			$H(\epsilon)$ is the Heaviside step function), and
			\begin{equation*}
				\gamma_{p_{\hat{f}}}(\epsilon)\ = \ \frac{\epsilon}{2\E_{\cZ^{tr}}\{\max_{c
						\in [c^*-\epsilon,\ c^*+\epsilon]}
					\E(\tau_i\mid \hat{s}_{\cZ^{tr}}(\bX_i)=c)\}}.
			\end{equation*}
			The variance is given by,
			\begin{eqnarray*}
				\V(\widehat\Gamma_{F})
				& = & \E\l[- \frac{1}{m} \l\{\sum_{z=1}^Z \frac{k(n-z)}{m^2(m-1)}\kappa_{F1}(z/m)\kappa_{F0}(z/m)
				+  \frac{Z(m-Z)^2}{m^2(m-1)}\kappa_{F1}(Z/m)\kappa_{F0}(Z/m)\r\}\r.\\
				& & \hspace{.275in} - \frac{2}{m^4(m-1)}\sum_{z=1}^{Z-1}\sum_{z^\prime=z+1}^{Z}
				z(m-z^\prime)\kappa_{F1}(z/m)\kappa_{F1}(z^\prime/m) - \frac{Z^2(m-Z)^2}{m^4(m-1)}\kappa_{F1}(Z/m)^2\\
				& & \hspace{.275in} \l.
				-  \frac{2(m-Z)^2}{m^4(m-1)}\sum_{z=1}^Z k\kappa_{F1}(Z/m)\kappa_{F1}(z/m)
				+\frac{1}{m^4}\sum_{z=1}^Z  z(m-z)\kappa_{F1}(z/m)^2\r] \\
				& & + \V\l(\sum_{z=1}^Z \frac{z}{m}  \kappa_{F1}(z/m) + \frac{(m-Z)Z}{m}
				\kappa_{F1}(Z/m)\r)+  \frac{\E(S_{\hat{f}1}^{\ast 2})}{m_1} +
				\frac{\E(S_{\hat{f}0}^{\ast 2})}{m_0} - \frac{K-1}{K}\E(S_{F}^{\ast 2}),
			\end{eqnarray*}
			where $Z$ is a Binomial random variable with size $m$ and success
			probability $p_{\hat{f}}$,
			$S^{\ast 2}_{\hat{f}t} = \sum_{i=1}^m (Y^\ast_{\hat{f}i}(t) -
			\overline{Y_{\hat{f}}^\ast(t)})^2/(m-1)$,
			$S_{F}^2 = \sum_{k=1}^K (\widehat{\Gamma}_{\hat{f}_{-k}}(\cZ_k)-
			\overline{\widehat{\Gamma}_{\hat{f}_{-k}}(\cZ_k)})^2/(K-1)$, and
			$\kappa_{Ft}(z/m) = \E(\tau_i\mid
			\hat{f}_{\cZ^{tr}}(\bX_i,\hat{c}_{z/m}(\hat{f}_{\cZ^{tr}})=t)$ with
			$\overline{Y_{\hat{f}}^\ast(t)} = \sum_{i=1}^m
			Y_{\hat{f}i}^\ast(t)/m$,
			$\overline{\widehat{\Gamma}_{\hat{f}_{-k}}(\cZ_k)} = \sum_{k=1}^K
			\widehat{\Gamma}_{\hat{f}_{-k}}(\cZ_k)/K$, and
			$Y_{\hat{f}i}^\ast(t) = \left[\left\{\sum_{z=1}^{m_f}
			\hat{f}_{\cZ^{tr}}(\bX_i,\hat{c}_{z/m}(\hat{f}_{\cZ^{tr}}))+(m-m_f)
			\hat{f}_{\cZ^{tr}}(\bX_i,\hat{c}_{z/m}(\hat{f}_{\cZ^{tr}}))\right\}/m
			- \frac{1}{2}\right]Y_i(t)$ for $t=0,1$.
		\end{theorem}
		Proof is similar to that of Theorem~\ref{thm:AUPECest}. The estimation
		of $\E(S_{\hat{f}1}^{\ast 2})$, $\E(S_{\hat{f}0}^{\ast 2})$, and
		$\E(S_{F}^{\ast 2})$ is the same as before, and the $\kappa_{Ft}(p)$
		term can be estimated using Eqn~\eqref{eq:kappacvest}.

		\section{A Simulation Study}
		\label{sec:synthetic}
		
		We conduct a simulation study to examine the finite sample performance
		of our methodology, for both fixed and estimated ITRs.  We find that
		the empirical coverage probability of the confidence interval, based
		on the proposed variance, approximates its nominal rate even in a
		small sample.  We also find that the bias is minimal even when the
		proposed estimator is not unbiased and that our variance bounds are
		tight.
		
		\subsection{Data Generation Process}
		
		Our data generating process (DGP) is based on the one used in the
		2017 Atlantic Causal Inference Conference (ACIC) Data Analysis
		Challenge \citep{hahn:dori:murr:18}.  
		A total of 8 covariates $\bX$ are taken from the Infant Health and
		Development Program, which originally had 58 covariates and $4,302$
		observations.  In our simulation, the population distribution of
		covariates is assumed to equal the empirical distribution of this data
		set.  Therefore, we obtain each simulation sample via boostrap.  We
		vary the sample size: $n= 100, 500, 2000$.
		
		We use the same outcome model as the one used in the competition,
		\begin{equation}
			\E(Y_i(t) \mid \bX_i) \ = \ \mu(\bX_i)+\tau(\bX_i)t, \label{eq:outcome}
		\end{equation}
		where $\pi(\bx) = 1/[1+\exp\{3(x_1+x_{43}+0.3(x_{10}-1))-1\}]$,
		$\mu(\bx) = -\sin(\Phi(\pi(\bx)))+x_{43}$, and
		$\tau(\bx) = \xi(x_3x_{24}+(x_{14}-1)-(x_{15}-1))$ with $\Phi(\cdot)$
		representing the standard Normal CDF and $x_j$ indicating a specific
		covariate in the data set.  One important difference is that we assume
		a complete randomized experiment whereas the original DGP generated
		the treatment using a function of covariates to emulate an
		observational study.  As in the competition, we focus on two scenarios
		regarding the treatment effect size by setting $\xi$ equal to $2$
		(``high'') and $1/3$ (``low'').  Although the original DGP included
		four different error distributions, we use the i.i.d. error,
		$\sigma(\bX_i)\epsilon_i$ where
		$\sigma(\bx) = 0.25 \sqrt{\V(\mu(\bx)+\pi(\bx)\tau(\bx))}$ and
		$\epsilon_i \iid \mathcal{N}(0,1)$.
		
		For fixed ITRs, we can directly compute the true values of our causal
		quantities of interest using the outcome model specified in
		Eqn~\eqref{eq:outcome} and evaluate each quantity based on the
		entire original data set. This computation is valid because we assume
		the population distribution of covariates is equal to the empirical
		distribution of the original data set.  For the estimated ITR case,
		however, we do not have an analytical expression for the true value of
		a causal quantity of interest.  Therefore, we obtain an approximate
		true value via Monte Carlo simulation.  We generate 10,000 independent
		datasets based on the same DGP, and train the specified algorithm $F$
		on each of the datasets using 5-fold cross-validation (i.e., $K=5$).
		Then, we use the sample mean of our estimated causal quantity across
		10,000 data sets as our approximate truth.
		
		We evaluate Bayesian Additive Regression Trees (BART)
		\citep[][]{chipman2010bart,hahn2020bayesian}, which had the best
		overall performance in the original competition.  We compare this
		model with two other popular methods: Causal Forest
		\citep{athey2019generalized} as well as the LASSO, which includes all
		main effects and two-way interaction effects between the treatment and
		all covariates \citep{tibshirani1996regression}.  All three models are
		trained on the original data from the 2017 ACIC Data Challenge. The
		number of trees was tuned through the 5-fold cross validation for BART
		and Causal Forest. The regularization parameter was tuned similarly
		for LASSO. All models were cross-validated on the PAPE.  For
		implementation, we use {\sf R 3.4.2} with {\sf bartMachine} (version
		1.4.2) for BART, {\sf grf} (version 0.10.2) for Causal Forest, and
		{\sf glmnet} (version 2.0.13) for LASSO.  Once the models are trained,
		an ITR is derived based on the magnitude of the estimated CATE
		$\hat\tau(\bx)$, i.e., $f(\bX_i) = \mathbf{1}\{\hat\tau(\bX_i) > 0\}$.
		
		\subsection{Results}
		\label{subsec:coverage}
		
		\begin{table}[t!]
			\centering\setlength{\tabcolsep}{2pt}
			\small \spacingset{1}
			\resizebox{1.1\textwidth}{!}{\begin{tabular}{l.|...|...|...}
					\hline
					& & \multicolumn{3}{c|}{$\bm{n=100}$} & \multicolumn{3}{c|}{$\bm{n=500}$} & \multicolumn{3}{c}{$\bm{n=2000}$}\\
					Estimator & \multicolumn{1}{c|}{truth}
					& \multicolumn{1}{c}{coverage} & \multicolumn{1}{c}{bias}
					& \multicolumn{1}{c|}{s.d.}
					& \multicolumn{1}{c}{coverage} & \multicolumn{1}{c}{bias}
					& \multicolumn{1}{c|}{s.d.}
					& \multicolumn{1}{c}{coverage} &
					\multicolumn{1}{c}{bias}
					&  \multicolumn{1}{c}{s.d.} \\ \hline
					\multicolumn{2}{l|}{\textbf{Low treatment effect}} &  & & & & &\\
					$\hat\tau_f$ & 0.066 & 94.3\% & 0.005 & 0.124 & 96.2\% & 0.001  &  0.053  &  95.1\%  & 0.001  &  0.026  \\
					$\hat\tau_f(c_{0.2})$ & 0.051 & 93.2  & -0.002 & 0.109 & 94.4 & 0.001  &  0.046  &  95.2  &  0.002  &  0.021  \\
					$\widehat\Gamma_f$  & 0.053 & 95.3 & 0.001 & 0.106 & 95.1 &  0.001  &   0.045  &  94.8  &  -0.001  &  0.024  \\
					$\widehat\Delta_{0.2}(f,g)$ & -0.022 & 94.0 & 0.006 & 0.122 & 95.4 & 0.002  &  0.051  &  96.0  & 0.000  &  0.026  \\
					$\widehat\Delta_{0.2}(f,h)$ & -0.014 & 93.9 & -0.001 & 0.131 & 94.9 & -0.000  &  0.060  &  95.3  & -0.000  &  0.030 \\\hline
					\multicolumn{2}{l|}{\textbf{High treatment effect}} &  & & & & &\\
					$\hat\tau_f$ &  0.430  &  94.7\%  &  -0.000  &  0.163  &  95.7\%  & 0.000  &  0.064  &  94.4\%  & -0.000  &  0.031  \\
					$\hat\tau_f(c_{0.2})$ &  0.356  &  94.7  & 0.004  &  0.159  &  95.7  & 0.002  &  0.072  &  95.8  &  0.000  &  0.035  \\
					$\widehat\Gamma_f$  &  0.363  &  94.3  &  -0.005  &  0.130  &  94.9  & 0.003  &   0.058  &  95.7  &  0.000  &  0.029  \\
					$\widehat\Delta_{0.2}(f,g)$ &  -0.000  &  96.9  & 0.008  &  0.151  &  97.9  & -0.002  &  0.073  &  98.0  & -0.000  &  0.026  \\
					$\widehat\Delta_{0.2}(f,h)$ &  0.000  &  94.7  & -0.004  &  0.140  &  97.7  & -0.001  &  0.065  &  96.6  & 0.000  &  0.033 \\\hline\hline
			\end{tabular}}
			\caption{The Results of the Simulation Study for Fixed
				Individualized Treatment Rules (ITRs).  The table presents the
				bias and standard deviation of each estimator as well as the
				coverage of its 95\% confidence intervals under the ``Low
				treatment effect'' and ``High treatment effect'' scenarios.
				The first three estimators shown here are for BART $f$:
				Population Average Prescription effect (PAPE; $\hat\tau_f$),
				PAPE with a budget constraint of 20\% treatment proportion
				($\hat\tau_f(c_{0.2})$), Area Under the Prescriptive Effect
				Curve (AUPEC; $\widehat{\Gamma}_f$).  We also present the
				results for the difference in the PAPE between BART and Causal
				Forest $g$ ($\widehat\Delta_{0.2}(f,g)$) and between BART and
				LASSO $h$ ($\widehat\Delta_{0.2}(f,g)$) under the budget
				constraint.} \label{tb:simulation}
		\end{table}
		
		We first present the results for fixed ITRs followed by those for
		estimated ITRs.  Table~\ref{tb:simulation} presents the bias and
		standard deviation of each estimator for fixed ITRs as well as the
		coverage probability of its 95\% confidence intervals based on 1,000
		Monte Carlo trials.  The results are shown separately for the high and
		low treatment effects scenarios.  We estimate the PAPE $\tau_f$ for
		BART without a budget constraint as well as the PAPE with a budget
		constraint of 20\% as the maximal proportion of treated units,
		$\tau_f(c_{0.2})$.  In addition, we estimate the AUPEC $\Gamma_f$ and
		compute the difference in the PAPE or PAPD between BART and Causal
		Forest ($\Delta(f,g)$), and between BART and LASSO ($\Delta(f, h)$).
		
		Under both scenarios and across sample sizes, the bias of our
		estimator is small.  Moreover, the coverage rate of 95\% confidence
		intervals is close to their nominal rate even when the sample size is
		small.  Although we can only bound the variance when estimating the
		PAPD between two ITRs (i.e., $\Delta_{0.2}(f, g)$ and
		$\Delta_{0.2}(f, h)$), the coverage stay close to 95\%, implying that
		the bound for covariance has little effect on the variance estimation.
		
		\begin{table}[t!]
			\centering\setlength{\tabcolsep}{2pt}
			\small \spacingset{1}
			\resizebox{1.1\textwidth}{!}{
				\begin{tabular}{l|....|....|....}
					\hline
					&  \multicolumn{4}{c|}{$\bm{n=100}$} & \multicolumn{4}{c|}{$\bm{n=500}$} & \multicolumn{4}{c}{$\bm{n=2000}$}\\
					Estimator & \multicolumn{1}{c}{truth} & \multicolumn{1}{c}{coverage} & \multicolumn{1}{c}{bias}
					& \multicolumn{1}{c|}{s.d.}
					& \multicolumn{1}{c}{truth} & \multicolumn{1}{c}{coverage} & \multicolumn{1}{c}{bias}
					& \multicolumn{1}{c|}{s.d.}
					& \multicolumn{1}{c}{truth} & \multicolumn{1}{c}{coverage} &
					\multicolumn{1}{c}{bias} &  \multicolumn{1}{c}{s.d.} \\ \hline
					\textbf{Low Effect}& & & & & & & & & & & & \\
					$\hat{\lambda}_F$ & 0.073 & 96.4\% & 0.001 & 0.216 & 0.095 & 96.7\% & 0.002 & 0.100 & 0.112 & 97.2\% & 0.002 & 0.046 \\
					$\hat\tau_F$ & 0.021 & 94.6 & -0.002 & 0.130 & 0.030 & 95.5 & -0.002 & 0.052 & 0.032 & 94.4 & -0.000 & 0.027 \\
					$\hat\tau_{F}(c_{0.2})$ & 0.023 & 95.4 & -0.003 & 0.120 & 0.034 & 95.4 & -0.002 & 0.057 & 0.043 & 96.8 & 0.001 & 0.029 \\
					$\widehat\Gamma_F$  & 0.009 & 98.2 & 0.002 & 0.117 & 0.029 & 96.8 & -0.001 & 0.048 & 0.039 & 95.9 & 0.001 & 0.001 \\
					\textbf{High Effect} & & & & & & & & & & & & \\
					$\hat{\lambda}_H$& 0.867 & 96.9\% & -0.007 & 0.261 & 0.875 & 96.5\% & -0.003 & 0.125 & 0.875 & 97.3\% & 0.001 & 0.062 \\
					$\hat\tau_F$ & 0.338 & 93.6 & -0.000 & 0.171 & 0.358 & 93.0 & 0.000 & 0.093 & 0.391 & 95.3 & 0.001 & 0.041 \\
					$\hat\tau_{F}(c_{0.2})$ & 0.341 & 94.8 & -0.002 & 0.170 & 0.356 & 96.2 & -0.005 & 0.075 & 0.356 & 95.8 & 0.001 & 0.037 \\
					$\widehat\Gamma_F$  & 0.344 & 98.5 & 0.001 & 0.126 & 0.362 & 98.9 & 0.005 & 0.053 & 0.363 & 99.0 & 0.001 & 0.026 \\\hline\hline
				\end{tabular}
			}
			\caption{The Results of the Simulation Study for Cross-validated ITR.
				The table presents the true value (truth) of each quantity along
				with the bias and standard deviation of each estimator as well as
				the coverage of its 95\% confidence intervals under the ``Low
				treatment effect'' and ``High treatment effect'' scenarios.  All of
				the results shown here are for LASSO.} \label{tb:simulationcv}
		\end{table}
		
		For estimated ITRs, Table~\ref{tb:simulationcv} presents the results
		of LASSO under cross-validation.  For BART and Causal Forest,
		obtaining an accurate Monte Carlo estimate of the true causal
		parameter values under cross-validation takes a prohibitively large
		amount of time.  While the out-of-bag estimates of such true values
		can be computed, they have been shown to create bias under certain
		scenarios \citep{janitza2018overestimation}.  These true values are
		generally greater for a larger sample size because LASSO performs
		better with more data.
		
		The proposed cross-validated estimators are approximately unbiased
		even for $n=100$. The coverage is generally around or above the
		nominal $95\%$ value, reflecting the conservative estimate of the
		variance.  For the PAPE without and with budget constraint, i.e.,
		$\hat{\tau}_F$ and $\hat{\tau}_{F}(c_{0.2})$, the coverage is close to
		the nominal value.  This indicates that the bias of the proposed
		conservative variance estimator is relatively small even though the
		number of folds for cross-validation is only $K=5$.  The performance
		of the proposed methodology is good even when the sample size is as
		small as $n=100$.  When the sample size is $n=500$, the standard
		deviation of the cross-validated estimator is roughly half of the
		corresponding $n=100$ fixed ITR estimator (this is a good comparison
		because each fold has 100 observations).  This confirms the
		theoretical efficiency gain that results from cross-validation.
		
		\section{An Empirical Application}
		\label{sec:empirical}
		
		We apply the proposed methodology to the data from the Tennessee's
		Student/Teacher Achievement Ratio (STAR) project, which was a
		longitudinal study experimentally evaluating the impacts of class size
		in early education on various outcomes \citep{most:95}.  Another
		application based on a canvassing experiment is shown in
		Appendix~A.8.
		
		\subsection{Data and Setup}
		\label{subsec:STARdata}
		
		The STAR project randomly assigned over 7,000 students across 79
		schools to three different groups: small class,
		regular class,
		and regular class with a full-time teacher's aid.  The experiment
		began when students entered kindergarden and continued through third
		grade.  To create a binary treatment, we focus on the first two
		groups: small class and regular class without an aid.  The treatment
		effect heterogeneity is important because reducing class size is
		costly, requiring additional teachers and classrooms.  Policy makers
		who face a budget constraint may be interested in finding out which
		groups of students benefit most from a small class size so that the
		priority can be given to those students.
		
		We follow the analysis strategies of the previous studies
		\citep[e.g.,][]{ding:lehr:11,mcke:sims:rivk:15} that estimated the
		heterogeneous effects of small classes on educational attainment.
		These authors adjust for school-level and student-level
		characteristics, but do not consider within-classroom interactions.
		Unfortunately, addressing this limitation is beyond the scope of this
		paper.  We use a total of 10 pre-treatment covariates $\bX_i$ that
		include four demographic characteristics of students (gender, race,
		birth month, birth year) and six school characteristics (urban/rural,
		enrollment size, grade range, number of students on free lunch, number
		of students on school buses, and percentage of white students).  Our
		treatment variable is the class size to which they were assigned at
		kindergarten: small class $T_i=1$ and regular class without an aid
		$T_i = 0$.  For the outcome variables $Y_i$, we use three standardized
		test scores measured at third grade: math, reading, and writing SAT
		scores.
		
		The resulting data set has a total of 1,911 observations.  The
		estimated average treatment effects (ATE) based on the entire data set
		are $6.78$ (s.e. = $1.71$), $5.78$ (s.e. = $1.80$), and $3.65$ (s.e. =
		$1.63$), for the reading, math, and writing scores, respectively.  For
		the fixed test data, the estimated ATEs are similar; $5.10$ (s.e. =
		$3.07$), $2.78$ (s.e.= $3.15$), and $1.48$ (s.e. = $2.96$).
		
		We evaluate the performance of ITRs using two settings considered
		above.  First, we randomly split the data into the training data
		(70\%) and test data (30\%).  We estimate an ITR from the training
		data and then evaluate it as a fixed ITR using the test data.  This
		follows the setup considered in
		Sections~\ref{sec:evalmetrics}~and~\ref{sec:est-inf}.  Second, we
		consider the evaluation of estimated ITRs based on the same
		experimental data.  We utilize Algorithm~\ref{alg:cross_validation}
		with 5-fold cross-validation (i.e., $K=5$).  For both settings, we use
		the same three machine learning algorithms. For Causal Forest, we set
		{\tt tune.parameters = TRUE}. For BART, tuning was done on the number
		of trees. For LASSO, we tuned the regularization parameter while
		including all interaction terms between covariates and the treatment
		variable.  All tuning was done through the 5-fold cross validation
		procedure on the training set using the PAPE as the evaluation metric.
		We then create an ITR as $\mathbf{1}\{\hat\tau(\bx)>0\}$ where
		$\hat\tau(\bx)$ is the estimated CATE obtained from each fitted model.
		As mentioned in Section~\ref{sec:est-inf}, we center the outcome
		variable $Y$ in evaluating the metrics to minimize the variance of the
		estimators.
		
		\subsection{Results}
		\label{subsec:results}
		
		\begin{table}[t!]
			\centering \spacingset{1} \setlength{\tabcolsep}{4.25pt}
			\begin{tabular}{l..c|..c|..c}
				\hline
				& \multicolumn{3}{c|}{\textbf{BART}} & \multicolumn{3}{c|}{\textbf{Causal Forest}}
				& \multicolumn{3}{c}{\textbf{LASSO}}\\
				& \multicolumn{1}{c}{est.}  & \multicolumn{1}{c}{s.e.} & \multicolumn{1}{c|}{treated}
				& \multicolumn{1}{c}{est.}  & \multicolumn{1}{c}{s.e.} & \multicolumn{1}{c|}{treated}
				& \multicolumn{1}{c}{est.}  & \multicolumn{1}{c}{s.e.} &
				\multicolumn{1}{c}{treated}\\\hline
				\multicolumn{3}{l}{\textbf{Fixed ITR}} & & & & & \\
				\multicolumn{3}{l}{\textit{No budget constraint}} & & & & & & \\
				\quad Reading &      0  &  0    & 100\%   & -0.38 & 1.14 & 84.3\%  &  -0.41 &  1.10 &  84.4\% \\
				\quad Math      & 0.52 & 1.09  & 86.7  &  0.09 & 1.18  & 80.3 & 1.73 & 1.25 &  78.7 \\
				\quad Writing  & -0.32  &   0.72 &   92.7  & -0.70 & 1.18 &  78.0 &  -0.30 & 1.26 & 80.0\\\hdashline
				\multicolumn{3}{l}{\textit{Budget constraint}} & & & & & & \\
				\quad Reading & -0.89 & 1.30 & 20 & 0.66 & 1.23  & 20 & -1.17 & 1.18 &  20 \\
				\quad Math     & 0.70 & 1.25  & 20  & 2.57   & 1.29 & 20 & 1.25 & 1.32 &  20 \\
				\quad Writing  & 2.60 & 1.17 & 20 & 2.98 & 1.18 & 20 & 0.28 & 1.19 &  20\\\hline
				\multicolumn{3}{l}{\textbf{Estimated ITR}}  & & & & & & \\
				\multicolumn{3}{l}{\textit{No budget constraint}} & & & & & & \\
				\quad Reading  &  0.19  &  0.37    & 99.3\%   & 0.31 & 0.77 & 86.6\%  &  0.32 &  0.53 &  87.6\% \\
				\quad Math     & 0.92 & 0.75  & 84.7  &  2.29 & 0.80 & 79.1 & 1.52 & 1.60 &  75.2 \\
				\quad Writing  & 1.12       &    0.86    &   88.0  & 1.43 & 0.71 &  67.4 &  0.05 & 1.37 & 74.8\\\hdashline
				\multicolumn{3}{l}{\textit{Budget constraint}} & & & & & & \\
				\quad Reading  &  1.55& 1.05 & 20 & 0.40 & 0.69 & 20 & -0.15 & 1.41 &  20 \\
				\quad Math     & 2.28 & 1.15  & 20  & 1.84   & 0.73 & 20 & 1.50 & 1.48 &  20 \\
				\quad Writing  & 2.31 & 0.66 & 20 & 1.90 & 0.64 & 20 & -0.47 & 1.34 &  20\\\hline
			\end{tabular}
			\caption{The Estimated Population Average Prescription Effect (PAPE)
				for BART, Causal Forest, and LASSO with and without a Budget
				Constraint.  We estimate the PAPE for fixed and estimated
				individualized treatment rules (ITRs).  The fixed ITRs are based
				on the training (70\%) and test data (30\%), whereas the estimated
				ITRs are based on 5 fold cross-validation.  In addition, the
				average treatment effect estimates using the entire dataset. For
				each of the three outcomes, the point estimate, the standard
				error, and the average proportion treated are shown.  The budget
				constraint considered here implies that the maximum proportion
				treated is 20\%.} \label{tb:comparison_cv}
		\end{table}
		
		The upper panel of Table~\ref{tb:comparison_cv} presents the estimated
		PAPEs, their standard errors, and the proportion treated for fixed
		ITRs.  We find that without a budget constraint, none of the machine
		learning algorithms significantly improves upon the random treatment
		rule.  With a budget constraint of 20\%, however, the ITRs based on
		Causal Forest appear to improve upon the random treatment rule at
		least for the math and writing scores.  In contrast, the ITR based on
		BART only performs well for the writing score whereas the performance
		of LASSO is not distinguishable from that of random treatment rule
		across all scores.  The results are largely similar for the estimated
		ITRs, as shown in the lower panel of Table~\ref{tb:comparison_cv}.
		The only difference is that BART performs slightly better while the
		performance of Causal Forest is slightly better for the fixed ITRs and
		worse for the estimated ITRs.
		
		In the case of BART and Causal Forest, the standard errors for
		estimated ITRs are generally smaller than those for fixed ITRs,
		reflecting the efficiency gain due to cross-validation.  For LASSO,
		however, the standard errors for fixed ITRs are often smaller than
		those for estimated ITRs.  This is due to the fact that the ITRs
		estimated by LASSO are significantly variable across different folds
		under cross-validation.  This variability results in the poor
		performance of LASSO in this application.  In contrast, Causal Forest
		is most stable, generally leading to the best performance and the
		smallest standard errors.  Lastly, we note that when compared to the
		estimated ATEs, some estimated PAPEs are of substantial size.

		\begin{table}[t!]
			\centering\spacingset{1}
			\begin{tabular}{l.r|.r|.r}
				\hline   & \multicolumn{4}{c|}{\textbf{Causal
						Forest}} & \multicolumn{2}{c}{\textbf{BART}} \\
				& \multicolumn{2}{c}{vs. \textbf{BART}}
				& \multicolumn{2}{c|}{ vs. \textbf{LASSO}}
				& \multicolumn{2}{c}{vs. \textbf{LASSO}} \\ \hline
				& \multicolumn{1}{c}{est.}  & \multicolumn{1}{c|}{$95\%$ CI}
				& \multicolumn{1}{c}{est.} & \multicolumn{1}{c|}{$95\%$ CI}
				& \multicolumn{1}{c}{est.}  & \multicolumn{1}{c}{$95\%$
					CI} \\
				\multicolumn{2}{l}{\bf Fixed ITR} & & & & \\
				\quad Math      & 1.55   & [$-$0.35, 3.45] & 1.83 & [$-$0.50, 4.16] &
				0.28 & [$-$2.39, 2.95]\\
				\quad Reading & 1.86 & [$-$0.79, 4.51]  & 1.31 & [$-$1.49, 4.11] &
				-0.55 & [$-$4.02, 2.92] \\
				\quad Writing  & 0.38   & [$-$1.66, 2.42]   & 2.69 & [$-$0.27, 5.65] &
				2.32
				& [$-$0.53, 5.15] \\\hline
				\multicolumn{2}{l}{\bf Estimated ITR} & & &  &\\
				\quad Reading      & -1.15   & [$-$3.99, 1.69] & 0.55 & [$-$1.05, 2.15] &
				1.70 & [$-$0.90, 4.30]\\
				\quad Math & -0.43 & [$-$2.57, 3.43]  & 0.34 & [$-$1.32, 2.00] &
				0.77 & [$-$1.99, 3.53] \\
				\quad Writing  & -0.41   & [$-$1.63, 0.80]   & 2.37 & [0.76, 3.98] &
				2.79 & [1.32, 4.26] \\\hline
			\end{tabular}
			\caption{The Estimated Population Average Prescriptive Effect
				Difference (PAPD) between Causal Forest, BART, and LASSO
				under a Budget Constraint of 20\%. The point estimates and
				95\% confidence intervals are shown.} \label{tb:PAPD_cv}
		\end{table}
		
		Table~\ref{tb:PAPD_cv} directly compares these three ITRs based on
		Causal Forest, BART, and LASSO by estimating the Population Average
		Prescriptive Effect Difference (PAPD) under the same budget constraint
		(i.e., 20\%) as above.  Causal Forest outperforms BART and LASSO in
		essentially all cases though the difference is not statistically
		significant. Under the cross-validation setting, Causal Forest and
		BART are statistically significantly more effective than LASSO in
		identifying students with grater treatment effects on their writing
		scores.
		
		\begin{figure}[t!]
			\spacingset{1}
			\includegraphics[width=\linewidth]{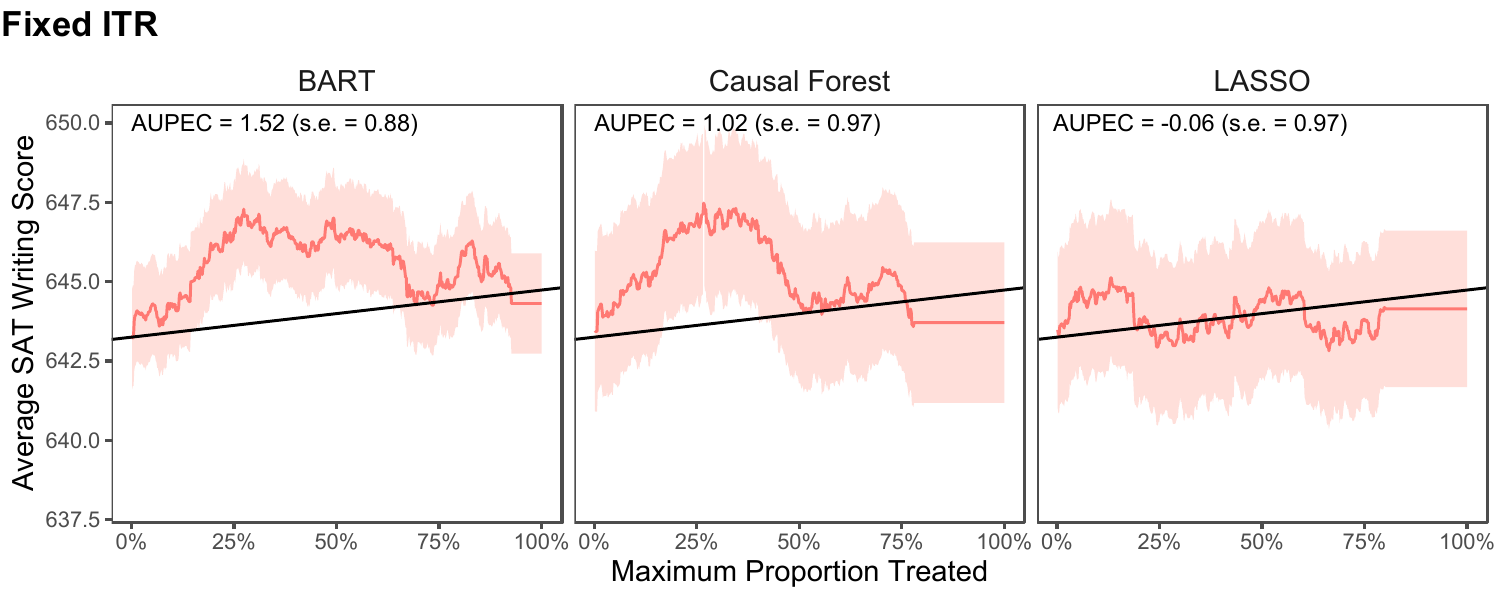}
			\includegraphics[width=\linewidth]{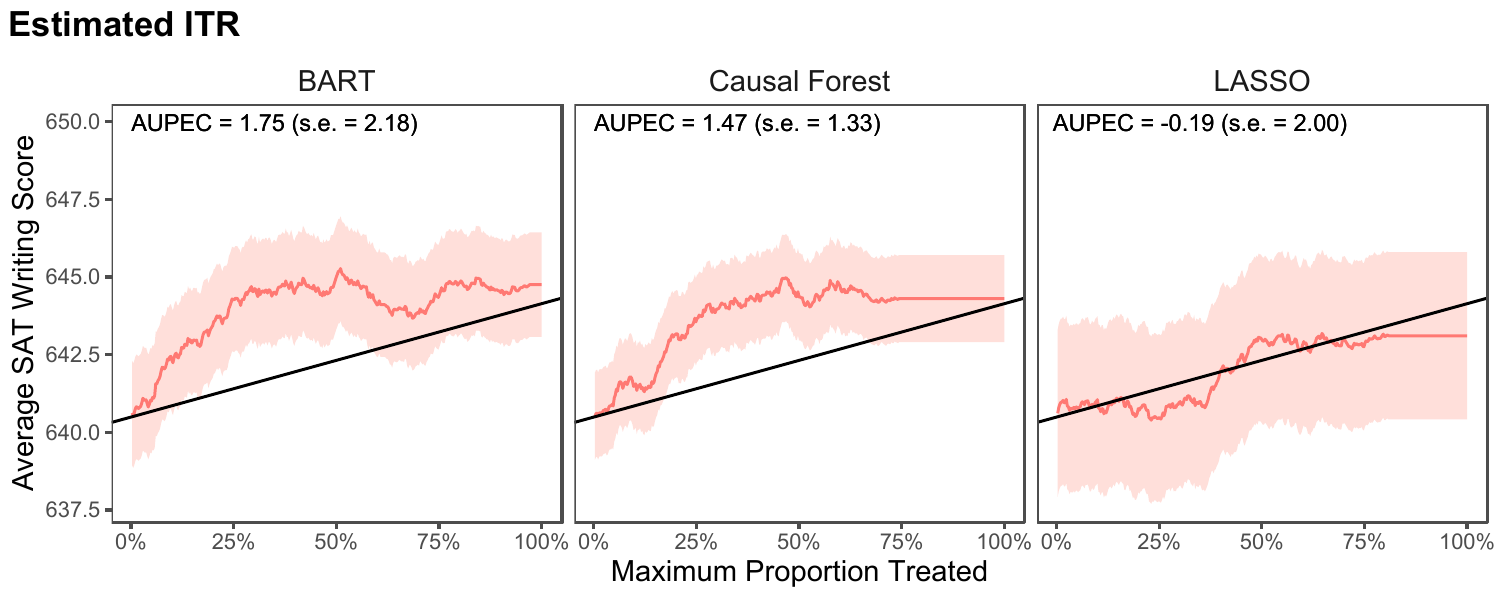}
			\caption{Estimated Area Under the Prescriptive Effect Curve
				(AUPEC).  The results are presented for the fixed (upper panel)
				and estimated (bottom panel) individualized treatment rule (ITR)
				settings.  A solid red line in each plot represents the
				Population Average Prescriptive Effect for SAT writing scores
				across a range of budget constraint (horizontal axis) with the
				pointwise 95\% confidence intervals. The area between this line
				and the black line representing random treatment is the AUPEC.
				The results are presented for the individualized treatment rules
				based on BART (left column), Causal Forest (middle column), and
				LASSO (right column), whereas each row presents the results for
				a different outcome.  } \label{fg:AUPEC}
		\end{figure}
		
		Finally, Figure~\ref{fg:AUPEC} presents the estimated PAPE for the
		writing score across a range of budget constraints with the pointwise
		confidence intervals based on the variance formulas, for the fixed and
		estimated ITR cases. The difference between the solid red and black
		lines is the estimated PAPE, while the area between the red line and
		the black line corresponds to the Area Under the Prescriptive Effect
		Curve (AUPEC).  In each plot, the horizontal axis represents the
		budget constraint as the maximum proportion treated. In both the fixed
		and estimated ITR settings, BART and Causal Forest identify students
		who benefit positively from small class sizes when the maximum
		proportion treated is relatively small.  In contrast, LASSO has a
		difficult time in finding these individuals.  Additionally, we find
		that the standard error of the estimated AUPEC is greater for
		estimated ITRs than for fixed ITRs even though in the case of BART and
		Causal Forest, the opposite pattern is found for the PAPE with a given
		budget constraint.  This finding is due to the high variance of the
		estimated AUPEC for different folds of cross-validation.
		
		As the budget constraint is relaxed, the ITRs based on BART and Causal
		Forest yields the population average value similar to the one under the random treatment
		rule. This results in the inverted V-shape AUPEC curves observed in
		the left and middle plots of the figure. These inverted V-shape curves
		illustrate two separate phenomena.  First, the students with the
		highest predicted CATE under BART and Causal Forests do have a higher
		treatment effect than the average, yielding an uplift in the PAPE
		curve compared to the random treatment rule. This shows that BART and
		Causal Forests are able to capture some causal heterogeneity that
		exists in the STAR data.  Indeed, both BART and the Causal Forest
		estimated the CATE to be higher for non-white students and those who
		attend schools with a high percentage of students receiving free
		lunch.  According to the variable importance statistic, these two
		covariates play an essential role in explaining causal heterogeneity
		\citep[see e.g.,][similar
		findings]{finn1990answers,jackson2013estimating,nye2000disadvantaged}.
		
		However, as we further relax the budget constraint, the methods start
		treating students who are predicted to have a smaller (yet still
		positive) value of the CATE. These students tend to benefit less than
		the ATE, resulting in a smaller value of the PAPE as the budget
		increases.  Eventually, both BART and Causal Forest start
		``over-treating'' students who are estimated to have a small positive
		CATE but actually do not benefit from a small class.  This results in
		the overall insignificant PAPE when no budget constraint is imposed.

		\section{Concluding Remarks}
		\label{sec:conclude}
		
		As the application of individualized treatment rules (ITRs) becomes
		more widespread in a variety of fields, a rigorous perfomance
		evaluation of ITRs becomes essential before policy makers deploy them
		in a target population.  We believe that the inferential approach
		proposed in this paper provides a robust and widely applicable tool to
		policy makers.  The proposed methodology also opens up opportunities
		to utilize the existing randomized controlled trial data for the
		efficient evaluation of ITRs as done in our empirical application. In
		addition, although we do not focus on the estimation of ITRs in this
		paper, the proposed evaluation metrics can be used to tune
		hyper-parameters when cross validating machine learning algorithms.
		In future research, we plan to consider the extensions of the proposed
		methodology to other settings, including non-binary treatments,
		dynamic treatments, and treatment allocations in the presence of
		interference between units.
		
		\spacingset{1.6}
		\pdfbookmark[1]{References}{References}
		\bibliography{sample,my,imai}

\clearpage
\appendix
\spacingset{1.15}

\setcounter{table}{0}
\renewcommand{\thetable}{A\arabic{table}}
\setcounter{figure}{0}
\renewcommand{\thefigure}{A\arabic{figure}}
\setcounter{equation}{0}
\renewcommand{\theequation}{A\arabic{equation}}
\setcounter{theorem}{0}
\renewcommand {\thetheorem} {A\arabic{theorem}}

\section{Supplementary Appendix for ``Experimental Evaluation of
  Individualized Treatment Rules''}

\subsection{Estimation and Inference for Fixed ITRs with No Budget Constraint}
\label{app:nobudget}

\subsubsection{The Population Average Value}
\label{app:PAV}

For a fixed ITR with no budget constraint, the following unbiased
estimator of the population average value (equation~\eqref{eq:PAV}),
based on the experimental data $\cZ$, is used in the literature
\citep[e.g.,][]{qian:murp:11},
\begin{equation}
\hat\lambda_f(\cZ) \ = \ \frac{1}{n_1} \sum_{i=1}^n Y_i T_i f(\bX_i) +
\frac{1}{n_0} \sum_{i=1}^n Y_i(1-T_i)(1-f(\bX_i)). \label{eq:PAVest}
\end{equation}
Under Neyman's repeated sampling framework, it is straightforward to
derive the unbiasedness and variance of this estimator where the
uncertainty is based solely on the random sampling of units and the
randomization of treatment alone.  The results are summarized as the
following theorem.

\begin{theorem} {\sc (Unbiasedness and Variance of the Population Average Value
  Estimator)} \label{thm:PAVest} \spacingset{1} Under
  Assumptions~\ref{asm:SUTVA},\ref{asm:randomsample},~and~\ref{asm:comrand},
  the expectation and variance of the population average value estimator defined in
  equation~\eqref{eq:PAVest} are given by,
  \begin{eqnarray*}
    \E\{\hat\lambda_f(\cZ)\} -  \lambda_f& = & 0, \quad
                                               \V\{\hat\lambda_f(\cZ)\} \ = \  \frac{\E(S_{f1}^2)}{n_1} +
                                               \frac{\E(S_{f0}^2)}{n_0}
  \end{eqnarray*}
  where
  $S_{ft}^2 = \sum_{i=1}^n (Y_{fi}(t) - \overline{Y_f(t)})^2/(n-1)$
  with $Y_{fi}(t) = \mathbf{1}\{f(\bX_i)=1\}Y_i(t)$ and
  $\overline{Y_f(t)} = \sum_{i=1}^n Y_{fi}(t)/n$ for $t=\{0,1\}$.
\end{theorem}
Proof is straightforward and hence omitted.

\subsubsection{The Population Average Prescriptive Effect (PAPE)}
\label{app:PAPEest}

To estimate the PAPE with no budget constraint
(equation~\eqref{eq:PAPE}), we propose the following estimator,
\begin{eqnarray}
  \hat\tau_f(\cZ) & = & \frac{n}{n-1}\l[\frac{1}{n_1} \sum_{i=1}^n Y_i T_i f(\bX_i) +
                   \frac{1}{n_0} \sum_{i=1}^n Y_i(1-T_i)(1-f(\bX_i)) \r.
                   \nonumber \\
                   & & \hspace{1.25in} \l. -
                   \frac{\hat{p}_f}{n_1} \sum_{i=1}^n Y_iT_i -
                   \frac{1-\hat{p}_f}{n_0} \sum_{i=1}^n Y_i(1-T_i)\r]  \label{eq:PAPEest}
\end{eqnarray}
where $\hat{p}_f = \sum_{i=1}^n f(\bX_i)/n$ is a sample estimate of
$p_f$, and the term $n/(n-1)$ is due to the finite sample
degree-of-freedom correction resulting from the need to estimate
$p_f$.  The following theorem proves the unbiasedness of this
estimator and derives its exact variance.

\begin{theorem}[Unbiasedness and Exact Variance of the PAPE
  Estimator] \label{thm:PAPEest} \spacingset{1} Under
  Assumptions~\ref{asm:SUTVA},~\ref{asm:randomsample},~and~\ref{asm:comrand},
  the expectation and variance of the PAPE estimator defined
  equation~\eqref{eq:PAPEest} are given by,
  \begin{eqnarray*}
    \E\{\hat\tau_f(\cZ)\} & = & \tau_f \\
    \V\{\hat\tau_f(\cZ)\} & = &  \frac{n^2}{(n-1)^2}\l[\frac{\E(\widetilde{S}_{f1}^2)}{n_1} +
                                \frac{\E(\widetilde{S}_{f0}^2)}{n_0} + \frac{1}{n^2} \l\{\tau_f^2 -n p_f(1-p_f)
                                \tau^2 + 2(n-1)(2p_f-1) \tau_f\tau \r\} \r]
  \end{eqnarray*}
  where
  $\widetilde{S}_{ft}^2 = \sum_{i=1}^n (\widetilde{Y}_{fi}(t) -
  \overline{\widetilde{Y}_f(t)})^2/(n-1)$ with
  $\widetilde{Y}_{fi}(t) = (f(\bX_i) - \hat{p}_f)Y_i(t)$, and
  $\overline{\widetilde{Y}_f(t)} = \sum_{i=1}^n
  \widetilde{Y}_{fi}(t)/n$ for $t=\{0,1\}$.
\end{theorem}
Note that $\E(\widetilde{S}_{ft}^2)$ does not equal
$\V(\widetilde{Y}_{fi}(t))$ because the proportion of treated units
$p_f$ is estimated. The additional term in the variance accounts for
this estimation uncertainty of $p_f$. The variance of the proposed
estimator can be consistently estimated by replacing the unknown
terms, i.e., $p_f$, $\tau_f$, $\tau$, $\E(\widetilde{S}^2_{ft})$, with
their unbiased estimates, i.e., $\hat{p}_f$, $\hat\tau_f$,
\begin{equation*}
  \hat\tau \ = \ \frac{1}{n_1}\sum_{i=1}^n T_i Y_i -
  \frac{1}{n_0}\sum_{i=1}^n (1-T_i)Y_i, \quad {\rm and} \quad \widehat{\E(\widetilde{S}^2_{ft})} \
  = \ \frac{1}{n_t - 1}\sum_{i=1}^n \mathbf{1}\{T_i = t\}(\widetilde{Y}_{fi}
  - \overline{\widetilde{Y}_{ft}})^2,
\end{equation*}
where $\widetilde{Y}_{fi} = (f(\bX_i) - \hat{p}_f) Y_i$ and
$\overline{\widetilde{Y}_{ft}} = \sum_{i=1}^n \mathbf{1}\{T_i =
t\}\widetilde{Y}_{fi}/n_t$.

To prove Theorem~\ref{thm:PAPEest}, we first consider the sample average
prescription effect (SAPE),
\begin{equation}
  \tau_f^s \ = \ \frac{1}{n} \sum_{i=1}^n \l\{Y_i(f(\bX_i)) - \hat{p}_f
  Y_i(1) - (1-\hat{p}_f) Y_i(0)\r\}. \label{eq:SAPE}
\end{equation}
and its unbiased estimator,
\begin{equation}
  \hat\tau_f^s \ = \ \frac{1}{n_1} \sum_{i=1}^n Y_i T_i f(\bX_i) +
                   \frac{1}{n_0} \sum_{i=1}^n Y_i(1-T_i)(1-f(\bX_i)) -
                   \frac{\hat{p}_f}{n_1} \sum_{i=1}^n Y_iT_i -
                   \frac{1-\hat{p}_f}{n_0} \sum_{i=1}^n Y_i(1-T_i) \label{eq:SAPEest}
\end{equation}
This estimator differs from the estimator of the PAPE by a small
factor, i.e., $\hat\tau_f^s = (n-1)/n \hat\tau_f$.  The following
lemma derives the expectation and variance in the Neyman framework.
Thus, it only requires the randomization-based finite sample inference
and does not need Assumption~\ref{asm:randomsample}.
\begin{lemma}[Unbiasedness and Exact Variance of the Estimator for the
  SAPE] \label{lemma:SAPE} \spacingset{1} Under
  Assumptions~\ref{asm:SUTVA},~\ref{asm:randomsample},~and~\ref{asm:comrand},
  the expectation and variance of the estimator of the PAPE given in
  equation~\eqref{eq:SAPEest} for estimating the SAPE defined in
  equation~\eqref{eq:SAPE} are given by,
  \begin{eqnarray*}
    \E(\hat\tau_f^s \mid \cO_n)
    & =  & \tau_f^s \\
    \V(\hat\tau_f^s \mid \cO_n) &
     =  &  \frac{1}{n} \l( \frac{n_0}{n_1} \widetilde{S}_{f1}^2 + \frac{n_1}{n_0} \widetilde{S}_{f0}^2
        + 2 \widetilde{S}_{f01}\r)
  \end{eqnarray*}
  where  $\cO_n= \{Y_i(1), Y_i(0),
  \bX_i\}_{i=1}^n$ and
  \begin{eqnarray*}
                   \widetilde{S}_{f01} \ = \
                   \frac{1}{n-1} \sum_{i=1}^n
                    (\widetilde{Y}_{fi}(0) -  \overline{\widetilde{Y}_{fi}(0)}) (\widetilde{Y}_{fi}(1) - \overline{\widetilde{Y}_{fi}(1)}).
  \end{eqnarray*}
\end{lemma}
\begin{proof}
  We begin by computing the expectation with respect to the
  experimental treatment assignment, i.e., $T_i$,
\begin{eqnarray*}
\E(\hat\tau_f^s \mid \cO_n)&= & \E\left\{\frac{1}{n_1}\sum_{i=1}^n
                           f(\bX_i)T_i Y_i(1) +\frac{1}{n_0}\sum_{i=1}^n
                            (1-f(\bX_i))(1-T_i) Y_i(0) \r.\\
  & & \hspace{1in} \l. -\frac{\hat{p}_f}{n_1}\sum_{i=1}^n
    T_iY_i(1)-\frac{1-\hat{p}_f}{n_0}\sum_{i=1}^n (1-T_i) Y_i(0) \
      \Biggl | \
      \cO_n\right\} \\
  & = & \frac{1}{n} \sum_{i=1}^n
        Y_i(1)f(\bX_i)+\frac{1}{n}\sum_{i=1}^nY_i(0)(1-f(\bX_i))-\frac{\hat{p}_f}{n}\sum_{i=1}^n
        Y_i(1)-\frac{1- \hat{p}_f}{n}\sum_{i=1}^n Y_i(0) \\
  & = & \tau_f^s
\end{eqnarray*}
To derive the variance, we first rewrite the proposed estimator as,
\begin{eqnarray*}
  \hat\tau_f^s & = & \tau_f^s + \sum_{i=1}^n D_i (f(\bX_i) - \hat{p}_f)
                   \l(\frac{Y_i(1)}{n_1} +
                   \frac{Y_i(0)}{n_0} \r)
\end{eqnarray*}
where $D_i = T_i - n_1/n$.  Thus, noting $\E(D_i)=0$, $\E(D_i^2) =
n_0n_1/n^2$, and $\E(D_iD_j) = -n_0n_1/\{n^2(n-1)\}$ for $i\ne j$,
after some algebra, we have,
\begin{eqnarray*}
  \V(\hat\tau_f^s \ \mid \cO_n)  & = &  \V(\hat\tau_f^s - \tau_f^s \mid \cO_n)
  \ = \ \E\l[\l\{ \sum_{i=1}^n D_i
                   \l(\frac{\widetilde{Y}_{fi}(1)}{n_1} +
        \frac{\widetilde{Y}_{fi}(0)}{n_0} \r) \r\}^2\ \Biggl | \ \cO_n \r] \\
 & = & \frac{1}{n} \l( \frac{n_0}{n_1} \widetilde{S}_{f1}^2 + \frac{n_1}{n_0} \widetilde{S}_{f0}^2
        + 2 \widetilde{S}_{f01}\r)
\end{eqnarray*}
\end{proof}

Now, we prove Theorem~\ref{thm:PAPEest}.  Using Lemma~\ref{lemma:SAPE}
and the law of iterated expectation, we have,
\begin{equation*}
  \E(\hat\tau_f ^s) \ = \ \E\l[\frac{1}{n} \sum_{i=1}^n \l\{Y_i(f(\bX_i)) - \hat{p}_f
  Y_i(1) - (1-\hat{p}_f) Y_i(0)\r\} \r].
\end{equation*}
We compute the following expectation for $t=0,1$,
\begin{eqnarray*}
  \E\l( \sum_{i=1}^n \hat{p}_f
  Y_i(t)\r) & = & \E\l( \sum_{i=1}^n \frac{\sum_{j=1}^n f(\bX_j)}{n}
                  Y_i(t)\r)
  \ = \ \frac{1}{n}\E\l\{ \sum_{i=1}^n f(\bX_i)  Y_i(t)  + \sum_{i=1}^n
        \sum_{j \ne i} f(\bX_j) Y_i(t)\r\} \\
  & = & \E\{ f(\bX_i) Y_i(t)\} + (n-1) p_f \E(Y_i(t)).
\end{eqnarray*}
Putting them together yields the following bias expression,
\begin{eqnarray*}
  \E(\hat\tau_f ^s) & = & \E\l[\{Y_i(f(\bX_i)) -
  \frac{1}{n}f(\bX_i)\tau_i - \frac{n-1}{n} p_f \tau_i - Y_i(0)]\r] \\
  & = & \tau_f - \frac{1}{n}\E\l[\{f(\bX_i)Y_i(1) - (1-f(\bX_i))Y_i(0)\} -
        \{p_fY_i(1) - (1-p_f)Y_i(0)\}\r] \\
  & = & \tau_f - \frac{1}{n} {\rm Cov}(f(\bX_i), \tau_i).
\end{eqnarray*}
where $\tau_i = Y_i(1) - Y_i(0)$.  We can further rewrite the bias as,
\begin{eqnarray}
   - \frac{1}{n} {\rm Cov}(f(\bX_i), \tau_i)
  & = & \frac{1}{n} p_f \{\E(\tau_i \mid f(\bX_i) = 1) -
        \tau\} \nonumber \\
  & = & \frac{1}{n} p_f (1-p_f) \{\E(\tau_i \mid f(\bX_i) =
        1) - \E(\tau_i \mid f(\bX_i) = 0)\} \nonumber\\
  & = & \frac{\tau_f}{n}. \label{eq:cov.tau}
\end{eqnarray}
where $\tau= \E(Y_i(1)-Y_i(0))$. This implies the estimator for the
PAPE is unbiased, i.e., $\E(\hat\tau_f) = \tau_f$.

To derive the variance, Lemma~\ref{lemma:SAPE}
implies,
\begin{equation}
  \V(\hat\tau_f) \ = \ \frac{n^2}{(n-1)^2}\l[\V\l(\frac{1}{n} \sum_{i=1}^n \{\widetilde{Y}_{fi}(1) - \widetilde{Y}_{fi}(0)\}\r) +
  \E\l\{\frac{1}{n}\l(\frac{n_0}{n_1} \widetilde{S}_{f1}^2 + \frac{n_1}{n_0}
  \widetilde{S}_{f0}^2 + 2 \widetilde{S}_{f01}\r)\r\}\r]. \label{eq:var.decomp}
\end{equation}
Applying Lemma~1 of \citet{nadeau2000inference} to the first term
within the square brackets yields,
\begin{equation}
  \V\l(\frac{1}{n} \sum_{i=1}^n \{\widetilde{Y}_{fi}(1) - \widetilde{Y}_{fi}(0)\}\r)
  \ = \ {\rm Cov}(\widetilde{Y}_{fi}(1) -\widetilde{Y}_{fi}(0), \widetilde{Y}_{fi}(1) -
  \widetilde{Y}_{fi}(0)) + \frac{1}{n}   \E(\widetilde{S}_{f1}^2  + \widetilde{S}_{f0}^2 - 2 \widetilde{S}_{f01}),  \label{eq:cov.decomp}
\end{equation}
where $i \ne j$.  Focusing on the covariance term, we have,
\begin{eqnarray*}
  & & {\rm Cov}(\widetilde{Y}_{fi}(1) -
  \widetilde{Y}_{fi}(0), \widetilde{Y}_{fi}(1) -
  \widetilde{Y}_{fi}(0)) \\
  & = & {\rm Cov}\l(\l\{f(\bX_i) - \frac{1}{n}\sum_{i^\prime=1}^n
        f(\bX_{i^\prime})\r\}\tau_i, \l\{f(\bX_j) -
        \frac{1}{n} \sum_{j^\prime=1}^n
        f(\bX_{j^\prime})\r\}\tau_j\r) \\
  & = & -2{\rm Cov}\l(\frac{n-1}{n}f(\bX_i)\tau_i,
        \frac{1}{n}f(\bX_i)\tau_j\r) +  \sum_{i^\prime \ne i, j}{\rm Cov}\l(
    \frac{1}{n} f(\bX_{i^\prime})\tau_i, \frac{1}{n}
      f(\bX_{i^\prime})\tau_j\r) \\
  & & +  2 \sum_{i^\prime \ne i, j}{\rm Cov}\l(
  \frac{1}{n} f(\bX_{j})\tau_i, \frac{1}{n}
  f(\bX_{i^\prime})\tau_j\r) +{\rm Cov}\l(
  \frac{1}{n} f(\bX_{j})\tau_i, \frac{1}{n}
  f(\bX_{i})\tau_j\r) \\
  & = & -\frac{2(n-1) \tau}{n^2}{\rm Cov}\l(f(\bX_i),
        f(\bX_i)\tau_i \r) + \frac{(n-2)\tau^2}{n^2} \V(f(\bX_i)) \\
  & & + \frac{2(n-2)\tau}{n^2}
      p_f {\rm Cov}\l(f(\bX_i), \tau_i\r) + \frac{1}{n^2} \l\{{\rm Cov}^2\l(f(\bX_i), \tau_i\r)+2p_f\tau {\rm Cov}\l(f(\bX_i), \tau_i\r) \r\}\\
  & = & \frac{1}{n^2} {\rm Cov}^2\l(f(\bX_i), \tau_i\r) + \frac{(n-2)\tau^2}{n^2} p_f(1-p_f) + \frac{2(n-1)\tau}{n^2} {\rm Cov}\l(f(\bX_i),
    (p_f-f(\bX_i))\tau_i \r) \\
  & = & \frac{1}{n^2} \l[\tau_f^2 + (n-2) p_f(1-p_f)
        \tau^2 + 2(n-1) \tau \l\{p_f
      \tau_f - (1-p_f)\E(f(\bX_i)\tau_i)\r\}\r] \\
 & = & \frac{1}{n^2} \l[\tau_f^2 + (n-2) p_f(1-p_f)
        \tau^2+ 2(n-1) \tau \l\{\tau_f(2p_f-1)
       - (1-p_f)p_f\tau\r\}\r] \\
 & = & \frac{1}{n^2} \l\{\tau_f^2 -n p_f(1-p_f)
        \tau^2 + 2(n-1)(2p_f-1) \tau_f\tau \r\},
\end{eqnarray*}
where the third equality follows from the formula for the covariance
of products of two random variables \citep{bohr:gold:69}.  Finally,
combining this result with
equations~\eqref{eq:var.decomp}~and~\eqref{eq:cov.decomp} yields,
\begin{eqnarray*}
   \V(\hat\tau_f)
  &  = & \frac{n^2}{(n-1)^2}
                       \l[\frac{\E(\widetilde{S}^2_{f1})}{n_1} +
         \frac{E(\widetilde{S}^2_{f0})}{n_0} + \frac{1}{n^2} \l\{\tau_f^2 -n p_f(1-p_f)
         \tau^2 + 2(n-1)(2p_f-1) \tau_f\tau \r\} \r].
\end{eqnarray*}
\qed

A potential complication with this estimator in practice is that its
estimate (along with the variance) would change under an additive
transformation, i.e., $Y_i(t) \to Y_i(t)+\delta$ for $t=0,1$ and a
given constant $\delta$. This issue is not due to the specific
construction of the proposed estimator. It instead reflects the
fundamental issue of many prescription effects including the population average value and
PAPE that they cannot be defined solely in terms of multiples of
$Y_i(1)-Y_i(0)$ (see Appendix~\ref{app:PAVex} for a numerical
example).  One solution is to center the outcome variable such that
$\sum_{i=1}^n Y_iT_i/n_1 + \sum_{i=1}^n Y_i(1-T_i)/n_0=0$ holds.  This
solution is motivated by the fact that when the condition holds in the
population (i.e. $\E(Y_i(1)+Y_i(0))=0$), the variance of the PAPE
estimator is minimized.

\subsubsection{A Numerical Example Showing the Lack of Additive
  Invariance for the Population Average Value}
\label{app:PAVex}

\begin{table}[t!]
  \centering\setlength{\tabcolsep}{2pt}
  \spacingset{1}
  \begin{tabular}{|c|ccccc|}
    \hline
    Individual & $T_i$ & $f(\bX_i)$& $Y_i$ & $Y_i(0)$ & $Y_i(1)$\\\hline
    A & 1 & 1& 2 & 0 & 2 \\
    B  & 1 & 0& 3 & 1 & 3\\
    C  & 0 & 0 & $-$1 & $-$1 & $-$1\\
    D  & 0 & 1&  1 & 1 & 0 \\
    E  & 1 & 0& 3 & 0 & 3\\\hline
  \end{tabular}
  \caption{A Numerical Example for Binary Treatment Assignment and Outcomes} \label{tb:sample_pav_table}
\end{table}

Consider an ITR $f: \cX \to \{0,1\}$, and we would like to know its
population average value.  Table~\ref{tb:sample_pav_table} shows
an numerical example with the observed outcome $Y_i$, the ITR
$f(X_i)$, the actual assignment $T_i$, and the potential outcomes
$Y_i(0), Y_i(1)$.  Then, in this example, we have $n_1=3$ (A,B,E),
$n_0=2$ (C,D), and the population average value estimator would be:
\begin{align*}
\hat\lambda_f(\cZ) \ &= \  \frac{1}{n_1} \sum_{i=1}^n Y_i T_i f(\bX_i) +
\frac{1}{n_0} \sum_{i=1}^n Y_i(1-T_i)(1-f(\bX_i))\\ & = \ \frac{1}{3} (1 \cdot 2 + 0 \cdot 3 + 0 \cdot 3)+ \frac{1}{2} (1 \cdot -1 + 0 \cdot 1)\\& = \frac{1}{6}
\end{align*}
Now let us consider an additive transformation
$Y_i(t) \to Y_i(t) + 1:=Y_i'(t)$ for $t =0,1$, where every outcome
value is raised by 1. Then, its population average value estimator is
now:
 \begin{align*}
 \hat\lambda'_f(\cZ) \ &= \  \frac{1}{n_1} \sum_{i=1}^n Y_i' T_i f(\bX_i) +
 \frac{1}{n_0} \sum_{i=1}^n Y_i'(1-T_i)(1-f(\bX_i))\\ & = \ \frac{1}{3} (1 \cdot 3 + 0 \cdot 4 + 0 \cdot 4)+ \frac{1}{2} (1 \cdot 0 + 0 \cdot 2)\\& = 1
 \end{align*}
 Note that the difference
 $\hat\lambda'_f(\cZ)-\hat\lambda_f(\cZ)=\frac{5}{6}\neq 1$ does not
 equal to the amount of additive transformation. The problem arises
 because they are not multiples of $Y_i(1)-Y_i(0)$ but rather they
 depend on what the actual assignments of the ITR.

\subsection{Proof of Theorem~\ref{thm:PAPEpest}}
\label{app:PAPEpest}

We begin by deriving the variance.  The derivation proceeds in the
same fashion as the one for Theorem~\ref{thm:PAPEest}.  The main
difference lies in the derivation of the covariance term, which we
detail below.  First, we note that,
\begin{eqnarray*}
  \Pr(f(\bX_i, \hat{c}_p(f))=1)
	& = & \int^{\infty}_{-\infty} \Pr(f(\bX_i, c)=1 \mid
              \hat{c}_p(f)=c) P(\hat{c}_p(f)=c) \d c\\
	& = & \int^{\infty}_{-\infty} \frac{\lfloor np \rfloor}{n} P(\hat{c}_p(f)=c) \d c\\
	& = & \frac{\lfloor np \rfloor}{n},
\end{eqnarray*}
where the second equality follows from the fact that once conditioned
on $\hat{c}_p(f) = c$, exactly $\lfloor np \rfloor$ out of $n$ units will
be assigned to the treatment condition.  Given this result, we can
compute the covariance as follows,
\begin{eqnarray*}
	& & {\rm Cov}(\widetilde{Y}_i(1) - \widetilde{Y}_i(0), \widetilde{Y}_j(1) - \widetilde{Y}_j(0)) \\
	& = & {\rm Cov}\l\{ \left(f(\bX_i,\hat{c}_p(f))- p\right)\tau_i, \left(f(\bX_j,\hat{c}_p(f))- p\right)\tau_j\r\} \\
	& = & {\rm Cov}\l\{ f(\bX_i,\hat{c}_p(f))\tau_i, f(\bX_j,\hat{c}_p(f))\tau_j\r\} - 2p{\rm Cov}\l(\tau_i, f(\bX_j,\hat{c}_p(f))\tau_j\r)\\
	& = & \frac{n \lfloor np \rfloor (\lfloor np \rfloor -1)-{\lfloor np\rfloor}^2(n-1)}{n^2(n-1)}\E(\tau_i\mid f(\bX_i,\hat{c}_p(f))=1)^2- 2p{\rm Cov}\l(\tau_i, f(\bX_j,\hat{c}_p(f))\tau_j\r)\\
	& = & \frac{\lfloor np \rfloor(\lfloor np \rfloor-n)}{n^2(n-1)}\kappa_1(\hat{c}_p(f))^2 + \frac{2p\lfloor np \rfloor(n-\lfloor np \rfloor)}{n^2(n-1)}\l(\kappa_1(\hat{c}_p(f))^2-\kappa_1(\hat{c}_p(f))\kappa_0(\hat{c}_p(f))\r)\\
	& = & (2p-1)\frac{\lfloor np \rfloor(n-\lfloor np \rfloor)}{n^2(n-1)}\kappa_1(\hat{c}_p(f))^2-\frac{2p\lfloor np \rfloor(n-\lfloor np \rfloor)}{n^2(n-1)}\kappa_1(\hat{c}_p(f))\kappa_0(\hat{c}_p(f))\\
	& = & \frac{\lfloor np \rfloor(n-\lfloor np \rfloor)}{n^2(n-1)}\l\{(2p-1)\kappa_1(\hat{c}_p(f))^2-2p\kappa_1(\hat{c}_p(f))\kappa_0(\hat{c}_p(f))\r\}.
\end{eqnarray*}
Combining this covariance result with the expression for the marginal
variances yields the desired variance expression for
$\hat{\tau}_f(\hat{c}_p(f))$.

Next, we derive the upper bound of bias.  Using the same technique as
the proof of Theorem~\ref{thm:PAPEest}, we can rewrite the expectation
of the proposed estimator as,
\begin{equation*}
\E(\hat \tau_f(c_p)) \ = \ \E\l[\frac{1}{n} \sum_{i=1}^n \l\{Y_i\l(f(\bX_i,\hat{c}_p(f))\r) - pY_i(1)-(1-p)Y_i(0)\r\}
\r].
\end{equation*}
Now, define $F(c)=\mc{P}(s(\bX_i)\leq c)$. Without loss of generality,
assume $\hat{c}_p(f) > c_p$ (If this is not the case, we simply switch
the upper and lower limits of the integrals below).  Then, the bias of
the estimator is given by,
\begin{eqnarray*}
	\l|\E(\hat \tau_f(\hat{c}_p(f)))-\tau_f(c_p)\r|&= & \l|\E\l[\frac{1}{n} \sum_{i=1}^n \l\{Y_i\l(f(\bX_i,\hat{c}_p(f))\r)-Y_i\l(f(\bX_i,c_p)\r)\r\}\r]\r| \\
	& = & \l|\E_{\hat{c}_p(f)}\l[\int^{\hat{c}_p(f)}_{c_p}
              \E(\tau_i \mid s(\bX_i)=c) \d F(c)\r]\r|\\
	& = & \l|\E_{F(\hat{c}_p(f))}\l[\int^{F(\hat{c}_p(f))}_{F(c_p)}
              \E(\tau_i\mid s(\bX_i)=F^{-1}(x)) \d x\r]\r|\\
	& \leq & \E_{F(\hat{c}_p(f))}\l[\l|F(\hat{c}_p(f))-(1-p)\r| \times \max_{c \in
                 [c_p,\hat{c}_p(f)]} \l | \E(\tau_i\mid s(\bX_i)=c)\r|\r].
\end{eqnarray*}
By the definition of $\hat{c}_p(f)$, $F(\hat{c}_p(f))$ is the
$(n-\lfloor np\rfloor)$th order statistic of $n$ independent uniform
random variables, and thus follows the Beta distribution with the
shape and scale parameters equal to $n - \lfloor np \rfloor$ and
$\lfloor np \rfloor+1$, respectively. For the special case where
$p=1$, we define the $0$th order statistic of $n$ uniform random
variables to be 0, and by extension also define the ``beta
distribution'' with shape parameter $\leq 0$ to be $H(x)$ where $H(x)$
is the Heaviside step function.  Therefore, we have,
\begin{equation}
\mc{P}(|F(\hat{c}_p(f))-p|>\epsilon) \ = \ 1-B(1-p+\epsilon, n - \lfloor np \rfloor, \lfloor np \rfloor+1)+B(1-p-\epsilon, n-\lfloor np \rfloor,
\lfloor np \rfloor+1), \label{eq:Fcp}
\end{equation}
where $B(\epsilon, \alpha, \beta)$ is the incomplete beta function,
i.e.,
\begin{equation*}
  B(\epsilon,\alpha,\beta) \ = \ \int^\epsilon_0 t^{\alpha -1} (1-t)^{\beta -1} \d t.
\end{equation*}
Combining with the result above, the desired result follows. \qed

\subsection{Estimation and Inference of the Population Average
  Prescriptive Difference of Fixed ITRs}
\label{app:PAPDpest}

\begin{theorem} {\sc (Bias and Variance of the PAPD Estimator
    with a Budget Constraint)} \label{thm:PAPDpest} \spacingset{1}
  Under
  Assumptions~\ref{asm:SUTVA},~\ref{asm:randomsample},~and~\ref{asm:comrand},
  the bias of the proposed estimator of the PAPD with a budget
  constraint $p$ defined in equation~\eqref{eq:PAPDpest} can be
  bounded as follows,
	\begin{align*}
	\mc{P}_{\hat{c}_{p}(f),\hat{c}_{p}(g)}(|\E\{\widehat \Delta_p(f,g,\cZ) -\Delta_p(f,g)\mid \hat{c}_{p}(f),\hat{c}_{p}(g) \}|\geq
	\epsilon)  \ &\leq \  1-2B(1-p+\gamma_p(\epsilon), n-\lfloor np
	\rfloor, \lfloor np \rfloor+1)\\&+2B(1-p-\gamma_p(\epsilon),
	n-\lfloor np \rfloor, \lfloor np \rfloor+1),
	\end{align*}
	where any given constant $\epsilon > 0$, $B(\epsilon, \alpha,
        \beta)$ is the incomplete beta function (if $\alpha = 0$ and
        $\beta > 0$, we set $B(\epsilon,\alpha, \beta):=H(\epsilon)$
        for all $\epsilon$ where $H(\epsilon)$ is the Heaviside step
        function), and
	\begin{equation*}
	\gamma_p(\epsilon)\ = \ \frac{\epsilon}{\max_{c  \in
			[c_p(f)-\epsilon,\  c_p(f)+\epsilon],\ d  \in
			[c_p(g)-\epsilon,\  c_p(g)+\epsilon]} \{\E(\tau_i\mid s_f(\bX_i)=c), \E(\tau_i\mid s_g(\bX_i)=d)\}}.
	\end{equation*}
	The variance of the estimator is,
	\begin{eqnarray*}
		\V(\widehat \Delta_p(f,g,\cZ)) & = &
		\frac{\E(S_{fgp1}^2)}{n_1} +
		\frac{\E(S_{fgp0}^2)}{n_0}  + \frac{\lfloor np
                                                        \rfloor(\lfloor
                                                        np
                                                        \rfloor-n)}{n^2(n-1)}(\kappa_{f1}(p)^2
                                                        +\kappa_{g1}(p)^2)
          \\
          & & - 2\left(\Pr(f(\bX_i,\hat{c}_p(f))=g(\bX_i,\hat{c}_p(g))=1) -\frac{\lfloor np
          	\rfloor^2}{n^2}\right)\kappa_{f1}(p)\kappa_{g1}(p),
	\end{eqnarray*}
	where
        $S_{fgpt}^2 = \sum_{i=1}^n (Y_{fgpi}(t) -
        \overline{Y_{fgp}(t)})^2/(n-1)$ with
        $Y_{fgpi}(t) = \left\{f(\bX_i,\hat{c}_p(f)) -
          g(\bX_i,\hat{c}_p(g))\right\}Y_i(t)$ and
        $\overline{Y_{fgp}(t)}=\sum_{i=1}^n Y_{fgpi}(t)/n$ for
        $t=0,1$.
\end{theorem}
To estimate the variance, it is tempting to replace all the unknown
parameters with their sample analogues.  However, unlike the case of
the variance of the PAPE estimator under a budget constraint (see
Theorem~\ref{thm:PAPEpest}), there is no useful identity for the joint
probability $\Pr(f(\bX_i,\hat{c}_p(f))=g(\bX_i,\hat{c}_p(g))=1)$ under
general $g$.  Thus, an empirical analogue of $\hat{c}_p(f)$ and
$\hat{c}_p(g)$ is not a good estimate because it is solely based on
one realization.  Thus, we use the following conservative bound,
\begin{align*}
& \quad - \left(\Pr(f(\bX_i,\hat{c}_p(f))=g(\bX_i,\hat{c}_p(g))=1) -\frac{\lfloor np
  \rfloor^2}{n^2}\right)\kappa_{f1}(p)\kappa_{g1}(p)\\
  \leq & \quad \frac{\lfloor np  \rfloor  \max\{\lfloor  np \rfloor,
	n-\lfloor np \rfloor\}}{n^2(n-1)}|\kappa_{f1}(p)\kappa_{g1}(p)|,
\end{align*}
where the inequality follows because the maximum is achieved when the
scoring rules of $f$ and $g$, i.e., $s_f(\bX_i)$ and $s_g(\bX_i)$, are
perfectly correlated. We use this upper bound in our simulation and
empirical studies. In Section~\ref{sec:synthetic}, we find that this
upper bound estimate of the variance produces only a small
conservative bias.

\begin{proof}
The proof of the bounds for the expectation and variance of the proposed
estimator largely follows the proof given in
Appendix~\ref{app:PAPEpest}. The only significant difference is the
calculation of the covariance term, which is given below.
\begin{eqnarray*}
	& & {\rm Cov}(Y_i^\ast(1) - Y_i^\ast(0), Y_j^\ast(1) - Y_j^\ast(0)) \\
	& = & {\rm Cov}\l(
              \l\{f(\bX_i,\hat{c}_p(f))-g(\bX_i,\hat{c}_p(g))\r\}\tau_i,
              \l\{f(\bX_j,\hat{c}_p(f))-g(\bX_j,\hat{c}_p(g))\r\}\tau_j\r) \\
	& = & \Cov(f(\bX_i,\hat{c}_p(f))\tau_i,f(\bX_j,\hat{c}_p(f))\tau_j)+\Cov(g(\bX_i,\hat{c}_p(g))\tau_i,g(\bX_j,\hat{c}_p(g))\tau_j)\\
	& & - 2\Cov(f(\bX_i,\hat{c}_p(f))\tau_i,g(\bX_j,\hat{c}_p(g))\tau_j) \\
	& = & \frac{\lfloor np \rfloor(\lfloor np \rfloor-n)}{n^2(n-1)}\l(\kappa_{f1}(p)^2 +\kappa_{g1}(p)^2\r) -
            2\Cov(f(\bX_i,\hat{c}_p(f))\tau_i,g(\bX_j,\hat{c}_p(g))\tau_j)\\
     & = & \frac{\lfloor np \rfloor(\lfloor np
           \rfloor-n)}{n^2(n-1)}\l(\kappa_{f1}(p)^2
           +\kappa_{g1}(p)^2\r) \\
  & & - 2\left(\Pr(f(\bX_i,\hat{c}_p(f))=g(\bX_i,\hat{c}_p(g))=1) -\frac{\lfloor np
     	\rfloor^2}{n^2}\right)\kappa_{f1}(p)\kappa_{g1}(p)
\end{eqnarray*}
\end{proof}

\subsection{Proof of Theorem~\ref{thm:AUPECest}}
\label{app:AUPECest}

The derivation of the variance expression in
Theorem~\ref{thm:AUPECest} proceeds in the same fashion as
Theorem~\ref{thm:PAPEest} (see Appendix~\ref{app:PAPEpest}) with the
only non-trivial change being the calculation of the covariance term.
Note $\Pr(f(\bX_i, \hat{c}_{\frac{k}{n}}(f))=1)=k/n$ for $t=0,1$ and
$n_f = Z\sim {\rm Binom}(n, p_f)$.  Then, we have:
\begin{eqnarray*}
  & & {\rm Cov}(Y_i^\ast(1) - Y_i^\ast(0), Y_j^\ast(1) - Y_j^\ast(0)) \\
  & = & {\rm Cov}\l[ \left\{\frac{1}{n}\l(\sum_{k=1}^{n_f}
        f(\bX_i,\hat{c}_{k/n}(f))+\sum_{k=n_f+1}^{n}
        f(\bX_i,\hat{c}_{n_f/n}(f))\r)-
        \frac{1}{2}\right\}\tau_i ,\r. \\
  & & \l. \hspace{1in} \left\{\frac{1}{n}\l(\sum_{k=1}^{n_f} f(\bX_j,\hat{c}_{k/n}(f))+\sum_{k=n_f+1}^{n} f(\bX_j,\hat{c}_{n_f/n}(f))\r)- \frac{1}{2}\right\}\tau_j\r]\\
  & = & \E\l\{ {\rm Cov}\l[ \left\{\frac{1}{n}\l(\sum_{k=1}^{Z}
        f(\bX_i,\hat{c}_{k/n}(f))+\sum_{k=Z+1}^{n}
        f(\bX_i,\hat{c}_{Z/n}(f))\r)- \frac{1}{2}\right\}\tau_i
        ,\r.\r. \\
  & & \l.\l. \hspace{1in} \left\{\frac{1}{n}\l(\sum_{k=1}^{Z}
      f(\bX_j,\hat{c}_{k/n}(f))+\sum_{k=Z+1}^{n}
      f(\bX_j,\hat{c}_{Z/n}(f))\r)- \frac{1}{2}\right\}\tau_j
      \ \Bigl | \  Z\r] \r\}\\
  & & + {\rm Cov}\l\{\E\l[ \left\{\frac{1}{n}\left(\sum_{k=1}^{Z}
      f(\bX_i,\hat{c}_{k/n}(f))+\sum_{k=Z+1}^{n}
      f(\bX_i,\hat{c}_{Z/n}(f))\r)- \frac{1}{2}\right\}\tau_i
      \ \Bigl | \ Z \r] ,\r. \\
  & & \l.\biggl. \hspace{1in} \E\l[\left\{\frac{1}{n}\l(\sum_{k=1}^{Z}
      f(\bX_j,\hat{c}_{k/n}(f))+\sum_{k=Z+1}^{n}
      f(\bX_j,\hat{c}_{Z/n}(f))\r)- \frac{1}{2}\right\}\tau_j
      \ \Bigl | \ Z\r] \r\}\\
  & =  &\E\l[- \frac{1}{n} \l\{\sum_{k=1}^Z
	 \frac{k(n-k)}{n^2(n-1)}\kappa_{f1}(k/n)\kappa_{f0}(k/n)
	 +  \frac{Z(n-Z)^2}{n^2(n-1)}\kappa_{f1}(Z/n)\kappa_{f0}(Z/n)\r\}\r.\\
  & & \hspace{.25in}
      - \frac{2}{n^4(n-1)}\sum_{k=1}^{Z-1}\sum_{k^\prime=k+1}^{Z}
      k(n-k^\prime)\kappa_{f1}(k/n)\kappa_{f1}(k^\prime/n)\\
  & & \hspace{.25in} -
      \frac{Z^2(n-Z)^2}{n^4(n-1)}\kappa_{f1}(Z/n)^2
      - \frac{2(n-Z)^2}{n^4(n-1)}\sum_{k=1}^Z k
      \kappa_{f1}(Z/n)\kappa_{f1}(k/n)\\
  & & \hspace{.25in}\l. +\frac{1}{n^4}\sum_{k=1}^Z
      k(n-k)\kappa_{f1}(k/n)^2\r] + \V\l(\sum_{i=1}^Z \frac{i}{n} \kappa_{f1}(i/n) + \frac{(n-Z)Z}{n} \kappa_{f1}(Z/n)\r),
\end{eqnarray*}
where the last equality is based on the results from
Appendix~\ref{app:PAPEpest}.

For the bias, we can rewrite $\Gamma_f$ as,
\begin{equation*}
\Gamma_f \ = \ \int^{p_f}_0\E\{Y_i(f(\bX_i,c_p))\} \d p +(1-p_f)\E\{Y_i(f(\bX_i,c^*))\}-\frac{1}{2}\E\{Y_i(1)+Y_i(0)\},
\end{equation*}
and similarly its estimator $\widehat{\Gamma}_f$ as,
\begin{equation*}
\E(\widehat{\Gamma}_f) \ = \ \E\l\{\int^{\hat{p}_f}_0Y_i(f(\bX_i,c_p)) \d p\r\} +\E\{(1-\hat{p}_f)Y_i(f(\bX_i,c^*))\}-\frac{1}{2}\E\{Y_i(1)+Y_i(0)\}.
\end{equation*}
Therefore, the bias of the estimator is, using a derivation similar to Appendix \ref{app:PAPEpest}:
\begin{eqnarray*}
	\l|\E(\widehat \Gamma_{f})-\Gamma_{f}\r|&\leq &
                                                        \E\l[|p_f-\hat{p}_f|\max_{c \in \{\min\{\hat{p}_f,p_f\}, \max\{\hat{p}_f,p_f\}\}} \E\{Y_i(f(\bX_i,c))-Y_i(f(\bX_i,c^*))\}| \r.\\
                                                & & \hspace{.5in} +  |\E\{Y_i(f(\bX_i,c^*))\}-\E\{Y_i(f(\bX_i,\hat{c}_{p_f}))\}|\Bigr]\\
	&\leq &(\epsilon +1) \max_{c \in
		[c^*-\epsilon,c^*+\epsilon]} |\E[Y_i(f(\bX_i,c))-Y_i(f(\bX_i,c^*))]|\\
	& \leq & (\epsilon+1)\epsilon \max_{c \in
		[c^*-\epsilon,c^*+\epsilon]} \l | \E(\tau_i\mid s(\bX_i)=c)\r|.
\end{eqnarray*}
Now, taking the bound $\epsilon(1+\epsilon)\leq 2\epsilon$ for
$0\leq \epsilon\leq 1$ in equation~\eqref{eq:Fcp} of
Appendix~\ref{app:PAPEpest}, we have the desired result.
\qed

\subsection{Evaluation of an Estimated ITR with No Budget Constraint}
\label{app:cvnobudget}

Formally, we define an machine learning algorithm $F$ to be a deterministic map from
the space of observable data $\cZ=\{\cX, \cT, \cY\}$ to the space of
ITRs $\mathcal{F}$,
\begin{equation*}
  F:\cZ \longrightarrow \mathcal{F}.
\end{equation*}
We emphasize that no restriction is placed on the machine learning algorithm $F$ or
ITR $f$.

To extend the population average value (equation~\eqref{eq:PAV}), we consider the average
performance of an estimated ITR across training data sets of fixed
size.  First, for any given values of pre-treatment variables
$\bX_i=\bx$, we define the average treatment proportion under the
estimated ITR obtained by applying the machine learning algorithm $F$ to training
data $\cZ^{tr}$ of size $n-m$,
\begin{equation*}
  \bar{f}_{F}(\bx) \ = \ \E\{\hat{f}_{\cZ^{tr}}(\bx) \mid
  \bX_i = \bx\} \ = \ \Pr(\hat{f}_{\cZ^{tr}}(\bx) = 1 \mid \bX_i  = \bx)
\end{equation*}
where the expectation is taken over the random sampling of training
data $\cZ^{tr}$.  Although $\bar{f}_F$ depends on the training data
size, we suppress it to ease notational burden.

Then, the population average value of an estimated ITR can be defined as,
\begin{equation}
  \lambda_{F} \ = \ \E\l\{\bar{f}_{F}(\bX_i) Y_i(1) +
  (1-\bar{f}_{F}(\bX_i)) Y_i(0)\r\} \label{eq:PAVcv}
\end{equation}
where the expectation is taken over the population distribution of
$\{\bX_i, Y_i(1), Y_i(0)\}$.  In contrast to the population average value of a fixed ITR,
this estimand accounts for the estimation uncertainty of the ITR by
averaging over the random sampling of training sets.

To generalize the PAPE (equation~\eqref{eq:PAPE}), we first define
the population proportion of units assigned to the treatment condition
under the estimated ITR as follows,
\begin{equation*}
p_{F} \ = \ \E\{\Pr(\hat{f}_{\cZ^{tr}}(\bX_i)=1 \mid \bX_i)\} \ = \ \E\{\bar{f}_F(\bX_i)\}
\end{equation*}
where the expectation is taken with respect to the sampling of
training data of size $n-m$ and the population distribution of
$\bX_i$.  Then, the PAPE of an estimated ITR is given by,
\begin{equation}
  \tau_{F} \ = \ \E\{\lambda_{F} - p_{F} Y_i(1) - (1-p_{F}) Y_i(0)\}, \label{eq:PAPEcv}
\end{equation}
where $\lambda_F$ is the population average value of the estimated ITR defined in
equation~\eqref{eq:PAVcv}.

\subsubsection{The Population Average Value}
\label{app:PAVcvest}

We begin by considering the following cross-validation estimator of
the population average value for an estimated ITR (equation~\eqref{eq:PAVcv}),
\begin{equation}
  \hat{\lambda}_{F} \ = \ \frac{1}{K}\sum_{k=1}^K \hat\lambda_{\hat{f}_{-k}}(\cZ_k). \label{eq:PAVcvest}
\end{equation}
The following theorem proves the unbiasedness of this estimator and
derives its exact variance expression under the Neyman's repeated
sampling framework.
\begin{theorem} {\sc (Unbiasedness and Exact Variance of the
  Cross-Validation Population Average Value Estimator)} \label{thm:PAVcvest} \spacingset{1}
  Under
  Assumptions~\ref{asm:SUTVA},~\ref{asm:randomsample},~and~\ref{asm:comrand},
  the expectation and variance of the cross-validation Population Average Value estimator
  defined in equation~\eqref{eq:PAVcvest} are given by,
  \begin{eqnarray*}
    \E(\hat{\lambda}_{F}) & = & \lambda_{F} \\
    \V(\hat{\lambda}_{F}) & = & \frac{\E(S_{\hat{f}1}^2)}{m_1} +
                                 \frac{\E(S_{\hat{f}0}^2)}{m_0}
                                    +\E \l\{\Cov(\hat{f}_{\cZ^{tr}}(\bX_i),
                                    \hat{f}_{\cZ^{tr}}(\bX_j) \mid
                                \bX_i, \bX_j)\tau_i\tau_j\r\} - \frac{K-1}{K}\E(S_{F}^2)
	\end{eqnarray*}
	for $i \ne j$ where
        $S_{\hat{f}t}^2 = \sum_{i=1}^{m} \l(Y_{\hat{f}i}(t) -
        \overline{Y_{\hat{f}}(t)}\r)^2/(m-1)$,
        $S_{F}^2 = \sum_{k=1}^K \l\{\hat\lambda_{\hat{f}_{-k}}(\cZ_k)
        - \overline{\hat\lambda_{\hat{f}_{-k}}(\cZ_k)}\r\}^2/(K-1)$,
        and $\tau_i = Y_i(1)-Y_i(0)$ with
        $Y_{\hat{f}i}(t) =
        \mathbf{1}\{\hat{f}_{\cZ^{tr}}(\bX_i)=t\}Y_i(t)$,
        $\overline{Y_{\hat{f}}(t)} = \sum_{i=1}^{m}
        Y_{\hat{f}i}(t)/m$, and
        $\overline{\hat\lambda_{\hat{f}_{-k}}(\cZ_k)} = \sum_{k=1}^K
        \hat\lambda_{\hat{f}_{-k}}(\cZ_k)/K$ for $t=\{0,1\}$.
\end{theorem}
The proof of unbiasedness is similar to that of
Appendix~\ref{app:PAPEest} and thus is omitted.  To derive the
variance, we first introduce the following useful lemma, adapted from
\cite{nadeau2000inference}.
\begin{lemma}
	\label{lemma:cvcov}
	\begin{eqnarray*}
		\E(S_F^2) & = & \V(\hat\lambda_{\hat{f}_{-k}}(\cZ_k))-\Cov(\hat\lambda_{\hat{f}_{-k}}(\cZ_k),\hat\lambda_{\hat{f}_{-\ell}}(\cZ_\ell)),\\
		\V(\hat{\lambda}_F) & = & \frac{\V(\hat\lambda_{\hat{f}_{-k}}(\cZ_k))}{K}+ \frac{K-1}{K}\Cov(\hat\lambda_{\hat{f}_{-k}}(\cZ_k),\hat\lambda_{\hat{f}_{-\ell}}(\cZ_\ell)).
	\end{eqnarray*}
        where $k \ne \ell$.
\end{lemma}
The lemma implies,
\begin{equation}
\V(\hat{\lambda}_F) \ = \ \V(\hat{\lambda}_{f_{-k}}(\cZ_k))-\frac{K-1}{K}\E(S_F^2). \label{eq:varrelation}
\end{equation}
We then follow the same process of derivation as in
Appendix~\ref{app:PAPEest} while replacing $Y^*_i(t)$ with
$Y_{\hat{f}i}(t)$ for $t \in \{0,1\}$. The only difference lies in the
covariance term, which can be expanded as follows,
\begin{align*}
\Cov(Y_{\hat{f}i}(1)-Y_{\hat{f}i}(0),Y_{\hat{f}j}(1)-Y_{\hat{f}j}(0)) \
= \ &  {\rm Cov}(\hat{f}_{\cZ^{tr}}(\bX_i)\tau_i+Y_i(0),\hat{f}_{\cZ^{tr}}(\bX_j)\tau_j+Y_j(0)) \\
= \ &  {\rm Cov}(\hat{f}_{\cZ^{tr}}(\bX_i)\tau_i,\hat{f}_{\cZ^{tr}}(\bX_j)\tau_j) \\
= \ & \E\l[{\rm Cov}(\hat{f}_{\cZ^{tr}}(\bX_i)\tau_i,
\hat{f}_{\cZ^{tr}}(\bX_j)\tau_j \mid \bX_i, \bX_j,
\tau_i, \tau_j )\r]\\
= \ & \E\l[{\rm
	Cov}(\hat{f}_{\cZ^{tr}}(\bX_i),\hat{f}_{\cZ^{tr}}(\bX_j)
\mid \bX_i, \bX_j)\tau_i \tau_j\r].
\end{align*}
So the full variance expression is:
\begin{align*}
\V(\hat{\lambda}_F) \  = \ &\V(\hat{\lambda}_{f_{-k}}(\cZ_k))-\frac{K-1}{K}\E(S_F^2)\\
= \ &  \frac{\E(S^2_{\hat{f}1})}{m_1}+\frac{\E(S^2_{\hat{f}0})}{m_0}+\E\l[{\rm
	Cov}(\hat{f}_{\cZ^{tr}}(\bX_i),\hat{f}_{\cZ^{tr}}(\bX_j)
\mid \bX_i, \bX_j)\tau_i \tau_j\r]-\frac{K-1}{K}\E(S_F^2)
\end{align*}
\qed

When compared to the case of a fixed ITR with the sample size of $m$
(see Theorem~\ref{thm:PAVest} in Appendix~\ref{app:PAV}), the variance
has two additional terms.  The covariance term accounts for the
estimation uncertainty about the ITR, and is typically positive
because two units, for which an estimated ITR makes the same treatment
assignment recommendation, are likely to have causal effects with the
same sign.  The second term is due to the efficiency gain resulting
from the $K$-fold cross-validation rather than evaluating an estimated
ITR once.

The cross-validation estimate of $\E(S_{\hat{f}t}^2)$ is
straightforward and is given by,
\begin{equation*}
  \widehat{\E(S_{\hat{f}t}^2)} \ = \ \frac{1}{K(m_t-1)} \sum_{k=1}^K
  \sum_{i=1}^m \mathbf{1}\{T_i^{(k)} = t\} \l\{Y_{\hat{f}i}^{(k)}(t) - \overline{Y_{\hat{f}t}^{(k)}}\r\}^2,
\end{equation*}
where
$Y_{\hat{f}i}^{(k)}(t) = \mathbf{1}\{\hat{f}_{-k}(\bX_i^{(k)})=t\}
Y_i^{(k)}(t)$ and
$\overline{Y_{\hat{f}t}^{(k)}}=\sum_{i=1}^m \mathbf{1}\{T_i^{(k)} = t
\}Y_{\hat{f}i}^{(k)}(t) /m_t$.  In contrast, the estimation of this
cross-validation variance requires care.  In particular, although it
is tempting to estimate $\E(S_F^2)$ using the realization of $S_F^2$,
this estimate is highly variable especially when $K$ is small.  As a
result, it often yields a negative overall variance estimate.  We
address this problem by first noting that Lemma~\ref{lemma:cvcov}
implies,
\begin{equation*}
  \V(\hat{\lambda}_{\hat{f}_{-k}}(\cZ_k)) \ = \
\frac{\E(S_{\hat{f}1}^2)}{m_1} +
\frac{\E(S_{\hat{f}0}^2)}{m_0}
+\E \l\{\Cov(\hat{f}_{\cZ^{tr}}(\bX_i),
\hat{f}_{\cZ^{tr}}(\bX_j) \mid
\bX_i, \bX_j)\tau_i\tau_j\r\} \ \geq \ \E(S_{F}^2).
\end{equation*}
Then, this inequality suggests the following consistent estimator of
$\E(S^2_F)$,
\begin{equation*}
\widehat{\E(S_F^2)} \ = \ \min\l(S_F^2, \widehat{\V(\hat{\lambda}_{\hat{f}_{-k}}(\cZ_k))}\r).
\end{equation*}
Although this yields a conservative estimate of $\V(\hat\lambda_F)$ in
finite samples, the bias appears to be small in practice (see
Section~\ref{subsec:coverage}).

Finally, for the estimation of the covariance term, since
$\hat{f}_{\cZ^{tr}}$ is binary, we have,
\begin{eqnarray*}
& & \Cov(\hat{f}_{\cZ^{tr}}(\bX_i), \hat{f}_{\cZ^{tr}}(\bX_j) \mid \bX_i,
\bX_j) \\
& = & \Pr(\hat{f}_{\cZ^{tr}}(\bX_i)=\hat{f}_{\cZ^{tr}}(\bX_j)=1 \mid
      \bX_i, \bX_j)-\Pr(\hat{f}_{\cZ^{tr}}(\bX_i)=1 \mid \bX_i)
      \Pr(\hat{f}_{\cZ^{tr}}(\bX_j)=1 \mid \bX_j),
\end{eqnarray*}
for $i \ne j$.  An unbiased cross-validation estimator of this
covariance (given $\bX_i$ and $\bX_j$) is,
\begin{equation*}
\widehat{\Cov(\hat{f}_{\cZ^{tr}}(\bX_i),
  \hat{f}_{\cZ^{tr}}(\bX_j) \mid \bX_i, \bX_j)}
\ = \ \frac{1}{K}\sum_{k=1}^K \hat{f}_{-k}(\bX_i) \hat{f}_{-k}(\bX_j)-
\frac{1}{K}\sum_{k=1}^K\hat{f}_{-k}(\bX_i) \frac{1}{K}\sum_{k=1}^K \hat{f}_{-k}(\bX_j).
\end{equation*}
Thus, we have the following cross-validation
estimator of the required term,
\begin{align*}
 & \widehat{\E\l\{\Cov(\hat{f}_{\cZ^{tr}}(\bX_i), \hat{f}_{\cZ^{tr}}(\bX_j)
                \mid \bX_i, \bX_j) Y_i(s)Y_j(t) \r\}} \nonumber \\
  = \ &   \frac{\sum_{i=1}^n
        \sum_{j\neq i } \bone\{T_i = s, T_j
        = t\} Y_iY_j \cdot
        \widehat{\Cov(\hat{f}_{\cZ^{tr}}(\bX_i),
        \hat{f}_{\cZ^{tr}}(\bX_j) \mid \bX_i, \bX_j)}}{\sum_{i=1}^n \sum_{j\neq i } \bone\{T_i = s, T_j
        = t\}}.
\end{align*}
for $s,t \in \{0,1\}$.  However, since this naive calculation is
computationally expensive, we rewrite it as follows to reduce the
computational time from $O(n^2 K)$ to $O(nK)$,
\begin{align*}
      & \frac{\sum_{k=1}^K \left(\sum_{i=1}^n
    	\bone\{T_i = s\} Y_i \hat{f}_{-k}(\bX_i)\right)\left(\sum_{i=1}^n
    	\bone\{T_i = t\} Y_i \hat{f}_{-k}(\bX_i)\right)-\sum_{i=1}^n\bone\{T_i = s, T_i = t\}Y_i^2\hat{f}_{-k}(\bX_i)
    	}{K\left[\left(\sum_{i=1}^{n}\bone\{T_i =
          s\}\right) \left(\sum_{i=1}^{n}\bone\{T_i =
          t\}\right)-\sum_{i=1}^{n}\bone\{T_i = s, T_i = t\} \right]}\nonumber \\
      & - \  \frac{(\sum_{i=1}^n \sum_{k=1}^K
    	\bone\{T_i = s\} Y_i \hat{f}_{-k}(\bX_i))(\sum_{i=1}^n \sum_{k=1}^K
    	\bone\{T_i = t\} Y_i
         \hat{f}_{-k}(\bX_i))}{K^2\left[(\sum_{i=1}^{n}\bone\{T_i =
         s\}) (\sum_{i=1}^{n}\bone\{T_i =
         t\})-\sum_{i=1}^{n}\bone\{T_i = s, T_i = t\} \right]}\nonumber \\
     	 & + \ \frac{\sum_{i=1}^n\sum_{k=1}^K \bone\{T_i = s, T_i = t\}Y_i^2\hat{f}_{-k}(\bX_i)}{K^2\left[(\sum_{i=1}^{n}\bone\{T_i =
     		s\}) (\sum_{i=1}^{n}\bone\{T_i =
     		t\})-\sum_{i=1}^{n}\bone\{T_i = s, T_i = t\} \right]}.
\end{align*}

\subsubsection{The Population Average Prescriptive Effect (PAPE)}
\label{app:PAPEcvest}

Next, we propose the following cross-validation estimator of the PAPE
for an estimated ITR (equation~\eqref{eq:PAPEcv}),
\begin{equation}
  \hat{\tau}_F \ = \ \frac{1}{K}\sum_{k=1}^K \hat{\tau}_{\hat{f}_{-k}}(\cZ_k) \label{eq:PAPEcvest}
\end{equation}
where $\hat\tau_f(\cdot)$ is defined in equation~\eqref{eq:PAPEest}.
The next theorem shows the unbiasedness of this estimator and derives
its variance.
\begin{theorem} {\sc (Unbiasedness and Exact Variance of the
    Cross-Validation PAPE Estimator)} \label{thm:PAPEcvest}
  \spacingset{1} Under
  Assumptions~\ref{asm:SUTVA},~\ref{asm:randomsample},~and~\ref{asm:comrand},
  the expectation and variance of the cross-validation PAPE estimator
  defined in equation~\eqref{eq:PAPEcvest} are given by,
	\begin{eqnarray*}
          \E(\hat{\tau}_F)  & = & \tau_F \\
           \V(\hat{\tau}_F)
          & = & \frac{m^2}{(m-1)^2}\Bigg[\frac{\E(\widetilde{S}_{\hat{f}1}^2)}{m_1} +
                                 \frac{\E(\widetilde{S}_{\hat{f}0}^2)}{m_0} + \frac{1}{m^2}\l\{\tau_F^2- m
                                 p_F(1-p_F)\tau^2+2(m-1)(2p_F-1)\tau  \tau_F\r\}\\
          & &  \hspace{.75in} +
              \frac{1}{m^2}\E\Big\{\l\{(m-3)(m-2)\tau^2
              + (m^2-2m+2)\tau_i\tau_j -
             2(m-2)^2\tau\tau_i
              \r\}   \\  & & \hspace{1.35in}
              \times\Cov(\hat{f}_{\cZ^{tr}}(\bX_i), \hat{f}_{\cZ^{tr}}(\bX_j)
              \mid \bX_i, \bX_j) \Big\} \Bigg]  -  \frac{K-1}{K}\E(\widetilde{S}_{F}^2)
	\end{eqnarray*}
	for $i\neq j$, where
        $\widetilde{S}_{\hat{f}t}^2 = \sum_{i=1}^{m}
        (\widetilde{Y}_{\hat{f}i}(t) -
        \overline{\widetilde{Y}_{\hat{f}}(t)})^2/(m-1)$,
        $\widetilde{S}_{F}^2 = \sum_{k=1}^K
        (\hat\tau_{\hat{f}_{-k}}(\cZ_{k}) -
        \overline{\hat\tau_{\hat{f}_{-k}}(\cZ_k)})^2/(K-1)$ with
        $\widetilde{Y}_{\hat{f}i}(t) = (\hat{f}_{-k}(\bX_i) -
        \hat{p}_{\hat{f}_{-k}})Y_i(t)$,
        $\overline{\widetilde{Y}_{\hat{f}}(t)} = \sum_{i=1}^{m}
        \widetilde{Y}_{\hat{f}i}(t)/m$, and
        $\overline{\hat\tau_{\hat{f}_{-k}}(\cZ_k)} = \sum_{k=1}^K
        \hat\tau_{\hat{f}_{-k}}(\cZ_k)/K$, for
        $t=\{0,1\}$.
\end{theorem}
\begin{proof}
The proof of unbiasedness is similar to that of
Appendix~\ref{app:PAPEest} and thus is omitted. The derivation of the
variance is similar to that of Appendix~\ref{app:PAVcvest}. The key
difference is the calculation of the following covariance term, which
needs care due to the randomness of $\hat{f}_{\cZ^{tr}}$,
\begin{align*}
	&  {\rm Cov}(Y_i^\ast(1) - Y_i^\ast(0), Y_j^\ast(1) - Y_j^\ast(0)) \\
 = \ &  {\rm Cov}\l\{\l(\hat{f}_{\cZ^{tr}}(\bX_i) - \frac{1}{{m}}\sum_{i^\prime=1}^{m}
	\hat{f}_{\cZ^{tr}}(\bX_{i^\prime})\r)\tau_i, \l(\hat{f}_{\cZ^{tr}}(\bX_j) -
	\frac{1}{{m}} \sum_{j^\prime=1}^{m}
   \hat{f}_{\cZ^{tr}}(\bX_{j^\prime})\r)\tau_j\r\},
\end{align*}
where $i \ne j$. There are seven terms that
need to be carefully expanded,
\begin{align*}
	& \frac{{m}-2}{m^2}\Cov(\hat{f}_{\cZ^{tr}}(\bX_k)\tau_i,\hat{f}_{\cZ^{tr}}(\bX_k)\tau_j)	+  \frac{({m}-2)({m}-3)}{m^2}
              \Cov(\hat{f}_{\cZ^{tr}}(\bX_k)\tau_i,\hat{f}_{\cZ^{tr}}(\bX_\ell)\tau_j) \\
  + \ & \frac{({m}-1)^2}{m^2}\Cov(\hat{f}_{\cZ^{tr}}(\bX_i)\tau_i,\hat{f}_{\cZ^{tr}}(\bX_j)\tau_j)
 + \frac{1}{m^2}\Cov(\hat{f}_{\cZ^{tr}}(\bX_j)\tau_i,\hat{f}_{\cZ^{tr}}(\bX_i)\tau_j) \\
 - \	& \frac{2({m}-1)}{m^2}\Cov(\hat{f}_{\cZ^{tr}}(\bX_i)\tau_i,\hat{f}_{\cZ^{tr}}(\bX_i)\tau_j)
 + \frac{2({m}-2)}{m^2}\Cov(\hat{f}_{\cZ^{tr}}(\bX_j)\tau_i,\hat{f}_{\cZ^{tr}}(\bX_k)\tau_j) \\
 - \ &
       \frac{2({m}-2)({m}-1)}{m^2}\Cov(\hat{f}_{\cZ^{tr}}(\bX_i)\tau_i,\hat{f}_{\cZ^{tr}}(\bX_k)\tau_j),
\end{align*}
where $i,j,k,\ell$ represent indices that do not take an identical
value at the same time (e.g., $i \ne j$). Then, we rewrite the above
terms using the properties of covariance as follows,
\begin{align*}
  & \frac{({m}-2)\tau^2}{m^2}\V(\hat{f}_{\cZ^{tr}}(\bX_i))+\frac{({m}-2)({m}-3)\tau^2}{m^2} \Cov(\hat{f}_{\cZ^{tr}}(\bX_i),\hat{f}_{\cZ^{tr}}(\bX_j)) \\
 & +   \frac{({m}-1)^2}{m^2}\l[\E\l\{\Cov(\hat{f}_{\cZ^{tr}}(\bX_i),\hat{f}_{\cZ^{tr}}(\bX_j)
       \mid \bX_i, \bX_j)\tau_i\tau_j\r\}+\Cov(\bar{f}_F(\bX_i)\tau_i,\bar{f}_F(\bX_j)\tau_j)\r] \\
& +  \frac{1}{m^2}\l[\E\l\{\Cov(\hat{f}_{\cZ^{tr}}(\bX_i),\hat{f}_{\cZ^{tr}}(\bX_j)
       \mid \bX_i, \bX_j)\tau_i\tau_j\r\}+\Cov(\bar{f}_F(\bX_j)\tau_i,\bar{f}_F(\bX_i)\tau_j)\r] \\
& - \frac{2({m}-1)\tau}{m^2}(1-p_F)\E(\hat{f}_{\cZ^{tr}}(\bX_i)\tau_i) \\
& + \frac{2({m}-2)\tau}{m^2}\l[\E\l\{\Cov(\hat{f}_{\cZ^{tr}}(\bX_i),\hat{f}_{\cZ^{tr}}(\bX_j)
      \mid \bX_i, \bX_j)\tau_i]+p_F\Cov(\bar{f}_F(\bX_i),\tau_i)\r\}\r] \\
& -  \frac{2({m}-2)({m}-1)\tau}{m^2}\l[\E\l\{\Cov(\hat{f}_{\cZ^{tr}}(\bX_i),\hat{f}_{\cZ^{tr}}(\bX_j)
      \mid \bX_i, \bX_j)\tau_i\r\} +
      \Cov(\bar{f}_F(\bX_i)\tau_i, \bar{f}_F(\bX_j))\r] \\
= \ & \frac{({m}-2)\tau^2}{m^2}p_F(1-p_F)+\frac{({m}-2)({m}-3)\tau^2}{m^2} \Cov(\hat{f}_{\cZ^{tr}}(\bX_i),\hat{f}_{\cZ^{tr}}(\bX_j)) \\
& + \frac{m^2-2m+2}{m^2}\E\l\{\Cov(\hat{f}_{\cZ^{tr}}(\bX_i),\hat{f}_{\cZ^{tr}}(\bX_j)
      \mid \bX_i, \bX_j)\tau_i\tau_j\r\}+ \frac{1}{m^2}(\tau_F^2+2\tau p_F\tau_F) \\
& - \frac{2({m}-1)\tau^2}{m^2}p_F(1-p_F)
      -\frac{2({m}-1)\tau\tau_F}{m^2}(1-p_F) +\frac{2({m}-2)\tau}{m^2}p_F\tau_F\\
 & -  \frac{2({m}-2)^2\tau}{m^2} \E\l\{\Cov(\hat{f}_{\cZ^{tr}}(\bX_i),\hat{f}_{\cZ^{tr}}(\bX_k)
    \mid \bX_i, \bX_j)\tau_i\r\}\\
  = \  & \frac{1}{m^2}\l(\tau_F^2-
         mp_F(1-p_F)\tau^2+2({m}-1)(2p_F-1)\tau \tau_F\r)\\
  & + \E\l[\frac{({m}-2)({m}-3)\tau^2}{m^2}
    \Cov(\hat{f}_{\cZ^{tr}}(\bX_i),\hat{f}_{\cZ^{tr}}(\bX_j) \mid
    \bX_i, \bX_j)\r. \\
  & \hspace{.25in}-\frac{2(m-2)^2\tau}{m^2}\Cov(\hat{f}_{\cZ^{tr}}(\bX_i),\hat{f}_{\cZ^{tr}}(\bX_j) \mid
    \bX_i, \bX_j)\tau_i. \\
  &\l. \hspace{.5in} +\frac{m^2-2m+2}{m^2}\Cov(\hat{f}_{\cZ^{tr}}(\bX_i),\hat{f}_{\cZ^{tr}}(\bX_j)
  \mid \bX_i, \bX_j)\tau_i\tau_j\r],
\end{align*}
for $i\neq j$, where we used the
results from Appendix~\ref{app:PAPEest} as well as
$\V(\hat{f}_{\cZ^{tr}}(\bX_i))=p_F(1-p_F)$ and
$\tau_F=\Cov(\bar{f}_F(\bX_i),\tau_i)=\Cov(\hat{f}_{\cZ^{tr}}(\bX_i),\tau_i)$.
\end{proof}

Like the case of the
population average value, the variance has two extra terms when compared to the case of a
fixed ITR (see Theorem~\ref{thm:PAPEest} in
Appendix~\ref{app:PAPEest}).  The estimation of the variance is
similar to that for the population average value.

\subsection{Proof of Theorem~\ref{thm:PAPEpcvest}}
\label{app:PAPEpcvest}

We begin by deriving the variance.  The derivation proceeds in the
same fashion as the one for Theorem~\ref{thm:PAPEcvest} (see
Appendix~\ref{app:PAPEcvest}).  The only non-trivial change is the
derivation of the covariance term, which we detail below.  First,
similar to the proof of Theorem~\ref{thm:PAPEpest} (see
Appendix~\ref{app:PAPEpest}), we note the following relation:
\begin{eqnarray*}
	\Pr(\hat{f}_{\cZ^{tr}}(\bX_i, \hat{c}_p)=1)
	& = & \E\left(\int^{\infty}_{-\infty}
              \Pr(\hat{f}_{\cZ^{tr}}(\bX_i, c)=1 \mid \cZ^{tr},
              \hat{c}_p=c) P(\hat{c}_p=c \mid \cZ^{tr}) \d c\right)\\
	& = &  \E\left(\int^{\infty}_{-\infty} \frac{\lfloor mp
              \rfloor}{m} P(\hat{c}_p=c \mid \cZ^{tr}) \d c\right)\\
	& = & \frac{\lfloor mp \rfloor}{m},
\end{eqnarray*}
where the second equality follows from the fact that conditioned on a fixed training set $\cZ^{tr}$ and conditioned
on $\hat{c}_p = c$, exactly $\lfloor mp \rfloor$ out of $m$ units will
be assigned to the treatment condition.  Given this result, we can
compute the covariance as follows,
\begin{eqnarray*}
	& & {\rm Cov}(\widetilde{Y}_i(1) - \widetilde{Y}_i(0), \widetilde{Y}_j(1) - \widetilde{Y}_j(0)) \\
	& = & {\rm Cov}\l\{\left(\hat{f}_{\cZ^{tr}}(\bX_i,\hat{c}_p)- p\right)\tau_i, \left(\hat{f}_{\cZ^{tr}}(\bX_j,\hat{c}_p)- p\right)\tau_j\r\} \\
	& = & {\rm Cov}\l( \hat{f}_{\cZ^{tr}}(\bX_i,\hat{c}_p)\tau_i, \hat{f}_{\cZ^{tr}}(\bX_j,\hat{c}_p)\tau_j\r) - 2p{\rm Cov}\l( \tau_i, \hat{f}_{\cZ^{tr}}(\bX_j,\hat{c}_p)\tau_j\r)\\
	& = & \frac{m \lfloor mp \rfloor (\lfloor mp \rfloor -1)-{\lfloor mp\rfloor}^2(m-1)}{m^2(m-1)}\E(\tau_i\mid \hat{f}_{\cZ^{tr}}(\bX_i,\hat{c}_p)=1)^2 - 2p{\rm Cov}\l( \tau_i, \hat{f}_{\cZ^{tr}}(\bX_j,\hat{c}_p)\tau_j\r)\\
	& = & \frac{\lfloor mp \rfloor(\lfloor mp \rfloor-m)}{m^2(m-1)}\kappa_{F1}(p)^2 + \frac{2p\lfloor mp \rfloor(m-\lfloor mp \rfloor)}{m^2(m-1)}\l(\kappa_{F1}(p)^2-\kappa_{F1}(p)\kappa_{F0}(p)\r)\\
	& = & (2p-1)\frac{\lfloor mp \rfloor(m-\lfloor mp \rfloor)}{m^2(m-1)}\kappa_{F1}(p)^2-\frac{2p\lfloor mp \rfloor(m-\lfloor mp \rfloor)}{m^2(m-1)}\kappa_{F1}( p)\kappa_{F0}(p)\\
	& = & \frac{\lfloor mp \rfloor(m-\lfloor mp \rfloor)}{m^2(m-1)}\l\{(2p-1)\kappa_{F1}( p)^2-2p\kappa_1( p)\kappa_{F0}(p)\r\}
\end{eqnarray*}
Combining this covariance result with the expression for the marginal
variances yields the desired variance expression for
$\hat{\tau}_{Fp}$.

Next, we derive the upper bound of bias.  Using the same technique as
the proof of Theorem~\ref{thm:PAPEest}, we can rewrite the expectation
of the proposed estimator as,
\begin{equation*}
\E(\hat{\tau}_{Fp}) \ = \ \E\l[\frac{1}{n} \sum_{i=1}^n \l\{Y_i\l(\hat{f}_{\cZ^{tr}}(\bX_i,\hat{c}_p)\r) - pY_i(1)-(1-p)Y_i(0)\r\}
\r]
\end{equation*}
Now, define $F(c)=\mc{P}(s(\bX_i)\leq c)$. Without loss of generality,
assume $\hat{c}_p > c_p$ (If this is not the case, we simply switch
the upper and lower limits of the integrals below).  Then, the bias of
the estimator is given by,
\begin{eqnarray*}
	\l|\E(\hat{\tau}_{Fp})-\tau_{Fp}\r|&= &
                                                \l|\E\l\{\E(\hat{\tau}_{Fp}-\tau_{Fp}  \mid \cZ^{tr})\r\}\r|\\
                                           & \leq
                                              &\E\l\{\l|\E(\hat{\tau}_{Fp}-\tau_{Fp}\mid
                                                \cZ^{tr})\r|\r\}\\
	&= &\E\l\{\l|\E\l[\frac{1}{n} \sum_{i=1}^n
             \l\{Y_i\l(\hat{f}_{\cZ^{tr}}(\bX_i,\hat{c}_p)\r)-Y_i\l(\hat{f}_{\cZ^{tr}}(\bX_i,c_p)\r)\r\}
             \ \biggl | \ \cZ^{tr}\r]\r|\r\} \\
	& = & \E\l\{\l|\E\l[\int^{\hat{c}_p}_{c_p}
	\E(\tau_i \mid \hat{s}_{\cZ^{tr}}(\bX_i)=c, \cZ^{tr}) \d F(c)\r]\r|\r\}\\
	& = & \E\l\{\l|\E\l[\int^{F(\hat{c}_p)}_{F(c_p)}
	\E(\tau_i\mid \hat{s}_{\cZ^{tr}}(\bX_i)=F^{-1}(x), \cZ^{tr}) \d x\r]\r|\r\}\\
	& \leq & \E\l[\l|F(\hat{c}_p)-{1-p}\r| \times \max_{c \in
		[c_p,\hat{c}_p]} \l | \E(\tau_i\mid
                 \hat{s}_{\cZ^{tr}}(\bX_i)=c, \cZ^{tr})\r|\r].
\end{eqnarray*}
By the definition of $\hat{c}_p$, $F(\hat{c}_p)$ is the
$m-\lfloor mp\rfloor$th order statistic of $n$ independent uniform
random variables. This statistic does not depend on the training set
$\cZ^{tr}$ as the test samples are independent.  Thus, $F(\hat{c}_p)$
follows the Beta distribution with the shape and scale parameters
equal to $m-\lfloor mp \rfloor$ and $\lfloor mp \rfloor+1$,
respectively.  Therefore, we have,
\begin{equation}
\mc{P}(|F(\hat{c}_p)-p|>\epsilon) \ = \ 1-B(1-p+\epsilon, m-\lfloor mp
\rfloor, \lfloor mp \rfloor+1)+B(1-p-\epsilon, m-\lfloor mp \rfloor,
\lfloor mp \rfloor+1),
\end{equation}
where $B(\epsilon, \alpha, \beta)$ is the incomplete beta function,
i.e.,
\begin{equation*}
B(\epsilon,\alpha,\beta) \ = \ \int^\epsilon_0 t^{\alpha -1} (1-t)^{\beta -1} \d t.
\end{equation*}
Combining with the result above, the desired result follows. \qed

\subsection{The Population Average Prescriptive Difference of
  Estimated ITRs under a Budget Constraint}
\label{app:PAPDpcv}

We consider the estimation and inference for the PAPD of an estimated
ITR.  The cross-validation estimator of this quantity is given by,
\begin{equation}
  \widehat{\Delta}_p(F,G) \ = \ \frac{1}{K}\sum_{k=1}^K
  \widehat\Delta_p(\hat{f}_{-k},\hat{g}_{-k},\cZ_k),
  \label{eq:PAPDpcvest}
\end{equation}
where $\widehat\Delta_p(\hat{f}_{-k},\hat{g}_{-k},\cZ_k)$ is defined
in equation~\eqref{eq:PAPDpest}.  Although the bias of the proposed
estimator is not zero, we derive its upper bound as done in
Theorem~\ref{thm:PAPEpcvest}.
\begin{theorem}
  {\sc (Bias and Variance of the Cross-validation PAPD Estimator with
    a Budget Constraint)}
  \label{thm:PAPDpcvest}
  \spacingset{1} Under
  Assumptions~\ref{asm:SUTVA},~\ref{asm:randomsample},~and~\ref{asm:comrand},
  the bias of the cross-validation PAPD estimator with a budget
  constraint $p$ defined in equation~\eqref{eq:PAPDpcvest} can be
  bounded as follows,
  \begin{align*}
    &\E_{\cZ^{tr}}[\mc{P}_{c_p(\hat{f}_{\cZ^{tr}}),c_p(\hat{g}_{\cZ^{tr}})}(|\E\{\widehat\Delta_p(F,G) -\Delta_p(F,G) \mid c_p(\hat{f}_{\cZ^{tr}}),c_p(\hat{g}_{\cZ^{tr}})\}|\geq  \epsilon)]
    \\ &\leq \ 1-2B(1-p+\gamma_p(\epsilon),n-\lfloor np \rfloor, \lfloor np
    \rfloor+1)+2B(1-p-\gamma_p(\epsilon), n-\lfloor np \rfloor,
    \lfloor np \rfloor+1),
  \end{align*}
  where any given constant $\epsilon > 0$,
  $B(\epsilon, \alpha, \beta)$ is the incomplete beta function (if
  $\alpha = 0$ and $\beta > 0$, we set
  $B(\epsilon,\alpha, \beta):=H(\epsilon)$ for all $\epsilon$ where
  $H(\epsilon)$ is the Heaviside step function), and
  \begin{equation*}
    \gamma_p(\epsilon)  =
    \frac{\epsilon}{\E_{\cZ^{tr}}[\max_{c \in  [c_{p}(\hat{f}_{\cZ^{tr}})-\epsilon, c_{p}(\hat{f}_{\cZ^{tr}})+\epsilon],
        d\in  [c_{p}(\hat{g}_{\cZ^{tr}})-\epsilon, c_{g}(\hat{g}_{\cZ^{tr}})+\epsilon]}
      \{\E(\tau_i\mid  \hat{s}_{\hat{f}_{\cZ^{tr}}}(\bX_i)=c),  \E(\tau_i\mid  \hat{s}_{\hat{g}_{\cZ^{tr}}}(\bX_i)=d)\}]}.
  \end{equation*}
The variance of the estimator is,
  \begin{eqnarray*}
    \V(\widehat \Delta_p(F,G))
    & = &  \frac{\E(S_{\hat{f}\hat{g}1}^2)}{n_1} +  \frac{\E(S_{\hat{f}\hat{g}0}^2)}{n_0}
             +  \frac{\lfloor  np \rfloor(\lfloor np \rfloor-n)}{n^2(n-1)}(\kappa_{F1}(p)^2
             +\kappa_{G1}(p)^2)\\
    & &- 2\left(\Pr(\hat{f}_{\cZ^{tr}}(\bX_i,\hat{c}_p(\hat{f}_{\cZ^{tr}}))=\hat{g}_{\cZ^{tr}}(\bX_i,\hat{c}_p(\hat{f}_{\cZ^{tr}}))=1) -\frac{\lfloor np
    	\rfloor^2}{n^2}\right)\kappa_{F1}(p)\kappa_{G1}(p)\\
    & &      - \frac{K-1}{K}\E(S_{FG}^2),
  \end{eqnarray*}
  where
  $S_{\hat{f}\hat{g}t}^2 = \sum_{i=1}^n
  (\widetilde{Y}_{\hat{f}\hat{g}i}(t) -
  \overline{\widetilde{Y}_{\hat{f}\hat{g}}(t)})^2/(n-1)$,
  $S_{FG}^2 = \sum_{k=1}^K
  (\widehat\Delta_p(\hat{f}_{-k},\hat{g}_{-k},\cZ_k)) -
  \overline{\widehat\Delta_p(\hat{f}_{-k},\hat{g}_{-k},\cZ_k)})^2/(K-1)$,
  $\kappa_{Ft}(p)=\E(\tau_i\mid
  \hat{f}_{\cZ^{tr}}(\bX_i,\hat{c}_{p}(\hat{f}_{\cZ^{tr}}))=t)$, and
  $\kappa_{Gt}(p) = \E(\tau_i\mid
  \hat{g}_{\cZ^{tr}}(\bX_i,\hat{c}_p(g))=t)$ with
  $Y_{\hat{f}\hat{g}i}(t) =
  \left\{\hat{f}_{\cZ^{tr}}(\bX_i,\hat{c}_p(\hat{f}_{\cZ^{tr}})) -
    \hat{g}_{\cZ^{tr}}(\bX_i,\hat{c}_p(\hat{g}_{\cZ^{tr}}))\right\}Y_i(t)$,
  $\overline{Y(t)}=\sum_{i=1}^n Y_{\hat{f}\hat{g}i}(t)/n$, and
  $\overline{\widehat\Delta_p(\hat{f}_{-k}, \hat{g}_{-k},\cZ_k)} =
  \sum_{k=1}^K \widehat\Delta_p(\hat{f}_{-k},\hat{g}_{-k},\cZ_k)/K$
  for $t=0,1$.
\end{theorem}
Proof is similar to that of Theorem~\ref{thm:PAPEpcvest}, and hence is
omitted.  To estimate the variance, it is tempting to replace all
unknowns with their sample analogues.  However, the empirical analogue
for the joint probability
$\Pr(\hat{f}_{\cZ^{tr}}(\bX_i,\hat{c}_p(\hat{f}_{\cZ^{tr}}))=\hat{g}_{\cZ^{tr}}(\bX_i,\hat{c}_p(\hat{g}_{\cZ^{tr}}))=1)$
under general $\hat{f}_{\cZ^{tr}},\hat{g}_{\cZ^{tr}}$ is not a good
estimate because it is solely based on one realization.  Thus, we use
the following conservative bound,
\begin{align*}
- &\left(\Pr(\hat{f}_{\cZ^{tr}}(\bX_i,\hat{c}_p(\hat{f}_{\cZ^{tr}}))=\hat{g}_{\cZ^{tr}}(\bX_i,\hat{c}_p(\hat{f}_{\cZ^{tr}}))=1) -\frac{\lfloor np
	\rfloor^2}{n^2}\right)\kappa_{F1}(p)\kappa_{G1}(p)\\&\leq \frac{\lfloor np  \rfloor  \max\{\lfloor  np \rfloor,
	n-\lfloor np \rfloor\}}{n^2(n-1)}|\kappa_{F1}(p)\kappa_{G1}(p)|,
\end{align*}
where the inequality follows because the maximum is achieved when the
scoring rules of $\hat{f}_{\cZ^{tr}}$ and $\hat{g}_{\cZ^{tr}}$ are
perfectly negatively correlated. We use this upper bound in our simulation and
empirical studies. In Section~\ref{sec:synthetic}, we find that this
upper bound estimate of the variance produces only a small
conservative bias.

\subsection{An Additional Empirical Application}
\label{app:trans}

In this section, we describe an additional empirical application based
on the canvassing experiment \citep{broockman2016durably}.  This study
was also re-analyzed by \citet{kunz:etal:18}.  The original authors
find little heterogeneity in treatment effect.  Our analysis below
confirms this finding.

\subsubsection{The Experiment and Setup}

We analyze the transgender canvas study of
\citet{broockman2016durably}.  This is an experiment that randomly
assigned a door-to-door canvassing treatment to over 1,200 households
(with a total of over 1,800 members) in Florida to estimate the
treatment effect on support for a transgender rights law. The placebo
group received a conversation on recycling, while the treatment group
received a conversation about transgender issues. The support is
measured at various time points after the intervention (i.e., 3 days,
3 weeks, 6 weeks, 3 months) using an online survey. The treatment
effect heterogeneity is important in this scenario as canvassing is
both costly and time-consuming. An ITR may allow canvassers to contact
only those who are positively influenced by the message.

We follow the pre-experiment analysis plan by the original authors,
and select a total of 26 baseline covariates including political
inclination, gender, race, and opinions on social issues.  Our
treatment variable is whether or not the individual received the
conversation about transgender issues (as opposed to the recycling
message). Since the randomization was conducted on the household
level, we randomly select one individual from each household for our
analysis. We focus on the primary target (support for the transgender
law) at the 3 day time point after the intervention, which is measured
on a discrete scale with 7 possible values $\{-3,-2,-1,0,1,2,3\}$,
with positive values indicating support.

The resulting dataset consists of 409 observations.  We randomly
select approximately 70\% of the sample (i.e., 287 observations) as
the training data and the reminder of the sample (i.e., 122
observations) as the evaluation data. We center the outcome variable
using the mean in the training data to minimize variance, as discussed
in Section \ref{sec:evalmetrics}. We train three machine learning models designed to
measure heterogeneous treatment effects: Causal Forests, Bayesian
Causal Forests \citep{hahn2020bayesian}, and R-Learner
\citep{nie2017quasi}. All tuning was done through the 5-fold cross
validation procedure on the training set using the PAPE as the
evaluation metric. For Causal Forest, we set {\tt tune.parameters =
  TRUE}. For Bayesian Causal Forests (BCF), tuning was done using a
burn-in sample for MCMC sampling. For R-Learner, we utilized the lasso
loss function and the default cross-validation for the regularization
parameter.  We then create an ITR as $\mathbf{1}\{\hat\tau(\bx)>0\}$
where $\hat\tau(\bx)$ is the estimated conditional average treatment
effect obtained from each fitted model.  We will evaluate these ITRs
using the evaluation sample.

\subsubsection{Results}

\begin{table}[t!]
  \centering \spacingset{1}
  \resizebox{\linewidth}{!}{
    \begin{tabular}{l..c|..c|..c}
      \hline
      & \multicolumn{3}{c|}{\textbf{BCF}} & \multicolumn{3}{c|}{\textbf{Causal Forest}}
      & \multicolumn{3}{c}{\textbf{R-Learner}}\\
      & \multicolumn{1}{c}{est.}  & \multicolumn{1}{c}{s.e.} & \multicolumn{1}{c|}{treated}
      & \multicolumn{1}{c}{est.}  & \multicolumn{1}{c}{s.e.} & \multicolumn{1}{c|}{treated}
      & \multicolumn{1}{c}{est.}  & \multicolumn{1}{c}{s.e.} & \multicolumn{1}{c}{treated}\\\hline
      No budget constraint     &  -0.104 & 0.128  & 48.4\%  &  -0.349 & 0.137 & 47.5\% & 0 & 0 &  100\% \\
      20\% Budget constraint   & -0.02 & 0.121  & 20\%  & -0.120   & 0.107 & 20\% & 0 & 0.104 &  20\% \\\hline
    \end{tabular}
  }
  \caption{The Estimated Population Average Prescription Effect (PAPE)
    for Bayesian Causal Forest (BCF), Causal Forest, and R-Learner
    with and without a Budget Constraint.  For each of the three
    outcomes, the point estimate, the standard error, and the
    proportion treated are shown.  The budget constraint considered
    here implies that the maximum proportion treated is
    20\%.} \label{tb:comparison_add}
\end{table}

Table~\ref{tb:comparison_add} presents the results.  We find that
without a budget constraint, none of the ITRs based on the machine learning methods
significantly improves upon the random treatment rule. In particular, the R-Learner
leads to an ITR that treats everyone. Furthermore, we see that the ITR
based on Causal Forest performs worse than the random treatment rule by 0.349 (out of a
$-3$ to $3$ scale) with a standard error of 0.137. The results are
similar when we impose a budget constraint, and none of the resulting
ITRs perform statistically significantly better than the random treatment rule. The
result based on R-Learner is consistent with the conclusion of the
original study indicating that there was no heterogeneity detected
using LASSO.

\begin{figure}[t!]
  \spacingset{1}
  \includegraphics[width=\linewidth]{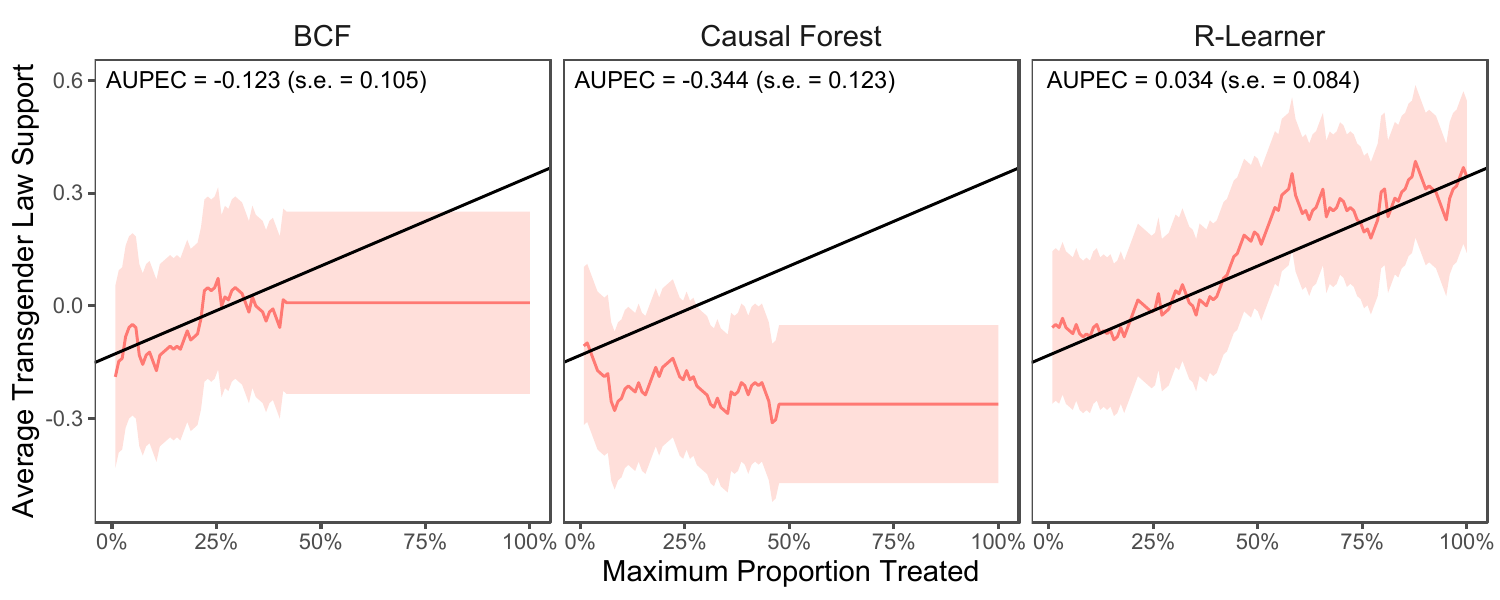}
  \caption{Estimated Area Under the Prescriptive Effect Curve (AUPEC).
    A solid red line in each plot represents the Population Average
    Prescriptive Effect (PAPE) with pointwise 95\% confidence
    intervals shaded.  The area between this line and the black line
    (representing random treatment) is the AUPEC.  The results are
    presented for the individualized treatment rules based on Bayesian
    Causal Forest (BCF), Causal Forest, and R-Learner.
  } \label{fg:AUPEC_add}
\end{figure}

We plot the estimated PAPE (with 95\% pointwise confidence interval)
as a function of budget constraint in Figure~\ref{fg:AUPEC_add}.  The
area between this line and the black horizontal line at zero
corresponds to the AUPEC.  In each plot, the horizontal axis
represents the budget constraint as the maximum proportion treated,
and the point estimate and standard error of the AUPEC are
shown. While BCF and R-Learner fail to create an ITR that is
significantly different from the random treatment rule, Causal Forest produces an ITR
that is statistically significantly worse than the random treatment rule. The result
illustrates a potential danger of using an advanced machine learning algorithm to
create an ITR.  Indeed, there is no guarantee that the resulting ITR
outperforms the random treatment rule.

\end{document}